\documentclass[12pt]{article}
\usepackage[left=1in,top=1in,right=1in,bottom=1in,head=.1in,nofoot]{geometry}
\usepackage{setspace,url,bm,amsmath} 
\usepackage[small]{titlesec}
\titlelabel{\thetitle.\quad}
\titleformat*{\section}{\bf\large\center}
\titlespacing*{\section}
{0pt}{0ex}{0ex}
\titlespacing*{\subsection}
{0pt}{0ex}{0ex}
\titlespacing*{\subsubsection}
{0pt}{0ex}{0ex}

\usepackage[linesnumbered,noline,noend]{algorithm2e}
\usepackage{graphicx} 
\usepackage{bbm}
\usepackage{latexsym}
\usepackage{caption}
\usepackage{subcaption}
\usepackage[margin=20pt]{subcaption}
\usepackage{hyperref}
\usepackage{booktabs}
\usepackage{multirow}
\usepackage{enumitem}
\usepackage[table]{xcolor}
\usepackage{textcomp}
\usepackage{soul}
\usepackage{amsthm}
\usepackage{amssymb}
\usepackage{amsmath}
\usepackage{color}
\usepackage{comment}
\usepackage{soul,xcolor}

\newcommand{\GG}[1]{}

\theoremstyle{definition}

\newtheoremstyle{myspacing}  
  {2pt}                     
  {2pt}                     
  {}                 
  {}                         
  {\bfseries}               
  {.}                        
  { }                        
  {}                         
\theoremstyle{myspacing}

\newtheorem*{theorem*}{Theorem}
\newtheorem{theorem}{Theorem}
\newtheorem*{rmk*}{Remark}

\newtheorem{lemma}{Lemma}

\newtheorem{remark}{Remark}
\newtheorem{corollary}{Corollary}
\newtheorem*{corollary*}{Corollary}

\usepackage{natbib} 
\bibpunct{(}{)}{;}{a}{}{,} 
 
\usepackage{etoolbox} 
\apptocmd{\sloppy}{\hbadness 10000\relax}{}{} 

\usepackage{color}
\usepackage{listings}
\usepackage{enumitem}

\def\rank{\text{rank}}

\def\iidsim{\stackrel{\text{i.i.d.}}{\sim}}

\def\I{\mathbbm{1}}

\def\hyper{\text{Hypergeometric}}

\def\bs{\boldsymbol}

\allowdisplaybreaks

\def\hyper{\mathfrak{H}}

\def\Unif{\text{Unif}}
\def\sp{\text{sp}}
\def\Pr{\mathbb{P}}

\def\sp{\textup{sp}}
\def\fp{\textup{fp}}

\def\treat{\textup{t}}
\def\control{\textup{c}}

\def\rev{\color{red}}

\def\rev{\color{black}}

\def\secadj{\vspace{-2ex}}

\RequirePackage[normalem]{ulem}

\let\oldnl\nl
\newcommand\nonl{
  \renewcommand{\nl}{\let\nl\oldnl}}

\usepackage{tikz}
\usetikzlibrary{shapes.geometric, arrows}
\tikzstyle{io} = [trapezium, trapezium left angle=70, trapezium right angle=110, minimum width=1cm, minimum height=1cm, text centered, draw=black,  trapezium stretches=true, thick]
\tikzstyle{process} = [rectangle, minimum width=3cm, minimum height=1cm, text centered, draw=black, thick]
\tikzstyle{decision} = [diamond, minimum width=1cm, minimum height=1cm, text centered, draw=black, thick]
\tikzstyle{arrow} = [thick,->,>=stealth]
\usepackage{subcaption}
\usepackage{upgreek}

\newtheorem{theoremA}{Theorem}

\usepackage[textsize=tiny, textwidth = 2cm, shadow]{todonotes}

\begin{document}
\setstcolor{purple}
\doublespacing

\title{\bf 
\Large
Enhanced inference for distributions and quantiles of individual treatment effects in various experiments
}
\author{
 \normalsize
 Zhe Chen and Xinran Li
 \footnote{
 Zhe Chen is Postdoctoral Researcher, Center for Clinical Trials Innovation, Department of Biostatistics, Epidemiology, Informatics, Perelman School of Medicine, University of Pennsylvania, Philadelphia, PA, 19104 (e-mail: \href{mailto:zhe.chen@pennmedicine.upenn.edu}{zhe.chen@pennmedicine.upenn.edu}). 
 Xinran Li is Assistant Professor, Department of Statistics, University of Chicago, Chicago, IL 60637 (e-mail: \href{mailto:xinranli@uchicago.edu}{xinranli@uchicago.edu}). 
 }
}
\date{}
\maketitle

\vspace{-0.5in}


\singlespacing
\begin{abstract}
Understanding treatment effect heterogeneity has become increasingly important in many fields. In this paper we study distributions and quantiles of individual treatment effects to provide a more comprehensive and robust understanding of treatment effects beyond usual averages, although they are more challenging to infer due to nonidentifiability from the observed data. 
Recent randomization-based approaches offer finite-sample valid inference for treatment effect distributions and quantiles in both completely randomized and stratified randomized experiments, but they can be overly conservative by assuming the worst-case scenario where units with large effects are all assigned to the treated (or control) group. We introduce two improved methods to enhance the power of these existing approaches. 
The first method reinterprets existing approaches as inferring treatment effects among only treated or control units, and then combines the inference for treated and control units to infer treatment effects for all units. 
The second method explicitly controls for the actual number of treated units with large effects. Both simulation and application demonstrate the substantial gain from the improved methods. 
These methods are further extended to sampling-based experiments as well as  quasi-experiments from matching, in which the ideas for both improved methods can play critical and complementary roles.
\end{abstract}



\doublespacing

\secadj
\section{Introduction}

\subsection{Treatment effect heterogeneity}

Understanding heterogeneous treatment effects has become increasingly important in many fields. 
Most of the existing work either tries to investigate how the subgroup or conditional average treatment effects vary with observed pretreatment covariates \citep{Hill2011, Imai2013, Tian2014, Imbens2016, Wager2018, Yu2019, Nie2020, Kennedy2023}, 
or seeks a treatment rule, a function of observed covariates that determines the treatment allocation, to maximize the overall outcomes for all units \citep{Manski2004, Murphy2011, zhang2012estimating,
Zhao2012, Luedtke2016, Stefan2021}.
In this paper, we focus on an alternative perspective aimed at understanding the distributions or quantiles of individual treatment effects, which can accommodate treatment effect heterogeneity due to unobserved covariates. 
This can help address questions such as: what proportion of units
benefit from the treatment, and what is the median value of
the individual treatment effects?
As a side note, 
the aforementioned strategies can be combined;
for example, we can investigate the individual treatment effect distribution among subpopulations defined by pretreatment covariates \citep{Zhang2013}. 
Nevertheless, throughout this paper, we will focus on inferring the marginal distributions and quantiles of individual treatment effects. 

Understanding treatment effect distributions is often empirically important for program, policy and intervention evaluation.
First, in the seminal work of \citet{Heckman97}, they emphasized the importance of understanding the distribution of impacts, which is crucial for answering questions about the political economy of programmes, the distribution of programme benefits and the option values conferred on programme participants. 
Recently, \citet{lu2018treatment} and \citet{Rosenblum2019} also emphasized the importance of understanding the proportion of units that benefit from the treatment in educational, behavior and medical studies. 
Second, the treatment effect distributions or quantiles are more robust measures of treatment effects compared to the usual averages, especially in the presence of extreme individual effects \citep{CDLM21quantile, SL22quantile}. 
For example, a treatment that is beneficial on average may actually be harmful for most units in the population, since the positive average effect could be driven by some extreme or even outlier effects on a tiny proportion of the population. 
Third, the asymptotic inference approaches designed for average treatment effects may have poor finite-sample performance, when the individual treatment effects and outcomes are heavy-tailed.  
Fourth, the distribution of individual effects can better accommodate ordinal outcomes. 
As \citet{lu2018treatment} commented,  it is often difficult to define or compare ``average" ordinal outcomes in a substantive way; for example, the averages of education levels, such as ``high school'' and ``PhD'', are difficult to make sense of and to compare. 

\subsection{Superpopulation inference for distributions of individual treatment effects}

Since \cite{Heckman97}, the distributions and quantiles of individual treatment effects have received substantive attention. 
Compared to the usual average or conditional average treatment effects, 
inference for the treatment effect distributions poses additional challenges due to their nonidentifiability issues. 
Specifically, from the observed data, 
the joint distributions between potential outcomes are generally not identifiable, since for each unit we can observe at most one potential outcome. 
Consequently, due to its dependence on the joint distribution of potential outcomes, the distribution of individual effects is generally not identifiable from the observed data, and thus we do not expect any consistent estimation for it without additional assumptions. 
Nevertheless, the individual effect distribution is still partially identifiable. 
\citet{Fan2010} studied sharp bounds for the distribution and quantile functions of individual effects, based on the identifiable marginal distributions of the potential outcomes, and then proposed large-sample asymptotic inference for the distribution function of individual treatment effects at any given point; 
see also \citet{FIRPO2019210, Frandsen21, ruiz2022non, Lee2024, kaji2023assessing} for various extensions. 
In a related context, \citet{LeiCandes2021} used conformal inference to construct prediction intervals for individual effects. 
However, most 
existing approaches rely on random sampling of units from some superpopulation, which may not be realistic in some applications.
For example, when analyzing causal effects of a policy on all 50 states in the United States, it could be difficult to articulate what the superpopulation is that generates these 50 samples of states; see \citet{Abadie2020} for more related discussions.

\subsection{Finite population inference for treatment effects}
The finite population inference, also called the design-based or randomization-based inference, 
was first proposed in the seminal work of \citet{Fisher:1935} and \citet{Neyman:1923}, 
and has received considerable attention in recent years;
see, e.g., \citet{Rosenbaum02a, Freedman2018aoas, lin2013, imbens2015causal, fcltxlpd2016, LDR18, Bojinov2019, Abadie2020, GuoBasse21, Basse2024}. 
However, most of the existing randomization-based inference approaches either require constant-effect-type assumptions, which are not able to accommodate unknown individual effect heterogeneity, or focus on average treatment effects, 
which often invoke large-sample approximations and may not be reliable when the sample size is small or data are heavy-tailed. 

Recently, \citet{CDLM21quantile} proposed a novel approach for conducting inference on quantiles of individual treatment effects, as well as the proportions of units with individual effects exceeding (or falling below) any thresholds. 
Importantly, their approach is randomization-based and finite-sample exact. 
\citet{SL22quantile} further extended it to stratified randomized experiments, and also proposed sensitivity analysis for quantiles of individual effects in matched observational studies in the presence of unmeasured confounding.
The approaches in \citet{CDLM21quantile} and \citet{SL22quantile} build upon the Fisher randomization test \citep{Fisher:1935}, 
and consider the worst-case randomization test over all possible potential outcome configurations satisfying a given null hypothesis on quantiles of individual effects; see Sections \ref{sec:review_drawback} and \ref{sec:review_SL} for brief reviews of these approaches. 
Importantly, the worst-case configurations assume that all units with large individual effects are in the treated (or control) group. This, however, is unlikely given that units are randomized into treatment and control, and makes their approaches overly conservative. 
For example, when half of the units are randomized to treatment and control, their approaches cannot obtain informative inference for any quantiles of individual effects below the median, with the corresponding confidence intervals being the noninformative whole real line. 

\subsection{Our contributions}

In this paper, we propose two methods to improve \citet{CDLM21quantile} for inferring individual treatment effect distributions in completely randomized experiments. 
Our first method provides a new interpretation of the existing approach, by viewing it as constructing \textit{prediction} intervals for the \textit{random} quantiles of individual treatment effects for only treated (or control) units. 
We then combine the prediction intervals for treated and control units to construct confidence intervals for quantiles of individual effects among all units. 
Our second method considers the number of treated units whose individual effects exceed the quantile of interest, and uses \citeauthor{Berger1994}'s \citeyearpar{Berger1994} idea to control for it to derive more powerful inference. 
Both applications and simulations demonstrate substantial gains from the improved methods. The trade-offs between the two improved methods are also investigated. 

We further extend the proposed methods to stratified randomized experiments, thereby improving \citet{SL22quantile}, and sampling-based randomized experiments \citep{BD20, YQL2021}. The former is one of the most popular designs in practice and the latter is important when we want to generalize experimental evidence to larger populations of interest. In addition, we also extend the proposed improvement to quasi-experiments constructed through matching, where we conduct sensitivity analysis for causal conclusions regarding distributions and quantiles of individual effects in the presence of unmeasured confounding. Moreover, both improved methods for the completely randomized experiments and their ideas play critical, unique and complementary roles in these extensions. In particular, different ideas can be preferred in different contexts, and they can both be important components of a single inferential procedure, such as in sampling-based experiments with stratified randomization of treatment assignments. 

The proposed approaches are all finite-sample valid, without requiring any large-sample approximations or assumptions on the tail behavior of the potential outcomes, except for the matched observational studies where we impose mild regularity conditions. 

The paper proceeds as follows. 
Section \ref{sec:framework} introduces the framework and notation. 
Section \ref{sec:cre} discusses the improved inference for treatment effect quantiles in completely randomized experiments, 
and Section \ref{sec:str} extends to stratified randomized experiments as well as quasi-experiments constructed through matching. 
Section \ref{sec:samp_based} studies sampling-based randomized experiments in which we want to generalize experimental evidence to larger populations of interest.
Section \ref{sec:illustration} conducts numerical illustration, 
and gives 
suggestions for practical implementation of the proposed methods. 
Section \ref{sec:conclusion} concludes with a short discussion. 

\secadj
\section{Framework and Notation}\label{sec:framework}

\subsection{Potential outcome and treatment assignment}\label{sec:po_assign}
We consider an experiment on $n$ units, with $n_{\treat}$ units assigned to treatment and the remaining $n_{\control}=n - n_{\treat}$ assigned to control. 
We invoke the potential outcome framework \citep{Neyman:1923, Rubin:1974} to define treatment effects. 
For each unit $1\le i \le n$, let $Y_i(1)$ and $Y_i(0)$ denote the treatment and control potential outcomes, and $\tau_i \equiv Y_i(1) - Y_i(0)$ denote the individual treatment effect. 
We use $Z_i$ to denote the treatment assignment indicator for unit $i$, where $Z_i = 1$ if the unit receives treatment and $0$ otherwise. 
The observed outcome for unit $i$ is then 
$Y_i = Z_i Y_i(1) + (1-Z_i) Y_i(0)$. 
Let $\bs{Y}(1) = (Y_1(1), \ldots, Y_n(1))^\top$, 
$\bs{Y}(0) = (Y_1(0), \ldots, Y_n(0))^\top$, 
$\bs{\tau} = (\tau_1, \ldots, \tau_n)^\top$,
$\bs{Y} = (Y_1, \ldots, Y_n)^\top$
and 
$\bs{Z} = (Z_1, \ldots, Z_n)^\top$
denote the vectors of potential outcomes, individual treatment effects, observed outcomes and treatment assignments. 

\subsection{Distributions of individual treatment effects}\label{sec:null}

Let $F_n(c) \equiv n^{-1} \sum_{i=1}^n \I(\tau_i \le c)$ denote the sample distribution function of the individual treatment effects, 
and $F_n^{-1}(\beta) = \inf\{c: F_n(c) \ge \beta\}$ for $\beta \in [0,1]$ denote the corresponding quantile function.  
We 
{\rev consider}
null hypotheses of the following form: 
\setlength{\abovedisplayskip}{2pt}
\setlength{\belowdisplayskip}{2pt}
\begin{align}\label{eq:null_fp}
H_{\beta, c}^n: F_n^{-1}(\beta) \le c
\Longleftrightarrow F_n(c) \ge \beta \Longleftrightarrow 1 - F_n(c) \le 1 - \beta, \quad (0 \le \beta \le 1; c \in \mathbb{R})
\end{align}
under which the $\beta$th quantile of individual effects is bounded by $c$ or equivalently the proportion of units with individual effects exceeding $c$ is at most $1-\beta$. 
The equivalence in \eqref{eq:null_fp} follows by definition. 
{\rev 
As discussed shortly, valid tests for \eqref{eq:null_fp} can be readily adapted to accommodate the other one-sided and two-sided hypotheses. We thus mainly focus on \eqref{eq:null_fp}.}

Moreover, 
with the finite sample of $n$ units, the quantile function only takes values in the sorted individual effects: $\tau_{(1)}\le \ldots \le \tau_{(n)}$. 
Specifically, for any $\beta \in [0,1]$, $F_n^{-1}(\beta) = \tau_{(k)}$ with $k=\lceil n \beta \rceil$, where we define  $\tau_{(0)} = -\infty$ for descriptive convenience. 
Thus, it suffices to consider hypotheses of the following form: 
\begin{align}\label{eq:Hnkc}
    H_{k, c}^n: \tau_{(k)} \le c 
    \ \Longleftrightarrow \ 
    n(c) \le n-k, 
    \qquad (0\le k \le n; c\in \mathbb{R})
\end{align}
where $n(c) = \sum_{i=1}^n  \I(\tau_i > c)$
denotes the number of units with individual effects greater than $c$.  
Below we give two additional remarks regarding the null hypotheses in \eqref{eq:null_fp} and \eqref{eq:Hnkc}.

First, we can use Lehmann-style test inversion to construct (simultaneously) valid one-sided lower confidence limits for (multiple) quantiles of individual effects, which can lead to confidence bands for the whole quantile (or distribution) function of individual treatment effects. 
In addition, we will consider cases where the experimental units constitute only a subset of the target population of interest, under which we will infer treatment effect distributions and quantiles for the whole target population.

{\rev 
Second, 
we can consider null hypotheses of form $F_n^{-1}(\beta) \ge c$  or $\tau_{(k)}\ge c$ with alternatives favoring smaller treatment effects.  
They can be tested analogously as \eqref{eq:null_fp} and \eqref{eq:Hnkc} by switching the treatment labels or changing the outcome signs. 
By test inversion, these can then lead to (simultaneous) upper confidence limits for (multiple) quantiles of individual effects. 
Importantly, by combining the two one-sided tests or one-sided confidence intervals with a Bonferroni correction, 
we can conduct two-sided tests or construct two-sided confidence intervals for quantiles of individual treatment effects; 
see Section \ref{sec:teacher} for an 
illustration.}



\secadj
\section{Completely randomized experiments}\label{sec:cre}

We will first consider a completely randomized experiment (CRE) with $n$ units, among which $n_{\treat}$ units are randomized to treatment and the remaining $n_{\control} = n - n_{\treat}$ are randomized to control. 
We want to infer the sample distribution function $F_n(\cdot)$ and sample quantile function $F_n^{-1}(\cdot)$ of individual treatment effects for the $n$ experimental units. 
Here we view the potential outcomes of the $n$ units as fixed constants (or equivalently condition on them), 
under which the randomness in the observed data comes solely from the random treatment assignment. 
That is, we will conduct randomization-based inference, sometimes also called the design-based or finite population inference, for the treatment effects of these experimental $n$ units. 
Our inference adapts and improves upon the approach recently proposed by \citet{CDLM21quantile}, which is reviewed in Section \ref{sec:review_drawback}, highlighting its limitations.
We then discuss two extensions that can enhance the power of their approach in Sections \ref{sec:effect_treated} and \ref{sec:berger}, respectively.


\subsection{Review of \citet{CDLM21quantile}'s  approach and its limitations}
\label{sec:review_drawback}

We briefly review the approach in \citet{CDLM21quantile}  for inferring quantiles of individual treatment effects, which builds upon the Fisher randomization test that was originally proposed for testing sharp null hypotheses that require speculation of all individual effects
\citep{Fisher:1935}. 
Specifically, 
\citet{CDLM21quantile} proposed to test the composite null hypothesis $H_{k,c}^n$ in \eqref{eq:Hnkc} using the supremum of the $p$-value from the Fisher randomization test over all possible sharp null hypotheses that satisfy $H_{k,c}^n$. 
To make it computationally feasible, they utilized a distribution-free rank score statistic
of the following form:
\begin{align}\label{eq:rank_score}
    t(\bs{z}, \bs{y}) = \sum_{i=1}^n z_i \phi( \rank_i(\bs{y}) ), 
    \quad (\bs{z}\in \{0,1\}^n; \bs{y}\in \mathbb{R}^n)
\end{align}
where $\phi$ is a monotone increasing\footnote{Throughout the paper, increasing means nondecreasing and analogously decreasing means nonincreasing.} rank transformation and 
$\rank_i(\bs{y})$ denotes the rank of the $i$th coordinate of $\bs{y}$. 
Importantly, we will break ties using index ordering, assuming that units' ordering has been randomly shuffled before analysis. 
The rank score statistic is distribution free under the CRE, in the sense that the distribution of $t(\bs{Z}, \bs{y})$ is the same for any constant $\bs{y}\in \mathbb{R}^n$. 
Such a property is crucial since it greatly simplifies the computation for the supremum of the Fisher randomization $p$-value. 
The resulting valid $p$-value proposed by \citet{CDLM21quantile} for testing $H_{k,c}^n$ in \eqref{eq:Hnkc} has the following form: 
\begin{align}\label{eq:p_nkc}
    p_{k,c}^n & = G
    \Big(
    \inf_{\bs{\delta}\in \mathcal{H}^n_{k,c}} 
    t(\bs{Z}, \bs{Y} - \bs{Z} \circ \bs{\delta})
    \Big),
    \quad 
    \text{where } G(x) = \Pr\{ t(\bs{Z}, \bs{y}) \ge x \} \text{ for any } \bs{y}\in \mathbb{R}^n, 
\end{align}
and $\mathcal{H}^n_{k,c} \subset \mathbb{R}^n$ denotes the set of all possible values of the individual effect vector $\bs{\tau}$ under $H_{k, c}^n$.
{\rev Importantly, $G(\cdot)$ in \eqref{eq:p_nkc} does not vary with $\bs{y}$ due to the distribution free property of the rank-based statistic in \eqref{eq:rank_score},  and it is essentially a known function depending only on numbers of treated and control units $(n_\treat, n_\control)$ and the rank transformation $\phi(\cdot)$ in \eqref{eq:rank_score}.}
\citet{CDLM21quantile} showed that the minimum test statistic value in \eqref{eq:p_nkc} is achieved when the $n-k$ treated units (or all treated units if $n-k>n_{\treat}$) with the largest observed outcomes have individual effects $\infty$ while the remaining have individual effects $c$. 
This is intuitive since 
(i) to minimize the test statistic, we want the imputed control potential outcomes $\bs{Y}-\bs{Z}\circ \bs{\delta}$ for treated units to be small, or equivalently the hypothesized individual effects for treated units to be large; 
(ii) under the null $H_{k, c}^n$, at most $n-k$ units can have individual effects greater than $c$, and the remaining units' individual effects are all bounded by $c$; 
(iii) in the worst case, we will assume that $n-k$ units have infinite individual effects and all receive treatment, with the corresponding observed outcomes being the largest among treated units. 

However, \citet{CDLM21quantile}'s $p$-value in \eqref{eq:p_nkc} can be overly conservative, since it presumes that all units with large effects receive treatment. 
Under the CRE, the set of treated units is a simple random sample of size $n_{\treat}$ from all units, 
and thus, in expectation, at most $(n-k)\cdot n_{\treat}/n$ treated units have effects greater than $c$ under the null $H_{k,c}^n$. 
Thus, considering the worst-case configuration of individual effects 
as in \citet{CDLM21quantile}  may provide overly conservative inference.
{\rev
Taking the education experiment in Section \ref{sec:teacher} as an example, \citet{CDLM21quantile} infer that, with $90\%$ confidence level, there are at least $88$ units (or equivalently $88/233=37.8\%$ of units) with positive individual effects. 
Since the approach in \citet{CDLM21quantile} presumes that all units with large effects are in the treated group, 
this essentially implies that at least $88$ treated units (or equivalently $88/164 = 53.7\%$ of treated units) have positive individual effects. 
Given that the treated and control groups are, on average, comparable in all aspects under the CRE, we can expect that, with $90\%$ confidence, about half of the units would benefit from the treatment, which is more informative and powerful than the original results in \citet{CDLM21quantile}. 
}


In the following two subsections, 
we will utilize two distinct ideas to improve the power of \citet{CDLM21quantile}'s approach, {\rev with a comparison between them discussed at the end of this section.} 
We will also consider simultaneous confidence intervals for multiple treatment effect quantiles, which can further lead to confidence bands for 
$F_n(\cdot)$ and $F_n^{-1}(\cdot)$. 

\begin{remark}
As discussed in \citet{CDLM21quantile}, we may consider a more general form of test statistics beyond rank-based ones and use large-sample approximation for the imputed null distribution to facilitate the optimization over the randomization $p$-value. 
We focus on rank-based statistics in the paper due to the following reasons. 
First, the optimization can often be solved efficiently, and it has a closed-form solution in the case of CRE. 
Second, the resulting test can be finite-sample valid and can accommodate heavy-tailed outcomes and individual effects, for which a statistic based on original outcomes may be sensitive and the corresponding large-sample approximation may perform poorly. 
Third, the classical literature \citep{Lehmann1963, Kraft1972, Beran1974} suggests that, with properly chosen rank transformations, 
rank-based tests can achieve efficiency comparable to tests based on original outcomes\footnote{\rev 
Note that the test in \eqref{eq:p_nkc} for the maximum individual effect $\tau_{(n)}$ shares the same procedure as the usual permutation test for the location shift between treated and control observations. 
It is worth investigating the power of the test, in particular with respect to the choice of rank transformation, in our context for inferring treatment effect quantiles. This is beyond the scope of this paper, and we leave it for future study.}.
Finally, as discussed shortly, the strategies discussed below for power improvement can still be useful for other choices of test statistics. 
\end{remark}

\subsection{Treatment effects for treated, control and all units}\label{sec:effect_treated}

By taking a deeper look at the $p$-value in \eqref{eq:p_nkc}, 
we notice that the test statistic $t(\bs{Z}, \bs{Y} - \bs{Z} \circ \bs{\delta})$ depends essentially on the hypothesized treatment effects for treated units. 
In particular, the infimum of the test statistic in \eqref{eq:p_nkc} under $H^n_{k,c}$ is the same as that under the hypothesis where at most $n-k$ treated units have individual effects greater than $c$. 
This then motivates us to consider the following hypotheses on treatment effects among only treated units: 
\begin{align}\label{eq:Hn1kc}
    H_{k,c}^{n,\treat}: \tau_{\treat(k)} \le c 
    \Longleftrightarrow 
    n_{\treat}(c) \le n_{\treat}-k
    \Longleftrightarrow 
    \bs{\tau} \in \mathcal{H}^{n,\treat}_{k,c}, 
    \qquad (0\le k \le n_{\treat}; c\in \mathbb{R})
\end{align}
where $\tau_{\treat(1)} \le \ldots \le \tau_{\treat(n_{\treat})}$ denote the sorted individual treatment effects for treated units, 
$n_\treat(c) = \sum_{i=1}^n Z_i \I(\tau_i > c)$
denotes the number of treated units with individual effects greater than $c$, 
and 
$\mathcal{H}^{n,\treat}_{k,c}$ denotes the set of possible values of $\bs{\tau}$ under the hypothesis $H_{k,c}^{n,\treat}$. 
We emphasize that, unlike traditional null hypotheses, the hypothesis $H^{n, \treat}_{k,c}$ in \eqref{eq:Hn1kc} is generally random. 
This is because hypotheses of form \eqref{eq:Hn1kc} concern treatment effects for treated units, which are randomly selected from all experimental units.  
With a slight abuse of terminology, 
we say $p$ is a valid $p$-value for testing a hypothesis $H$ that can 
be random, 
if $\Pr(p\le \alpha \textup{ and $H$ holds})\le \alpha$
for any $\alpha\in (0,1)$; see also \citet{jin2023selection}. 

\begin{theorem}\label{thm:treated}
    Consider the CRE and any rank score statistic $t(\cdot, \cdot)$ in \eqref{eq:rank_score}.
    \begin{enumerate}[label=(\roman*), topsep=0ex,itemsep=-1ex,partopsep=1ex,parsep=1ex]
        \item[(a)] For any $0\le k \le n_{\treat}$ and any $c\in \mathbb{R}$, 
        \begin{align}\label{eq:p1}
            p_{k,c}^{n, \treat} 
            & \equiv 
            G
            \big(
            \inf_{\bs{\delta}\in \mathcal{H}^{n, \treat}_{k,c}} 
            t(\bs{Z}, \bs{Y} - \bs{Z} \circ \bs{\delta})
            \big)
            = 
            G
            \big(
            \inf_{\bs{\delta}\in \mathcal{H}^{n}_{n_{\control}+k, c}} 
            t(\bs{Z}, \bs{Y} - \bs{Z} \circ \bs{\delta})
            \big)
            =
            p_{n_{\control}+k, c}^n
        \end{align}
        is a valid $p$-value for testing $H_{k,c}^{n, \treat}$ in \eqref{eq:Hn1kc}, 
        where $p_{n_{\control}+k, c}^n$ is defined as in \eqref{eq:p_nkc}. 
        \item[(b)] For any $0 \le k\le n_{\treat}$ and $\alpha\in (0,1)$, 
        $\mathcal{I}_{\treat(k)}^{\alpha} \equiv \{c\in \mathbb{R}: p_{k,c}^{n, \treat} > \alpha \}$ is a $1-\alpha$ prediction set for $\tau_{\treat(k)}$, in the sense that 
        $\Pr\{ \tau_{\treat(k)} \in \mathcal{I}_{\treat(k)}^{\alpha} \}\ge 1-\alpha$, and $\mathcal{I}_{\treat(k)}^\alpha$ must be an interval of form $(c, \infty)$ or $[c, \infty)$, 
        with $c = \inf\{c: p_{k,c}^{n, \treat} > \alpha, c \in \mathbb{R}\}$.
        \item[(c)] 
        The prediction intervals in (b) 
        are simultaneously valid, 
        in the sense that 
        $
            \Pr( 
            \cap_{k=1}^{n_{\treat}} 
            \{ 
            \tau_{\treat(k)} \in \mathcal{I}_{\treat(k)}^{\alpha}
            \}
            )
            \ge 1-\alpha. 
        $
    \end{enumerate}
\end{theorem}



Theorem \ref{thm:treated} has several implications. 
First, 
due to the equivalence in \eqref{eq:p1}, 
the prediction interval $\mathcal{I}_{\treat(k)}^{\alpha}$ in (b) is the same as the confidence interval for $\tau_{(n_{\control}+k)}$ constructed in \citet{CDLM21quantile}, whose validity follows immediately from Theorem \ref{thm:treated} since $\tau_{(n_{\control}+k)} \ge \tau_{\treat(k)}$. 
However, using $\mathcal{I}_{\treat(k)}^{\alpha}$ as a confidence interval for $\tau_{(n_{\control}+k)}$ may be overly conservative, since 
$\tau_{(n_{\control}+k)} = \tau_{\treat(k)}$ only when all the control units' effects are below $\tau_{\treat(k)}$. 
Analogously, the prediction interval in (c) is also the same as the confidence interval for $n(c)$ constructed in \citet{CDLM21quantile}, whose validity follows immediately from Theorem \ref{thm:treated} since $n(c) \ge n_\treat(c)$. 
Again, such a confidence interval for $n(c)$ may be overly conservative, since $n(c) = n_\treat(c)$ only when no control units have effects greater than $c$. 
From the above, \citet{CDLM21quantile} essentially used effect quantiles among treated units, {\rev supplemented with noninformative $(-\infty, \infty)$ intervals for effect quantiles among control units}, 
to infer those for all experimental units, which can be overly conservative. 

{\rev Second, due to the equivalence relation in \eqref{eq:Hnkc} and \eqref{eq:Hn1kc}, the simultaneous prediction (or confidence) intervals for effect quantiles can also provide simultaneous prediction (or confidence) intervals for numbers or proportions of units with effects exceeding any thresholds.
Similar results hold for the later theorems as well. 
}

Third, by the same logic as Theorem \ref{thm:treated}, we can infer treatment effects on control units. 
This can be achieved by applying Theorem \ref{thm:treated} with switched treatment labels and changed outcome signs. 
Furthermore, we can combine them to infer treatment effects for all units.

\begin{theorem}\label{thm:comb_all}
    For any $\alpha\in (0,1)$, 
    let $\mathcal{I}^\alpha_{\treat(k)}$, $1\le k \le n_{\treat}$, be $1-\alpha$ simultaneous prediction intervals for sorted individual effects of treated units, 
    and $\mathcal{I}^\alpha_{\control(k)}$, $1\le k \le n_{\control}$, be $1-\alpha$ simultaneous prediction intervals for sorted individual effects of control units; 
    these are all one-sided intervals of form $(c, \infty)$ or $[c, \infty)$. 
    Pool together $\mathcal{I}^\alpha_{\treat(k)}$s and $\mathcal{I}^\alpha_{\control(k)}$s, and arrange them 
    such that 
    $
    \mathcal{I}^\alpha_{(1)} \supseteq \mathcal{I}^\alpha_{(2)} \supseteq  \ldots \supseteq \mathcal{I}^\alpha_{(n)}. 
    $
    These intervals are simultaneous $(1-2\alpha)$ confidence intervals for all quantiles of individual effects $\tau_{(k)}$s, 
    i.e., 
        $
        \Pr( 
        \cap_{k=1}^{n} 
        \{ 
        \tau_{(k)} \in \mathcal{I}_{(k)}^\alpha
        \}
        ) \ge 1 - 2\alpha. 
        $
\end{theorem}

In Theorem \ref{thm:comb_all}, we use \citet{CDLM21quantile}'s approach twice, one with the original data and the other with the switched treatment labels and changed outcome signs, and then combine the resulting intervals for top quantiles, throwing away those uninformative $(-\infty, \infty)$ intervals. Consequently, the $(1-2\alpha)$ confidence interval in Theorem \ref{thm:comb_all} is narrower than or equal to the corresponding $(1-\alpha)$ interval in \citet{CDLM21quantile}. This comparison, however, is not entirely fair because the two intervals have different coverage guarantees. In general, the $(1-\alpha)$ interval from Theorem \ref{thm:comb_all} need not be narrower than the $(1-\alpha)$ interval in \citet{CDLM21quantile}. In practice, however, we often find it to be narrower; see the simulations and data analyses in Section \ref{sec:illustration}.


{\rev 
Our approach in Theorem \ref{thm:comb_all} shares a similar spirit as \citet[][Section 2]{Rigdon:2015}. 
Specifically, we combine prediction sets for treatment effect quantiles among treated and control units to construct confidence sets for treatment effect quantiles among all experimental units, whereas \citet{Rigdon:2015} combine prediction sets for average treatment effects among treated and control units (termed as attributable effects by \citet{Rosenbaum:2001}) to construct confidence sets for the average treatment effect among all experimental units. 
The similarity comes from the similar forms of test statistics for the randomization test.
\citet{Rigdon:2015} use $\sum_{i=1}^n Z_i Y_i(0)$ as the test statistic as suggested by \citet{Rosenbaum:2001}, 
while we essentially use $\sum_{i=1}^n Z_i \phi( \rank_i(\bs{Y}(0)) )$ as the test statistic, but in a worst case sense since the true $\bs{Y}(0)$ is unknown. 
Both test statistics are of form $t(\bs{Z}, \bs{Y}(0))$, depending only on the treatment assignments and control potential outcomes. Note that the control potential outcomes are observed for control units. 
Thus, the corresponding randomization test concerns essentially only about treatment effects for treated units, and the test inversion can further lead to prediction sets for treatment effects among treated units, such as their average or quantiles. 
By the same logic, we can construct prediction sets for treatment effects among control units using test statistic of form $t(\bs{Z}, \bs{Y}(1))$; 
this is equivalent to using test statistic of form $t(\bs{Z}, \bs{Y}(0))$ but with switched treatment labels. 
Finally, we can combine prediction sets for effects among treated and control units to construct confidence sets for effects among all units, as in \citet[][Section 2]{Rigdon:2015} and Theorem \ref{thm:comb_all}. 

From the above, the strategies discussed in this and next subsections 
are generally useful when treatment effects are heterogeneous and we consider test statistics of form $t(\bs{Z}, \bs{Y}(0))$ or $t(\bs{Z}, \bs{Y}(1))$,
which is common in causal inference \citep[see, e.g.,][]{Rosenbaum02a}.

\begin{remark}   \label{rmk: inference for attributable effect}
For binary outcomes, our approach can also be adapted to test average treatment effects among treated (or control) units. 
To construct prediction sets by test inversion, 
our approach requires at most $O(\log n)$ randomization tests due to the monotonicity of the corresponding randomization $p$-values, whereas the approaches in \citet{Rosenbaum:2001} and \citet{Rigdon:2015} generally require $O(n)$ randomization tests. 
However, 
our test can be sensitive to 
randomness in tie breaking for binary outcomes 
and is less preferable to \citet{Rosenbaum:2001}; see the supplementary material for a numerical illustration. 
Nevertheless, the computational benefit of our 
approach may be useful for handling more general discrete outcomes, a direction we leave for future investigation; see the supplementary 
material for a more detailed discussion as well as connections to \citet{LD2016binary}.
\end{remark}
}
\subsection{Controlling the number of treated or control units with large effects}\label{sec:berger}

\subsubsection{Inference for a single quantile of individual treatment effects}

From Section \ref{sec:review_drawback}, the $p$-value $p_{k,c}^n$ in \eqref{eq:p_nkc} proposed by \citet{CDLM21quantile} depends essentially on the number of treated units with individual effects greater than $c$, 
which, under the null $H_{k,c}^n$ in \eqref{eq:Hnkc}, is 
at most 
$(n-k) \cdot n_{\treat}/n$ 
in expectation. 
Hence,  intuitively, we may use the $p$-value $p_{k',c}^n$, with $n-k' = (n-k) \cdot n_{\treat}/n$, to test the null $H_{k,c}^n$; 
obviously, because $k'\ge k$, this leads to a smaller $p$-value $p_{k',c}^n \le p_{k,c}$.
However, simply using $p_{k',c}^n$ may fail to control the type-I error, 
because the number of treated units with effects greater than $c$ is only bounded by $n-k'$ in expectation and may be greater than $n-k'$ by chance. 
Indeed, the number of treated units with effects greater than $c$ is stochastically bounded by a Hypergeometric distribution under $H^n_{k, c}$.  
Based on this, we can utilize \citet{Berger1994}'s approach to deal with this nuisance parameter for our inference on effect quantiles, 
as detailed below; 
see also \cite{DFM2016}
and \cite{zhang2022bridging} for other applications of \citet{Berger1994}'s approach in randomization-based inference. 
For any integer $N\ge \max\{K, n\}$, 
let $\text{HG}(N, K, n)$ denote the Hypergeometric distribution that describes the number of successes in $n$ draws, without replacement, from a finite population of size $N$ with $K$ successes. 

\begin{theorem}\label{thm:BB_pval}
Consider the CRE 
and any rank score statistic in \eqref{eq:rank_score}. 
For any $0\le k \le n$ and $c\in \mathbb{R}$, 
and any prespecified integer $0\le k'\le n_{\treat}$, 
\begin{align}\label{eq:pval_BB}
    \overline{p}_{k, c}^n(k')
    &  \equiv 
    p_{n_{\control}+k', c}^n + 
    \Pr\{
    {\rev \hyper(k)} > n-(n_{\control}+k')
    \}
    = p^{n,\treat}_{k', c} + \Pr( 
    {\rev \hyper(k)} > n_{\treat} - k'
    )
\end{align}
is a valid $p$-value for testing the null hypothesis $H^n_{k,c}$ in \eqref{eq:Hnkc}, 
where ${\rev \hyper(k)} \sim \text{HG}(n, n-k, n_{\treat})$ is a Hypergeometric random variable, 
and 
$p_{n_{\control}+k', c}^n$ and $p^{n,\treat}_{k', c}$ are defined as in \eqref{eq:p_nkc} and \eqref{eq:p1}. 
\end{theorem}

The valid $p$-value in \eqref{eq:pval_BB} consists of two terms. 
The first can be viewed as a valid $p$-value, in the sense of being stochastically 
no less than $\textup{Unif}(0,1)$, when the number of treated units with effects greater than $c$ is bounded by $n_{\treat}-k'$, while the second is a correction term bounding the probability that this event fails.
When 
$k' = \max\{0, k-n_{\control}\}$, 
the correction term becomes zero, and the $p$-value in \eqref{eq:pval_BB} reduces to that in \citet{CDLM21quantile}. 

The choice of $k'$ in \eqref{eq:pval_BB} is a researcher choice, which impacts the test power and involves a trade-off. 
On the one hand, we want $k'$ to be large so that $p^{n,\treat}_{k', c}$ is small;  
on the other hand, we want $k'$ to be small so that the correction term is small. 
We suggest choosing $k'$ such that the correction term is close to, yet bounded from above by, $\gamma\alpha$, for some $\gamma\in [0,1)$. 
This is equivalent to choosing $k'$ as $n_{\treat} - q_{\textup{HG}}(1-\gamma\alpha; n, n-k, n_{\treat})$, where $q_{\textup{HG}}(\cdot)$ denotes the quantile function for the Hypergeometric distribution. 
The influence of $\gamma$ on testing power will be investigated through simulation; see Section \ref{sec:simu_sugg}.
As a side remark, in principle, the choice of $k'$ can depend on both $k$ and $c$ for the null $H_{k,c}$. 
Nevertheless, we will consider the choice of $k'$ that relies only on $k$, and make this implicit throughout the paper. 
This is intuitive since $n_{\treat}-k'$ represents the number of units with the largest $n-k$ individual effects that are assigned to the treated group, which is not related to $c$ directly. 

By test inversion,  
Theorem \ref{thm:BB_pval} then provides confidence intervals for treatment effect quantiles $\tau_{(k)}$s, as well as the number of units with effects greater than any threshold. 
\begin{theorem}\label{thm:sing_interval_BB}
    Consider the CRE,
    any rank score statistic in \eqref{eq:rank_score}, 
    and any $\alpha\in (0,1)$.  
    For $1\le k \le n$ and any prespecified $0\le k'\le n_\treat$, 
        $
        \{c: \overline{p}_{k,c}^n(k') > \alpha\}
        = 
        \{c: p^{n,\treat}_{k', c} > \alpha - \Pr \{ {\rev \hyper(k)} > n_{\treat}-k'\} \}
        = 
        \mathcal{I}^{\alpha'}_{\treat(k')}
        $
        is a $1-\alpha$ confidence interval for $\tau_{(k)}$, 
        where 
        ${\rev \hyper(k)} \sim \textup{HG}(n, n-k, n_{\treat})$, 
        $\alpha' = \alpha - \Pr\{ {\rev \hyper(k)} > n_{\treat}-k' \}>0$ and 
        $\mathcal{I}^{\alpha'}_{\treat(k')}$ is defined as in Theorem \ref{thm:treated}(b). 
\end{theorem}

From Theorem \ref{thm:sing_interval_BB}, we can also construct confidence set for $n(c)$ by checking whether the confidence intervals for $\tau_{(k)}$s contain $c$; this is equivalent to test inversion using the $p$-value in Theorem \ref{thm:BB_pval}. 
We will give a different interpretation of the confidence intervals in Theorem \ref{thm:sing_interval_BB} in the next subsection, which builds its connection to Theorem \ref{thm:treated} about 
effects on treated.

\subsubsection{Simultaneous inference for multiple quantiles of individual effects}

We next consider simultaneous inference for multiple treatment effect quantiles. 
Let 
$\tau_{(k_1)} \le \tau_{(k_2)} \le \ldots \le \tau_{(k_J)}$ with 
$0\le k_1 < k_2 < \ldots < k_J \le n$ be the $J$ effect quantiles of interest. 
One straightforward way is to combine individual confidence intervals in Theorem \ref{thm:sing_interval_BB}(a) using Bonferroni's method. 
Theorem \ref{thm:BB_simulCI} below improves this naive approach by utilizing dependence 
among numbers of treated units with large effects. 
For any $0\le k_1', \ldots, k_J' \le n_{\treat}$ 
define 
\begin{align}\label{eq:delta_multiple}
    \Delta_{\textup{H}}(k_{1:J}'; n, n_{\treat}, k_{1:J}) \equiv 
    \Pr\left( \bigcup_{j=1}^J \Big\{ \sum_{i=j}^J {\rev \hyper_{i}(k_{1:J}) } > n_{\treat} - k_j' \Big\} \right), 
\end{align}
where ${\rev (n_{\treat}-\sum_{j=1}^J \hyper_j(k_{1:J}), \hyper_1(k_{1:J}), \ldots, \hyper_{J-1}(k_{1:J}), \hyper_J(k_{1:J}))} \sim \textup{MHG}(
k_1, 
k_2-k_1, 
\ldots, 
k_J - k_{J-1}, 
n-k_J; n_{\treat}
)$
follows a multivariate Hypergeometric distribution. 
Specifically, 
$\textup{MHG}(
a_0, 
a_1,$ $ 
\ldots, 
a_{J-1}, 
a_{J}; n_{\treat}
)$
describes the numbers of balls of colors $0$ to $J$ in a random sample of size $n_{\treat}$, drawn without replacement from an urn with $a_j$ balls for each color $j$. 

\begin{theorem}\label{thm:BB_simulCI}
Consider the CRE and any rank score statistic. 
Let 
$\tau_{(k_1)} \le \tau_{(k_2)} \le \ldots \le \tau_{(k_J)}$ be the treatment effect quantiles of interest. 
For any $\alpha\in (0, 1)$ 
and
any $0\le k_1' \ldots, k_J' \le n_{\treat}$, 
the intervals $\mathcal{I}_{\treat(k_j')}^{\alpha'}$s  defined as in Theorem \ref{thm:treated}(b) 
with 
$\alpha' = \alpha - \Delta_{\textup{H}}(k_{1:J}'; n, n_{\treat}, k_{1:J})>0$ 
are simultaneously valid confidence intervals for 
$\tau_{(k_j)}$s, i.e., 
$
    \Pr( 
        \cap_{j=1}^J \{ \tau_{(k_j)} \in  \mathcal{I}_{\treat(k_j')}^{\alpha'} \}
        )
    \geq 1-\alpha. 
$
\end{theorem}

The construction of the simultaneous confidence intervals in Theorem \ref{thm:BB_simulCI} comprises mainly two components. First, we consider the event where the number of units with the largest $n-k_j$ individual effects that receive treatment is bounded by $n_{\treat} - k_j'$ for all $j$. 
Second, once this event holds, intervals from inverting the $p_{k_j', c}^{n,\treat}$s with respect to any level $\alpha$ can cover the quantiles $\tau_{(k_j)}$s simultaneously with probability at least $1-\alpha$.
Importantly, the probability that the previous event fails has a sharp bound in 
\eqref{eq:delta_multiple}
that is more accurate than Bonferroni's bound, which can be represented by multivariate Hypergeometric distributions and approximated using Monte Carlo method. 
Such a sharp bound is beneficial in our context because these events $\{\sum_{i=j}^J  {\rev \hyper_i(k_{1:J})} > n_{\treat}-k_j'\}$s for $1\le j \le J$,
which describe the numbers of units with large effects that are 
assigned to the treated group, 
are highly positively associated.

\begin{remark}\label{rmk:choice_k_prime}
The choice of $k_j'$s in Theorem \ref{thm:BB_simulCI} is a researcher choice and can impact the inferential power. 
Similar to the discussion for Theorem \ref{thm:BB_pval}, we may choose $k_j'$s such that the correction term $\Delta_{\textup{H}}$ is close to, but bounded from above by $\gamma\alpha$, for some $\gamma\in (0,1)$.
Below we give a strategy to choose such $k_j'$s. 
For each $\kappa \in (0,1)$, 
we let $k_j' = n_{\treat} - q_{\textup{HG}}(1-\kappa \gamma\alpha; n, n-{\rev k_{j}}, n_{\treat})$ for each $j$; 
as discussed shortly in the next subsection, 
$\tau_{\treat(k_j')}$ is actually a $1-\kappa \gamma\alpha$ individual ``confidence'' bound for $\tau_{(k_j)}$.  
We then choose the largest possible $\kappa$ such that the corresponding $\Delta_{\textup{H}}$ is bounded by $\gamma\alpha$. 
By Bonferroni's inequality, it is not difficult to see that this desired value of $\kappa$ must be between $J^{-1}$ and $1$.
When $J=1$, $\kappa$ must be $1$, and this strategy reduces to that discussed after Theorem  \ref{thm:BB_pval} for inferring a single effect quantile. 
\end{remark}

\begin{remark}\label{rmk:bb_band}
    Theorem \ref{thm:BB_simulCI} can be used to construct simultaneous confidence bands for the whole quantile and distribution functions. 
    Specifically, let $[\underline{c}_j, \infty)$ be a $1-\alpha$ simultaneous confidence interval for effect quantiles $\tau_{(k_j)}$, 
    for some prespecified $1\le k_1 < k_2 < \ldots < k_J\le n$ and 
    $\alpha\in (0,1)$, as constructed in Theorem \ref{thm:BB_simulCI}.
    By monotonicity of the quantiles, 
    a simultaneous confidence band for $\tau_{(k)}$ with $1\le k \le n$  is  $[\underline{c}_{q(k)}, \infty)$, where $q(k)$ is defined as the maximum $j$ such that $k_j\le k$.  
    These simultaneous confidence bands can be easily visualized, as illustrated in Section \ref{sec:teacher}. 
    In practice, we suggest choosing a set of quantiles $\tau_{(k_j)}$s of most interest. 
\end{remark}

\begin{remark}\label{rmk:BB_comb}
    Analogously, we can also construct confidence intervals for $\tau_{(k_j)}$s utilizing effect quantiles of control units; this can be achieved by switching treatment labels and changing outcome signs. 
    In practice, we can combine intervals obtained from effect quantiles of both treated and control units using Bonferroni's method.
\end{remark}

{\rev 
\subsection{A comparison between the two improved methods}\label{sec:comp}

In Sections \ref{sec:effect_treated} and \ref{sec:berger} we presented two methods to improve the approach in \citet{CDLM21quantile}, which can be overly conservative by assuming units with large effects are all in treated (or control) group. 
Below we make a connection between the two improved methods and compare their performance. 
As discussed below, both of them rely on the key result in Theorem \ref{thm:treated}, which reinterprets the confidence intervals in \citet{CDLM21quantile} for the top $n_\treat$ quantiles as the prediction intervals for the effect quantiles among only treated units.

Specifically, the first method in Section \ref{sec:effect_treated} constructs prediction intervals for treatment effects among treated and control units, separately, and then combine them to infer effects among all units. 
On the contrary, the second method in Section \ref{sec:berger}  essentially uses only treatment effects among treated (or control) units to infer effects among all units, utilizing the fact that the set of treated units is a simple random sample of all units. 
In particular, the confidence interval for a single effect quantile in Theorem \ref{thm:BB_simulCI} can be equivalently constructed in the following steps. 
First, we construct a lower ``confidence'' bound for the effect quantile $\tau_{(k)}$ of all units using the effect quantile $\tau_{\treat(k')}$ of treated units, 
whose uncoverage probability is
at most $\Pr({\rev \hyper(k)} > n_{\treat} - k')$ with ${\rev \hyper(k)} \sim \text{HG}(n, n-k, n_{\treat})$ \citep{Sedransk}; 
note that this is not a calculable confidence bound since $\tau_{\treat(k')}$ is generally unobserved. 
Second, we construct a prediction interval for the effect quantile $\tau_{\treat(k')}$ of treated units by Theorem \ref{thm:treated}. 
Analogously, the simultaneous intervals in Theorem \ref{thm:BB_simulCI} can also be constructed through the following two steps: 
construct first simultaneous ``confidence'' bounds for effect quantiles $\tau_{(k_j)}$s using effect quantiles $\tau_{\treat(k_j')}$s among treated units,  whose simultaneous coverage probability is at least $1-\Delta_{\textup{H}}(k_{1:J}'; n, n_{\treat}, k_{1:J})$ \citep{chen2023role}, 
and construct then simultaneous prediction intervals for effect quantiles among treated units by Theorem \ref{thm:treated}.

The discussion above also helps illustrate the advantage of our methods over \citet{CDLM21quantile}. 
Specifically, \citet{CDLM21quantile} essentially supplement the prediction intervals for treated units with uninformative $(-\infty, \infty)$ intervals for control units to obtain confidence intervals for all units.
Instead, our method in Section \ref{sec:effect_treated} supplements the prediction intervals for treated units with prediction intervals for control units to obtain confidence intervals for all units, 
and our method in Section \ref{sec:berger} uses the prediction intervals for treated units to infer effect quantiles among all units, using the property that treated units constitute a simple random sample of all units. 
These methods, although requiring some adjustment in the confidence level, can typically provide substantial gain in narrowing confidence intervals.

The relative performance between the two improved methods in Sections \ref{sec:effect_treated} and \ref{sec:berger} then relies on the informativeness of the prediction intervals for treated and control units, which depends crucially on the tail behavior of the two potential outcome distributions. In particular, the relative heaviness of the right (or left) tails of the two potential outcome distributions is more critical when inferring treatment effects among treated (or control) 
units. 
See the next subsection for a numerical illustration.

\begin{remark}\label{rmk:further_combine_two_improv}
    Both methods in Sections \ref{sec:effect_treated} and \ref{sec:berger} extend the inference from the treated group to all units in distinct ways:
    the former leverages inference for control units, 
    while the latter exploits the fact that the treated group constitutes a simple random sample from all units.
    Consequently, the two approaches are not readily integrable to further enhance inference, aside from a straightforward combination using Bonferroni's method.
    In the context of the CRE and the later stratified experiment, the approach in Section \ref{sec:effect_treated} has some desirable properties as discussed in Sections \ref{sec:num_illu_improv}, \ref{sec:scre} and \ref{sec:simu_sugg}. 
    However, in sampling-based randomized experiments studied in Section \ref{sec:samp_based}, the ideas in Section \ref{sec:berger} become essential since we cannot directly infer treatment effects among unsampled units. 
\end{remark}

\subsection{Numerical illustration}\label{sec:num_illu_improv}

To illustrate the improvement, we consider the education experiment in \citet{CDLM21quantile} with $n = 233$ units, among which $n_\treat = 164$ and $n_\control = 69$ units receive treatment and control, respectively; see Section \ref{sec:teacher}. 
We focus on constructing $90\%$ lower confidence bound for the third quartile of the individual effects, 
$\tau_{(175)}$, using the Stephenson rank sum statistic with $s=6$, which has the form in \eqref{eq:rank_score} with $\phi(r) = \binom{r-1}{s-1}$ if $r\ge s$ and $0$ otherwise. 
Depending on whether we perform treatment label switching, 
a $90\%$ lower confidence bound for $\tau_{(175)}$ is either $6.67$ or $-\infty$. 
From Theorem \ref{thm:treated}, these are equivalently $90\%$ lower prediction bounds for $\tau_{\treat(106)}$ and $\tau_{\control(11)}$, respectively, 
and the validity of them as a confidence bound for $\tau_{(175)}$ also follows from the fact that $\tau_{(175)} \ge \tau_{\treat(106)}$ and $\tau_{(175)} \ge \tau_{\control(11)}$\footnote{\rev We can prove these two inequalities by contradiction. For example, if $\tau_{(175)} < \tau_{\treat(106)}$, then at most $105$ treated units have effects bounded by $\tau_{(175)}$, and consequently at most $105+n_\control = 105+69 = 174$ units in total have effects bounded by $\tau_{(175)}$, leading to a contradiction.}. 

We consider first the improvement from Section \ref{sec:effect_treated}. 
The key is to note that, in addition to $\tau_{(175)} \ge \tau_{\treat(106)}$ and $\tau_{(175)} \ge \tau_{\control(11)}$, 
we also have
$\tau_{(175)} \ge \min\{ \tau_{\treat(j)}, \tau_{\control(176-j)} \}
$ for $107 \le j \le 164$\footnote{\rev We can also prove this analogously by contradiction. If $\tau_{(175)}<\min\{ \tau_{\treat(j)}, \tau_{\control(176-j)}\}$, then at most $j-1$ treated units have effects bounded by $\tau_{(175)}$, at most $175-j$ control units have effects bounded by $\tau_{(175)}$, and consequently at most $(j-1)+(175-j) = 174$ units in total have effects bounded by by $\tau_{(175)}$, leading to a contradiction.}.
Thus, for any $107 \le j \le 164$, by Bonferroni's method, the minimum of $95\%$ lower prediction bounds for $\tau_{\treat(j)}$ and $\tau_{\control(176-j)}$ from Theorem \ref{thm:treated} is also a valid $90\%$ lower confidence bound for $\tau_{(175)}$. 
Furthermore, because the prediction intervals for effect quantiles among both treated and control groups are simultaneously valid, we can search over  $107 \le j \le 164$ to maximize such a confidence bound; in particular, the maximum is $10.00$ and is achieved at $120\le j \le 125$. 
For example, both $\tau_{\treat(120)}$ and $\tau_{\control(56)}$ have $95\%$ lower prediction bounds equal to $10.00$, indicating that $10.00$ is also a $90\%$ lower confidence bound for $\tau_{(175)}$.  

We consider then the improvement from Section \ref{sec:berger}. 
As discussed in Section \ref{sec:comp}, in the first step, we construct a lower ``confidence'' bound for $\tau_{(175)}$ using an effect quantile $\tau_{\treat(k')}$ of treated units, where we choose $k'$ such that the uncoverage probability is bounded by but as close as possible to $\tfrac{1}{2} \cdot 10\% = 5\%$. 
In this context, $k' = 118$, with a corresponding uncoverage probability of $2.7\%$. 
In the second step, we then construct a $1-(10-2.7)\% = 92.7\%$ lower prediction bound for $\tau_{\treat(118)}$, which is $10.00$. These two steps then indicate that $10.00$ is a $90\%$  lower confidence bound for $\tau_{(175)}$.  
We can perform similar analysis using effects among the control group, which leads to a $90\%$ lower confidence bound $6.67$ for $\tau_{(175)}$. 

Section \ref{sec:teacher} shows the confidence bounds for all quantiles using the above methods. 
The confidence bounds from the method in Section \ref{sec:effect_treated} are more informative. 
This is because, based on the observed data, the relative heaviness of the right tails of the distributions of the two potential outcomes are similar to that of the left tails, and the inference for effects among treated and control groups is approximately equally informative; see the  discussion near the end of Section \ref{sec:comp}. 
Nevertheless, Section \ref{sec:berger} can be more powerful in some contexts. 
For example, if we rescale the outcome for treated units, maintaining its mean while tripling its standard deviation, 
then a $90\%$ lower confidence bound for $n(0)$, the number of units with positive effects, 
from the methods in Sections \ref{sec:effect_treated} and \ref{sec:berger} are, respectively, $54$
and $70$. 
Specifically, we apply the method in Section \ref{sec:berger} to use only treatment effects among treated units to infer effects among all units. This is more advantageous here, because none of the lower prediction bounds for effect quantiles among control units are positive and the method in Section \ref{sec:effect_treated} does not provide any improvement when inferring $n(0)$; see Section \ref{sec:simu_sugg} and the supplementary material for a more detailed simulation study.

}

\secadj
\section{Stratified randomized experiments}\label{sec:str}
Section \ref{sec:cre} focuses on the CRE where treatment assignments are exchangeable across all units. 
\citet{SL22quantile} extended \citet{CDLM21quantile}'s inference on effect quantiles to stratified randomized experiments, where treatment assignments are exchangeable only within each stratum. 
However, their approach also suffers from the same critique as in Section \ref{sec:review_drawback}. 
Below we will use ideas from Section \ref{sec:cre} to improve \citet{SL22quantile}'s approach.

\subsection[]{Review of \citet{SL22quantile}'s approach and its drawbacks}\label{sec:review_SL}

Consider a stratified completely randomized experiment (SCRE), where all the $n$ experimental units are divided into $S$ strata and units within each stratum are completely randomized into treatment and control. 
To facilitate discussion, 
we will index each unit by double indices $si$ with $1\le s\le S$ and $1\le i \le n_s$, 
where $n_s$ denotes the number of units within each stratum $s$ and 
$\sum_{s=1}^S n_s = n$. 
We adopt the notation from Section \ref{sec:po_assign}, including the potential outcomes, individual effects, observed outcomes, treatment assignments and their vector representations, except that they are now being double indexed.   
For each stratum $s$, we use $n_{s\treat}$ and $n_{s\control} \equiv n_s - n_{s\treat}$ 
to denote the numbers of treated and control units, respectively. 

\citet{SL22quantile} considered again the null hypothesis $H_{k, c}^n$ in \eqref{eq:Hnkc}, 
and proposed analogously to use the supremum of the Fisher randomization $p$-value over all possible sharp null hypotheses satisfying $H_{k, c}^n$. 
To make it computationally feasible, 
\citet{SL22quantile} utilized a stratified rank score statistic of the following form: 
\begin{align}\label{eq:strat_rank_score}
    \tilde{t}(\bs{z}, \bs{y}) = \sum_{s=1}^S t_s(\bs{z}_s, \bs{y}_s) 
    = 
    \sum_{s=1}^S \sum_{i=1}^{n_s} z_{si} \phi_s(\rank_i(\bs{y}_s)), 
\end{align}
where $\bs{z}_s$ and $\bs{y}_s$ denote the treatment assignment and outcome vectors corresponding to stratum $s$, 
and $t_s(\cdot, \cdot)$ is a rank score statistic defined as in \eqref{eq:rank_score}, i.e., $\phi_s$ is an increasing function for all $s$ and ties in ranking will be broken by randomly shuffled ordering of units. 
The resulting $p$-value from \citet{SL22quantile} has the following form: 
\begin{align}\label{eq:p_SL}
    \tilde{p}^n_{k,c} = 
    \tilde{G}\Big( \inf_{\bs{\delta} \in \mathcal{H}_{k,c}^n} \tilde{t}(\bs{Z}, \bs{Y} - \bs{Z} \circ \bs{\delta} ) \Big), 
    \quad \textup{where } \tilde{G}(x) = \Pr\{ \tilde{t}(\bs{Z}, \bs{y}) \ge x \} \textup{ for any } \bs{y}\in \mathbb{R}^n. 
\end{align}
In \eqref{eq:p_SL}, 
{\rev $\tilde{G}(\cdot)$ does not vary with $\bs{y}$ due to the distribution free property of the rank-based statistic in \eqref{eq:strat_rank_score}, i.e., $\tilde{t}(\bs{Z}, \bs{y})\sim \tilde{t}(\bs{Z}, \bs{y}')$ under the SCRE for any $\bs{y}, \bs{y}'\in \mathbb{R}^n$,} 
and $\mathcal{H}_{k,c}^n$ is defined as in Section \ref{sec:review_drawback}.
\citet{SL22quantile} showed that the optimization for $\tilde{p}^n_{k,c}$ can be transformed into instances of the multiple-choice knapsack problem, and can 
be efficiently solved either exactly or conservatively. 
However, the $p$-value in \eqref{eq:p_SL} can be overly conservative due to the same reason as discussed in Section \ref{sec:review_drawback}.  
That is, 
it 
considers the worst-case scenario where units with large effects are all assigned to 
treatment, 
which is unlikely 
due to randomization. 

\subsection{Stratified completely randomized experiments}\label{sec:scre}

Below we improve the approach in \citet{SL22quantile} utilizing ideas from Section \ref{sec:effect_treated}. 
Similar to \citet{CDLM21quantile}'s approach for the CRE discussed in Section \ref{sec:review_drawback}, 
the $p$-value $\tilde{p}^n_{k,c}$ in \eqref{eq:p_SL} used by \citet{SL22quantile} is based on the worst-case consideration of the randomization test using the test statistic $\tilde{t}(\bs{Z}, \bs{Y}(0))$ and thus 
depends essentially on the treatment effects for treated units.
Therefore, intuitively, we can also view the $p$-value in \eqref{eq:p_SL} as a $p$-value for testing the null hypothesis that the number of treated units with effects greater than $c$ is bounded by $n-k$; 
again, this is not a traditional null hypothesis since it can be random due to the randomization of treatment assignments. 
Define the null hypothesis $H_{k,c}^{n, \treat}$ and the corresponding set  $\mathcal{H}_{k,c}^{n, \treat}\in \mathbb{R}^n$ the same as in Section \ref{sec:effect_treated}.
The theorem below follows by the same logic as Theorem \ref{thm:treated}. 
For conciseness, we present its short version below and relegate the details to the supplementary material.

\begin{theorem}\label{thm:treated_scre}
    Consider the SCRE and any stratified rank score statistic $\tilde{t}(\cdot, \cdot)$ in \eqref{eq:strat_rank_score}. 
    \begin{enumerate}[label=(\roman*), topsep=0ex,itemsep=-1ex,partopsep=1ex,parsep=1ex]
        \item[(i)] For any $0\le k \le n_\treat \equiv \sum_{s=1}^S n_{s\treat}$ and any $c\in \mathbb{R}$, 
        \begin{align*}
            \tilde{p}_{k,c}^{n, \treat} 
            \equiv 
            \tilde{G}\Big( \inf_{\bs{\delta} \in \mathcal{H}_{k,c}^{n, \treat}} \tilde{t}(\bs{Z}, \bs{Y} - \bs{Z} \circ \bs{\delta} ) \Big)
            = 
            \tilde{G}\Big( \inf_{\bs{\delta} \in \mathcal{H}_{n_{\control}+k,c}^{n}} \tilde{t}(\bs{Z}, \bs{Y} - \bs{Z} \circ \bs{\delta} ) \Big)
            = \tilde{p}_{n_{\control}+k,c}^{n} 
        \end{align*}
    is a valid $p$-value for testing $H_{k,c}^{n, \treat}$ in \eqref{eq:Hn1kc}, 
    where $\tilde{p}_{n_{\control}+k,c}^{n}$ is defined as in \eqref{eq:p_SL}.  
    \item[(ii)] (b) and (c) in 
    Theorem \ref{thm:treated} hold with $p_{k,c}^{n, \treat}$ replaced by $\tilde{p}_{k,c}^{n, \treat}$. 
    \end{enumerate}
\end{theorem}

The discussions in Section \ref{sec:effect_treated} for Theorem \ref{thm:treated} hold analogously for Theorem \ref{thm:treated_scre}.
First, the prediction intervals for effect quantiles of treated units are the same as the confidence intervals for the top $n_{\treat}$ effect quantiles of all units under the original \citet{SL22quantile}'s approach. 
This new interpretation can thus make the inference results more informative. 
Second, by the same logic as Theorem \ref{thm:treated_scre}, 
we can also obtain simultaneous prediction intervals for effect quantiles of control units, simply by switching treatment labels and changing outcome signs. 
Finally, we can combine the prediction intervals for treated and control units to get simultaneous confidence intervals for effect quantiles of all units as in Theorem \ref{thm:comb_all}.

\begin{remark}
{\rev 
In principle, we can also extend ideas in Section \ref{sec:berger} to stratified experiments, noting that the treated units are from a stratified sampling of all units. However, this can be challenging due to 
the unknown treatment effect heterogeneity across strata \citep[see, e.g.,][]{Sedransk}.
Therefore, we focus on utilizing ideas from Section \ref{sec:effect_treated} here. 
}    
\end{remark}

\subsection{Matched observational studies}

We then extend the discussion to matched observational studies, 
which can be viewed as stratified experiments but with possibly unknown and nonexchangeable treatment assignments within each stratum due to unmeasured confounding. 
We adopt the notation from Section \ref{sec:scre}, where each stratum corresponds to a matched set. 
For descriptive simplicity, we assume there is only one treated unit per matched set, i.e., $n_{s\treat} = 1$ for all $1\le s\le S$.   
The results can be extended to matched studies with one treated or control unit per matched set; see, e.g., \citet[][pp. 161--162]{Rosenbaum02a} and \citet[][Remark 8]{SL22quantile}. 

Following \citet{Rosenbaum02a}, we assume that the treatment assignment mechanism has the following form for some unknown $\bs{u} = (u_{11}, u_{12}, \ldots, u_{Sn_S}) \in [0,1]^n$ and $\Gamma\ge 1$: 
\begin{align}\label{eq:sen_assign}
    \Pr_{\bs{u}, \Gamma}(\bs{Z} = \bs{z})
    & = \prod_{s=1}^S \frac{\exp(\gamma\sum_{i=1}^{n_{s}} z_{si} u_{si} )}{
    \sum_{i=1}^{n_{s}} \exp(\gamma u_{si} )
    }
    \cdot 
    \I(\bs{z} \in \mathcal{Z}), 
\end{align}
where $\gamma = \log(\Gamma)$ and $\mathcal{Z}$ denotes the set of all possible assignments such that there is exactly one treated unit per matched set. 
Under \eqref{eq:sen_assign}, within any matched set, a unit is at most $\Gamma$ times more likely to receive treatment than any other unit. 
Intuitively, $\Gamma$ measures the strength of the unmeasured confounding, 
and $u_{si}$s represent the hidden confounding associated with all the units. 
The sensitivity analysis investigates how the inference results for treatment effects vary with $\Gamma$. 
We refer readers to 
\citet{Rosenbaum1987}, \citet{GKR2000}, \citet{Rosenbaum02a}, \citet{WL23} and \citet{li2024sensitivity}
for more details about the formulation and interpretation of this sensitivity analysis framework. 

\citet{SL22quantile} considered again the null hypothesis in \eqref{eq:Hnkc} regarding treatment effect quantiles of all units in the matched study. 
To ensure the validity of a test for 
$H_{k,c}^n$, they proposed to take supremum of the Fisher randomization $p$-value over all possible treatment effects satisfying $H_{k,c}^n$ and all possible configurations of hidden confounders $\bs{u}\in [0,1]^n$. 
To overcome the computational difficulty, they further utilized the stratified rank score statistic in \eqref{eq:strat_rank_score}, 
and demonstrated that the optimization for the resulting valid $p$-value can be separated into two parts: one over all possible configurations of hidden confounders and the other over all possible individual treatment effects. 
Specifically, for any stratified rank score statistic, 
its worst-case distribution over all possible hidden confounders enjoys a distribution free property, i.e.,  
$\tilde{G}_{\Gamma} (c) = \max_{\bs{u}\in [0,1]^n}\Pr_{\bs{u}, \Gamma}\{ \tilde{t}(\bs{Z}, \bs{y}) \ge c
\}$ 
does not vary with 
$\bs{y} \in \mathbb{R}^n$. 
The resulting valid $p$-value in \citet{SL22quantile} for testing $H_{k,c}^n$ 
in \eqref{eq:Hnkc} 
under a given bound $\Gamma$ on unmeasured confounding strength 
has the following form:
\begin{align}\label{eq:p_SL_sen}
    \tilde{p}^n_{k,c}(\Gamma) = 
    \tilde{G}_{\Gamma}\Big( \inf_{\bs{\delta} \in \mathcal{H}_{k,c}^n} \tilde{t}(\bs{Z}, \bs{Y} - \bs{Z} \circ \bs{\delta} ) \Big).
\end{align}
In \eqref{eq:p_SL_sen}, 
the minimization of the test statistic can be done in the same way as in the SCRE, 
and the worst-case distribution $\tilde{G}_{\Gamma}(\cdot)$ over unmeasured confounding can be approximated through either the Monte Carlo method or a large-sample Gaussian approximation \citep{Rosenbaum02a}.
Note that when $\Gamma=1$, $\tilde{G}_{\Gamma}(\cdot)$ and $\tilde{p}^n_{k,c}(\Gamma)$ reduce to $\tilde{G}(\cdot)$ and $\tilde{p}^n_{k,c}$ in \eqref{eq:p_SL}. 

The $p$-value in \eqref{eq:p_SL_sen} can be overly conservative due to the same reason as discussed in Sections \ref{sec:review_drawback} and \ref{sec:review_SL}. 
By the same logic as Sections \ref{sec:effect_treated} and \ref{sec:scre},
we will improve \citet{SL22quantile}'s sensitivity analysis by reinterpreting the $p$-value in \eqref{eq:p_SL_sen} as a valid $p$-value for testing effect quantiles among treated units. 
For conciseness, we summarize the results in the following theorem, 
with details relegated to the supplementary material.  

\begin{theorem}\label{thm:treated_sen}
    Consider a matched observational study and any stratified rank score statistic $\tilde{t}(\cdot, \cdot)$ in \eqref{eq:strat_rank_score}. Suppose that the unmeasured confounding strength is bounded by $\Gamma\ge 1$, i.e., the treatment assignment follows \eqref{eq:sen_assign} for some unknown $\bs{u}\in [0,1]^n$. 
    \begin{enumerate}[label=(\roman*), topsep=0ex,itemsep=-1ex,partopsep=1ex,parsep=1ex]
        \item[(i)] For any $0\le k \le n_\treat \equiv \sum_{s=1}^S n_{s\treat}$ and any $c\in \mathbb{R}$, 
        \begin{align*}
            \tilde{p}_{k,c}^{n, \treat} (\Gamma)
            \equiv 
            \tilde{G}_{\Gamma}\Big( \inf_{\bs{\delta} \in \mathcal{H}_{k,c}^{n, \treat}} \tilde{t}(\bs{Z}, \bs{Y} - \bs{Z} \circ \bs{\delta} ) \Big)
            = 
            \tilde{G}_{\Gamma}\Big( \inf_{\bs{\delta} \in \mathcal{H}_{n_{\control}+k,c}^{n}} \tilde{t}(\bs{Z}, \bs{Y} - \bs{Z} \circ \bs{\delta} ) \Big)
            = \tilde{p}_{n_{\control}+k,c}^{n} (\Gamma)
        \end{align*}
    is a valid $p$-value for testing $H_{k,c}^{n, \treat}$ in \eqref{eq:Hn1kc}, 
    where $\tilde{p}_{n_{\control}+k,c}^{n} (\Gamma)$ is defined as in \eqref{eq:p_SL_sen}.
    \item[(ii)] 
    (b) and (c) in 
    Theorem \ref{thm:treated} hold with $p_{k,c}^{n, \treat}$ replaced by $\tilde{p}_{k,c}^{n, \treat}(\Gamma)$. 
    \end{enumerate}
\end{theorem}


Similar to Theorem \ref{thm:comb_all}, we can combine prediction intervals for treated units and control units to obtain confidence intervals for all units. 
Nevertheless, in a matched observational study, especially when we match each treated unit with one or multiple controls, 
we are often interested in treatment effects for treated units. 
In such cases, sensitivity analysis for effect quantiles among treated units as in Theorem \ref{thm:treated_sen} can be preferable to that for all units.  
Theorem \ref{thm:treated_sen} also extends \citet{Rosenbaum2002attribute} by allowing attributable effects defined in terms of treatment effect quantiles among treated units.

\secadj
\section{Sampling-based randomized experiments}\label{sec:samp_based}
We now study inference for distributions of individual treatment effects in sampling-based randomized experiments, 
where units are random samples from larger populations of interest. 
In particular, we assume that the study units are either simple random samples from a finite population or i.i.d.~samples from a superpopulation; 
the former is useful in design-based inference for multi-armed and survey experiments, 
while the latter often 
serves
as a convenient assumption for statistical inference. 
We focus on finite population setting in this section, and discuss briefly the superpopulation setting with details relegated to the supplementary material; the superpopulation setting can also be viewed as a limit of the finite population setting with the population size tending to infinity.

\subsection{Population distributions of individual treatment effects}

We 
consider the setting where 
the set of the $n$ study units is a simple random sample from a possibly larger finite population of size $N \ge n$.
For example, in a survey experiment, only a subset of units will be enrolled in an experiment due to, say, resource constraints; in a multi-arm experiment,
only a subset of units will receive the two treatments whose contrast is of interest.  
In such settings, we will use $\{( Y_{i}^\fp (1), Y_i^\fp(0) ): 1\le i\le N\}$ to denote the treatment and control potential outcomes for all $N$ units. 
The potential outcomes for the $n$ experimental units, 
$\{( Y_i(1), Y_i(0) ): 1\le i \le n\}$, 
are sampled randomly without replacement from the finite population $\{( Y_i^\fp(1), Y_i^\fp(0) ): 1\le i\le N\}$. 
Let $\tau_i^\fp \equiv Y_i^\fp(1) - Y_i^\fp(0)$, $1\le i \le N$, denote the individual treatment effect for unit $i$ in the finite population. 
We are interested in the finite population distribution of individual treatment effects with distribution function
$F_\fp(c) \equiv N^{-1} \sum_{i=1}^N \I(\tau_i^\fp \le c)$ for $c\in \mathbb{R}$ 
and quantile function $F_\fp^{-1}(\cdot)$. 


\subsection{
Inference for population distributions of individual treatment effects}\label{sec:fp}

Suppose we are interested in inferring treatment effect quantiles 
$F_\fp^{-1}(\beta_1), \ldots, F_\fp^{-1}(\beta_J)$, 
for some prespecified $0 \leq \beta_1 < \ldots \leq \beta_J < 1$. 
We utilize ideas in Section \ref{sec:berger} with their equivalent understanding in Section \ref{sec:comp} 
to construct simultaneous confidence bounds for $F_\fp^{-1}(\beta_j)$s. 
Specifically, we first construct simultaneous ``confidence'' bounds for the population effect quantiles using the sample effect quantiles among experimental units, and then construct simultaneous prediction bounds for these sample effect quantiles.  

\begin{theorem}\label{thm:samp_fp}
Consider an experiment in which the set of units is a simple random sample of size $n$ from a finite population of in total $N\ge n$ units, and suppose that the population effect quantiles $F_{\fp}^{-1}(\beta_j)$s are of interest, where $0\leq \beta_1 < \beta_2 < \ldots < \beta_J \leq 1$ and $J\ge 1$. 
For any $\alpha \in (0,1)$ and $0\le k_1' \le k_2' \le \ldots \le k_J' \le n$, 
define 
$\Delta_{\textup{H}}(k_{1:J}'; N, n, k_{1:J})$ the same as in \eqref{eq:delta_multiple} with $k_j = \lceil N \beta_j \rceil$ for all $j$, 
and let 
$\mathcal{I}^{\alpha'}_{(k_j')}$s be simultaneous $1-\alpha'$ one-sided prediction intervals of form $(c, \infty)$ or $[c, \infty)$ for sample effect quantiles $\tau_{(k_j')}$s,  
where 
$\alpha' = \alpha - \Delta_{\textup{H}}(k'_{1:J}; N, n, k_{1:J}) > 0$. 
Then
the $\mathcal{I}^{\alpha'}_{(k_j')}$s 
are simultaneous $1-\alpha$ confidence intervals for the quantiles $F_\fp^{-1}(\beta_j)$s, in the sense that 
$
    \Pr\{F_\fp^{-1}(\beta_j) \in  \mathcal{I}^{\alpha'}_{(k_j')} \textup{ for all } 1 \leq j \leq J \}\geq 1-\alpha. 
$    
\end{theorem}

In Theorem \ref{thm:samp_fp}, the simultaneous prediction intervals for sample effect quantiles can be constructed using, for example, Theorem \ref{thm:comb_all} or \ref{thm:treated_scre}, depending on the treatment assignment mechanism, and 
the $k_j'$s can be chosen in the same way as that discussed in Remark \ref{rmk:choice_k_prime}. 
Similar to Remark \ref{rmk:bb_band}, we can use the simultaneous intervals in Theorem \ref{thm:samp_fp} to construct simultaneous confidence bands for the whole quantile function $F_\fp^{-1}(\cdot)$ and equivalently the distribution function $F_\fp (\cdot)$. 
Moreover, 
similar to the discussion after Theorem \ref{thm:sing_interval_BB}, 
we can use individually valid confidence intervals from Theorem \ref{thm:samp_fp} with $J=1$ to construct individually valid confidence intervals for the proportion of units with effects greater than any threshold.


{\rev 
Theorem \ref{thm:samp_fp} is in the same spirit as Theorem \ref{thm:BB_simulCI}.
Theorem \ref{thm:BB_simulCI} uses treatment effects among treated units to infer those among all experimental units, 
while Theorem \ref{thm:samp_fp} uses treatment effects among experimental units to infer those among the whole population. 
As discussed in Section \ref{sec:cre}, Theorem \ref{thm:BB_simulCI} utilizes the fact that the set of treated units is a simple random sample from all experiment units, which avoids the overly conservative worst-case consideration in \citet{CDLM21quantile} and thus improves their approach. 
It is worth mentioning that ideas in Section \ref{sec:effect_treated} are not directly applicable here, because we cannot directly construct prediction intervals for treatment effects among unsampled units in the population. 
However, the method in Section \ref{sec:effect_treated} is still useful for constructing prediction intervals for treatment effects among experimental units; 
in this sense, the ideas in Sections \ref{sec:effect_treated} and \ref{sec:berger} 
can be 
complementary in Theorem \ref{thm:samp_fp}.} 
Below we give additional remarks related to Theorem \ref{thm:samp_fp}.

\begin{remark}\label{rmk:super_limit}
    {\rev When experimental units are i.i.d.~samples from a superpopulation, the superpopulation quantiles of individual treatment effects can be inferred in a way similar to Theorem \ref{thm:samp_fp}.
    Specifically, it can be viewed as a limiting case of Theorem \ref{thm:samp_fp} when the size $N$ of the finite population goes to infinity, under which the  multivariate Hypergeometric distribution as in \eqref{eq:delta_multiple} reduces to a Multinomial distribution.} 
\end{remark}

\begin{remark}\label{rmk:treated_generalize}
    {\rev 
    We can also extend Theorem \ref{thm:samp_fp} to settings where the set of treated (or control) units is a simple random sample from the population of interest. 
    For example, under the same setting as in Theorem \ref{thm:samp_fp}, when the experiment is a CRE, 
    both the treated and control units are also from a simple random sampling of the target population. 
    In this case, 
    we can also use sample effect quantiles among treated (or control) units to construct confidence intervals for population effect quantiles, by replacing $n$ by $n_t$ (or $n_c$), and $\tau_{(k'_j)}$ by $\tau_{t(k'_j)}$ (or $\tau_{c(k'_j)}$), respectively, in Theorem \ref{thm:samp_fp}.
    This essentially generalizes Theorem \ref{thm:BB_simulCI}, which corresponds to the special case of $N=n$. 
    }
\end{remark}

\begin{remark}\label{rmk:match_treated_iid}
    In matched observational studies, it is often assumed that the treated units in the matched study are from i.i.d.~samples from a superpopulation of units under treatment; see, e.g., \citet{Rosenbaum1987}. 
    Note that, in general, all units in a matched study are not i.i.d.~samples from the superpopulation of treated units. 
    Specifically, even when matching is exact, the matched control units can still have different unmeasured confounding from the treated units. 
    Consequently, the population of all matched units will generally depend on the distribution of unmeasured confounding in treated and control population as well as the number of matched controls for the treated units. 
    Therefore, it is more desirable to infer treatment effects among the treated units using, for example, Theorem \ref{thm:treated_sen}, and then generalize those to the population of treated units; the latter step is discussed in the previous Remark \ref{rmk:treated_generalize}. 
    For clarity and conciseness, we present a corollary for inferring population distributions of individual treatment effects among treated units in the supplementary material.
\end{remark}




\secadj
\section{Illustration}\label{sec:illustration}

\subsection{Simulation and practical suggestions}\label{sec:simu_sugg}

For conciseness, we relegate simulation studies to the supplementary material and summarize the main findings here. 
First, confidence intervals in Theorem \ref{thm:sing_interval_BB} are robust to a wide range of the tuning parameter $\gamma$. 
Second, 
as discussed in Sections \ref{sec:comp} and \ref{sec:num_illu_improv}, 
both improved methods in Sections \ref{sec:effect_treated} and \ref
{sec:berger}
can be preferable depending on the true data generating process, 
in particular the tail behavior of the two potential outcomes.


Below we give some practical suggestions. 
First, without prior knowledge, we suggest using the combination of prediction intervals for treatment effects among treated and control units as in Theorem \ref{thm:comb_all}
to infer treatment effects among all experimental units, 
since it requires no tuning and provides simultaneous confidence intervals for all effect quantiles.
Second, when analyzing matched observational studies, 
we suggest using Theorem \ref{thm:treated_sen} to infer treatment effects among treated units, which are often the target of interest.  
Third, in sampling-based randomized experiments, we suggest choosing the $k_j'$s in almost the same way as that in Remark \ref{rmk:choice_k_prime} with $\gamma = 1/2$. 
Intuitively,
we evenly split the allowed error rate between the following two steps: 
(i) using sample effect quantiles among experimental units to construct $1-\alpha/2$ simultaneous ``confidence'' intervals for population effect quantiles and (ii) using, for example, Theorem \ref{thm:comb_all} to construct $1-\alpha/2$ simultaneous prediction intervals for the sample effect quantiles. 


\subsection{Evaluating the effectiveness of professional development}\label{sec:teacher}

To demonstrate the benefits from the proposed methods in Section \ref{sec:cre}, we re-analyze the education experiment in \cite{CDLM21quantile},
in which a group of fourth grade teachers were randomly assigned to treatment or control and the treated teachers would participate in professional development courses focusing on the teaching of electric circuits \citep{HSMHSD10}.
The outcome of interest is the gain score based on tests before and after the courses. 
We preprocess the data in the same way as \citet{CDLM21quantile} 
and analyze the data as a completely randomized experiment with $164$ treated teachers and $69$ control teachers. 
 
\begin{figure}[htbp]
    \centering
    \begin{minipage}{.33\textwidth}
        \centering
        \includegraphics[width=0.8\linewidth]{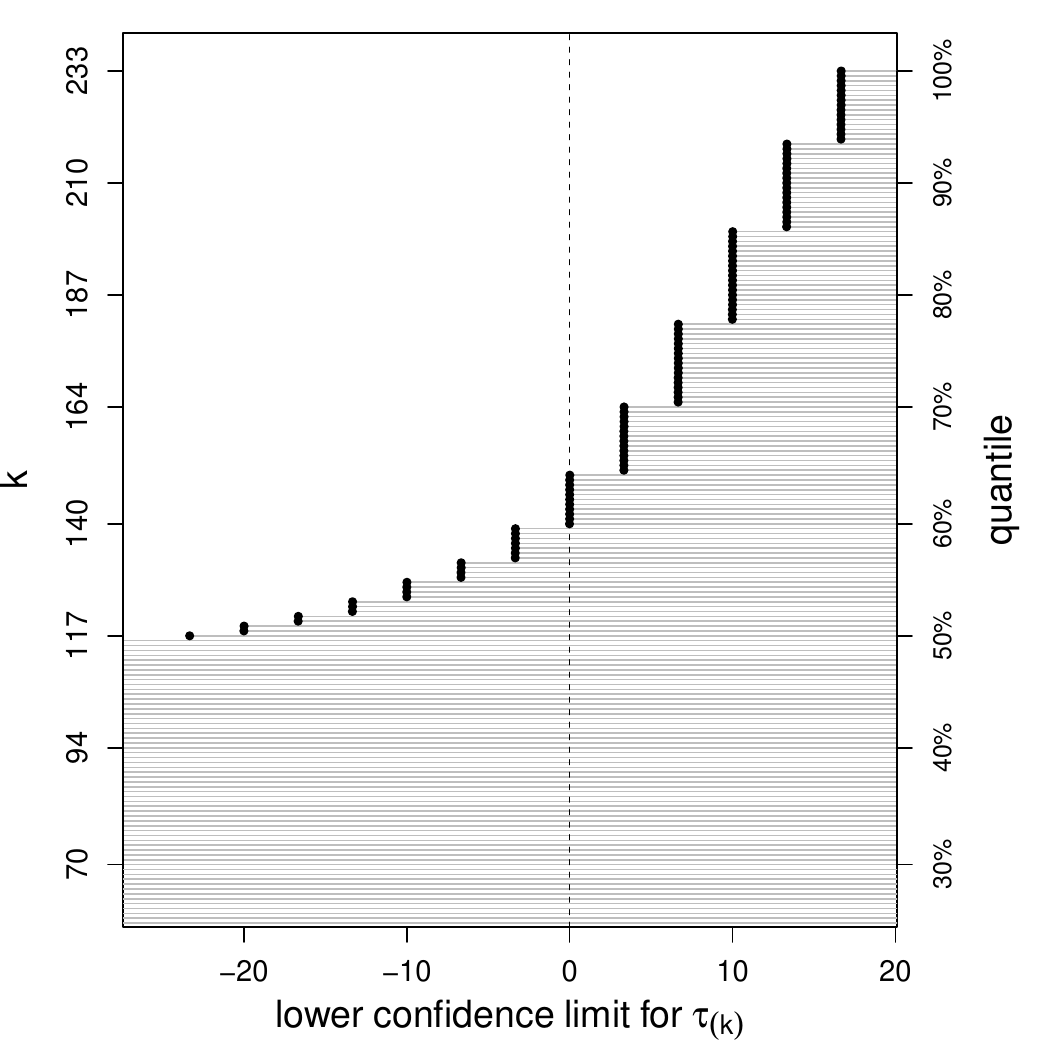}
        \subcaption{\citeauthor{CDLM21quantile}}
        \label{subfig:realdata_M0}
    \end{minipage}%
    \begin{minipage}{0.33\textwidth}
        \centering
        \includegraphics[width=0.8\linewidth]{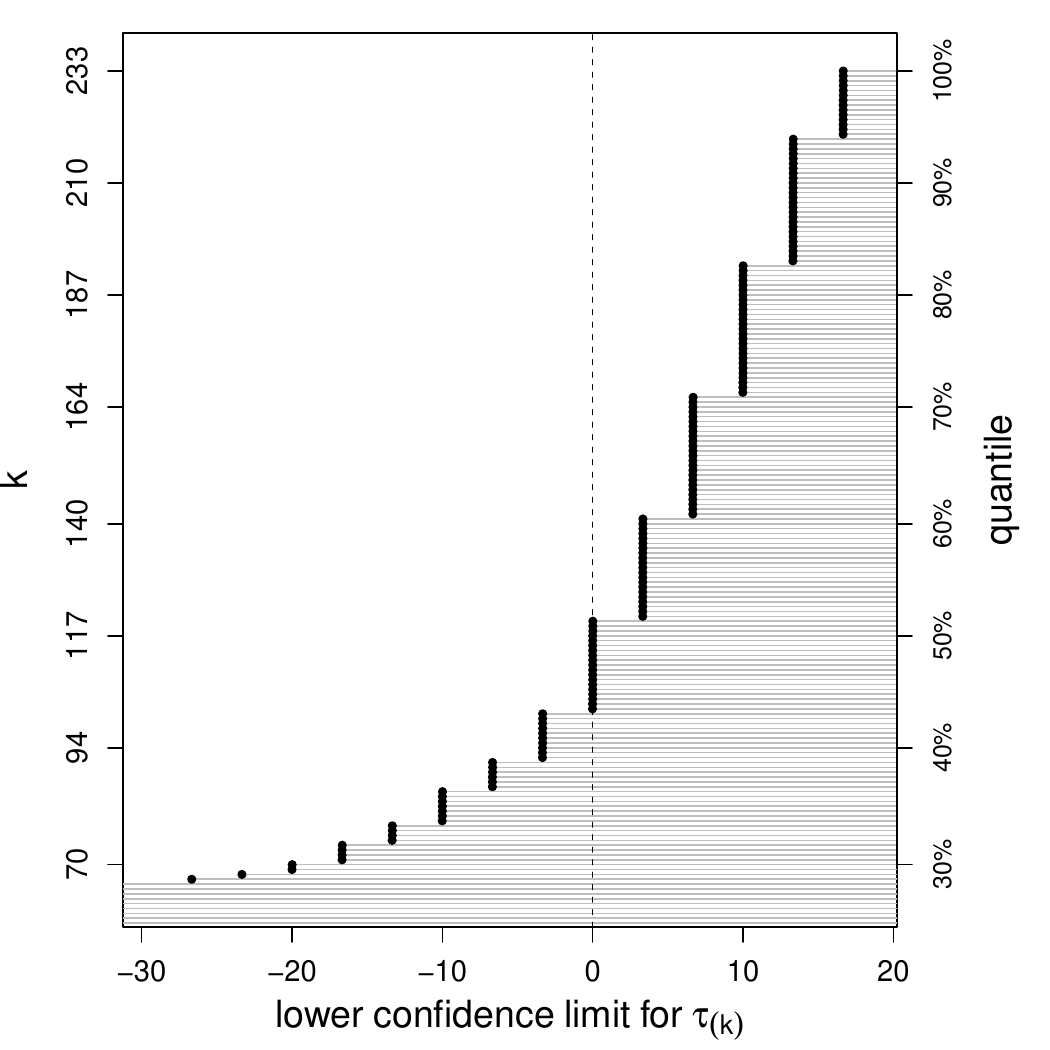}
        \subcaption{Section \ref{sec:effect_treated}}
        \label{subfig:realdata_M1}
    \end{minipage}%
    \begin{minipage}{0.33\textwidth}
        \centering
        \includegraphics[width=0.8\linewidth]{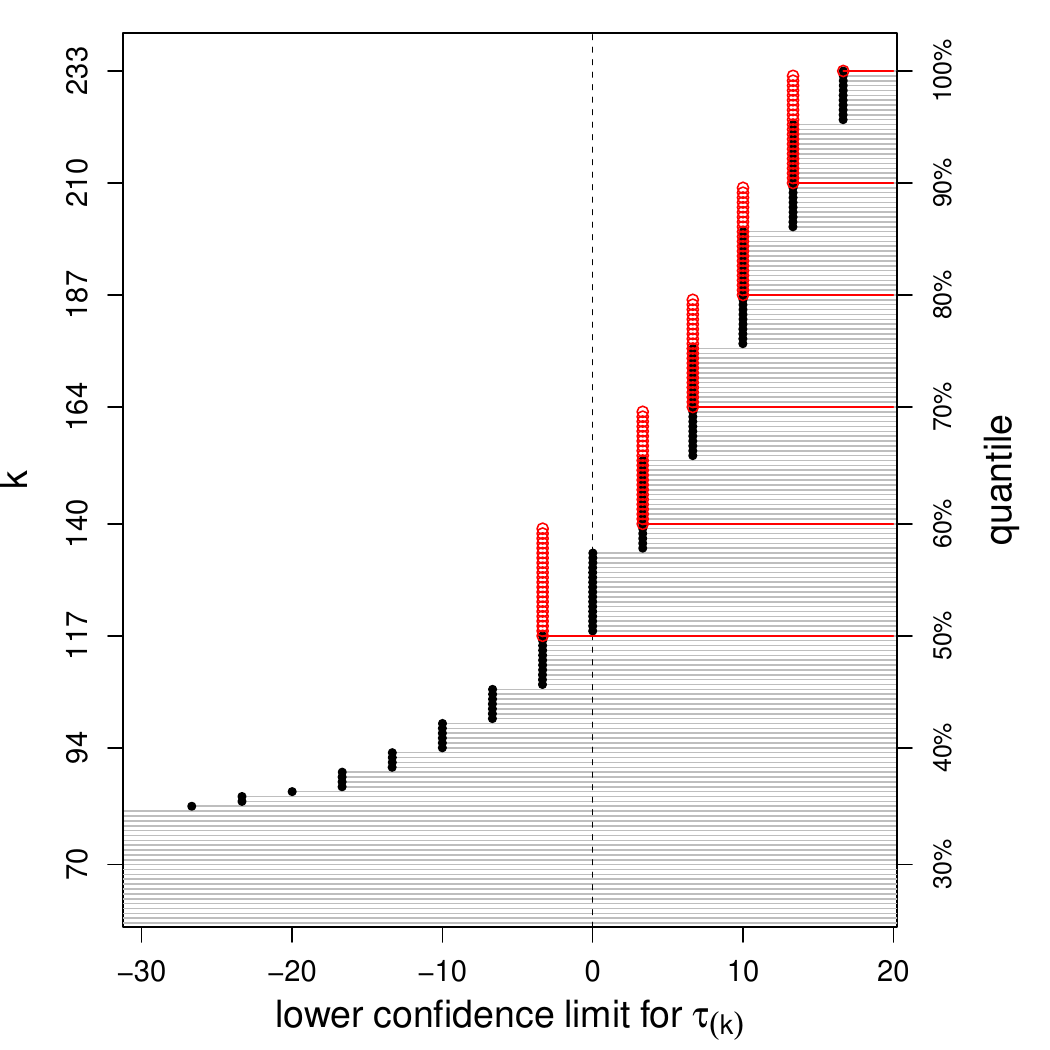}
        \subcaption{Section \ref{sec:berger}}
        \label{subfig:realdata_M2}
    \end{minipage}
    \caption{90\% confidence intervals for 
    treatment effect quantiles among all 233 teachers in the experiment.
    (a) and (b) show the simultaneous lower confidence limits using the original \cite{CDLM21quantile}'s method and our improved method from Section \ref{sec:effect_treated}. 
    (c) shows the individual lower confidence limits (black points) and simultaneous lower confidence limits (red points) using our improved method from Section \ref{sec:berger}.
    The uninformative lower confidence limits of $-\infty$ for lower quantiles are omitted from the plots.}
    \label{fig:realdata}
\end{figure}

We compare three methods: the original method in \citet{CDLM21quantile} 
and our proposed methods in Section \ref{sec:effect_treated} (particularly the one described in Theorem \ref{thm:comb_all}) and Section 3.3 (particularly the one described in Remark \ref{rmk:BB_comb}), where the latter two are invariant to treatment labeling. 
Following \citet{CDLM21quantile}, we use the Stephenson rank sum statistic with $s=6$ for all of these methods. 
Figure \ref{fig:realdata}(a)--(c) show the $90\%$ lower confidence limits for all effect quantiles $\tau_{(k)}$s using these three methods. 
The confidence intervals in (a) and (b) are both simultaneously valid, 
while the confidence intervals in (c) are either individually valid (shown in black color) or simultaneously valid (shown in red color). 
Specifically, we construct simultaneous confidence intervals for the $100\%, 90\%, \ldots, 50\%$ quantiles of individual effects in (c), and use that to construct a confidence band for all the effect quantiles as described in Remark \ref{rmk:bb_band}. 
{\rev
As a side remark, we also provide two-sided confidence intervals for all effect quantiles 
in the supplementary material; see the discussion at the end of Section \ref{sec:framework}. 
}

From Figure \ref{fig:realdata}, our newly proposed methods largely improve upon the original \citet{CDLM21quantile}'s approach and provide much more informative inference for treatment effect quantiles. 
For example, the confidence intervals for $\tau_{(k)}$s are not informative for $k \le 66$ and $k \le 80$ under our two improved methods, respectively, as opposed to $k \le 116$ under the original method. 
In addition, the lower confidence limits for the proportion of teachers with positive treatment effects, $n(0)/n$, are, respectively, $37.8\%$, $52.36\%$ and $46.4\%$, under the original and our improved methods. 
Here we choose the tuning parameter $\gamma = 0.5$ for 
the intervals in Figure \ref{fig:realdata}(c). 
We have also tried $\gamma = 0.3$ and $0.7$, and the results are similar.
This robustness against the choice of $\gamma$ agrees with our simulation results; see Section \ref{sec:simu_sugg}.


\begin{figure*}[htbp]
    \centering
    \begin{subfigure}[b]{0.32\textwidth}
        \centering
        \renewcommand\thesubfigure{a}
         \includegraphics[width=0.8\textwidth]{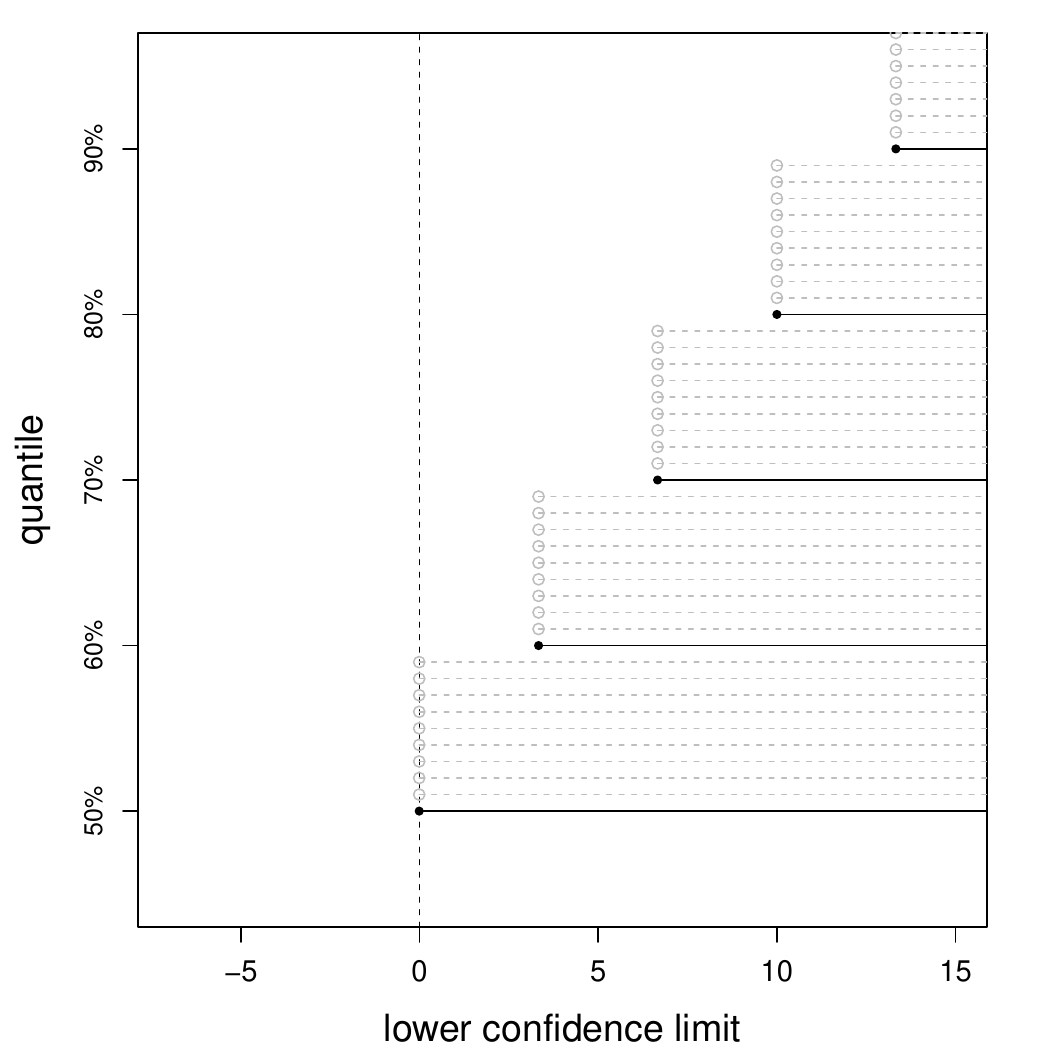}
        \caption{\small\centering $N = 500$}  
        \label{}
    \end{subfigure}
    \hfill
    \begin{subfigure}[b]{0.32\textwidth}    \addtocounter{subfigure}{-1}
    \renewcommand\thesubfigure{b}
        \centering 
        \includegraphics[width=0.8\textwidth]{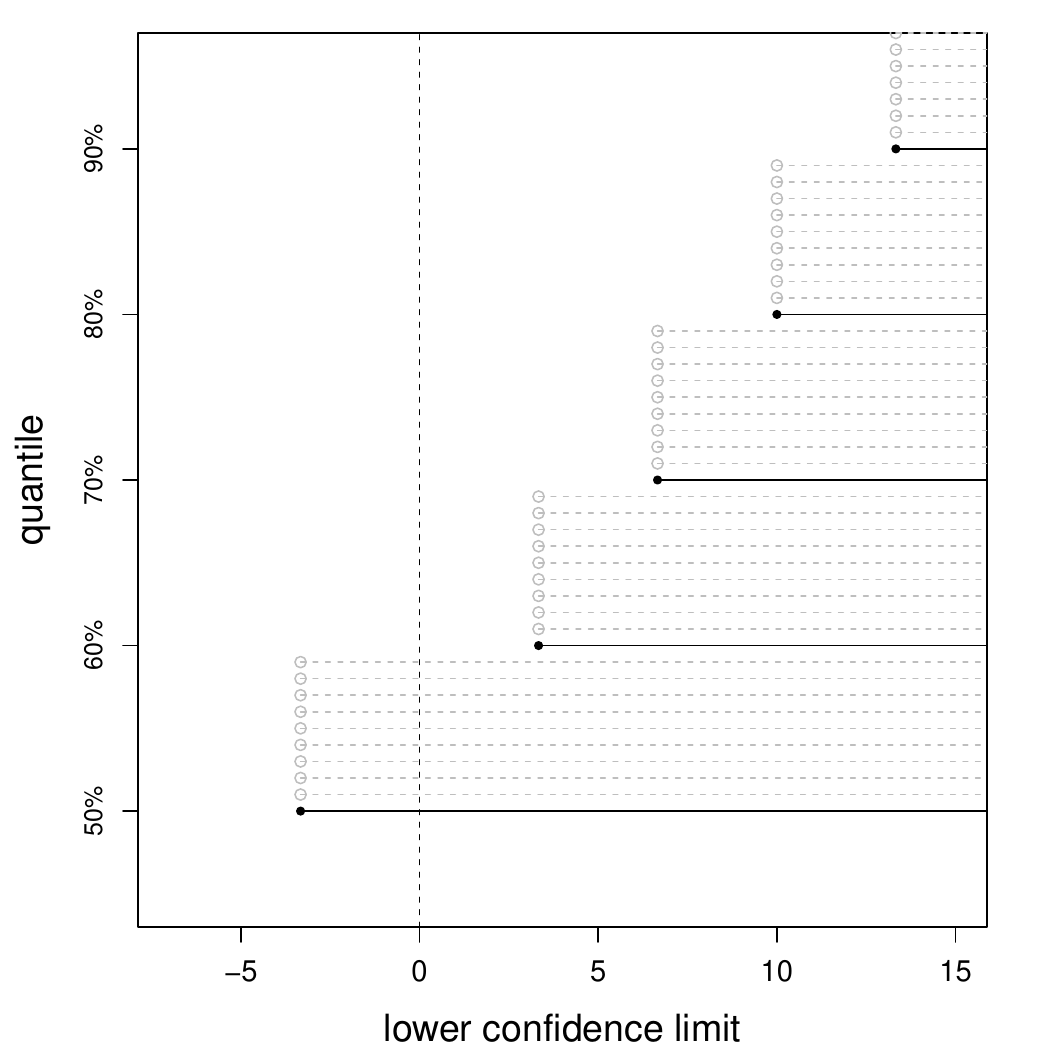}
        \caption{\small\centering $N = 1000$}    
        \label{}
    \end{subfigure}
   \hfill
    \begin{subfigure}[b]{0.32\textwidth}  \addtocounter{subfigure}{-1}   \renewcommand\thesubfigure{c}
        \centering 
        \includegraphics[width=0.8\textwidth]{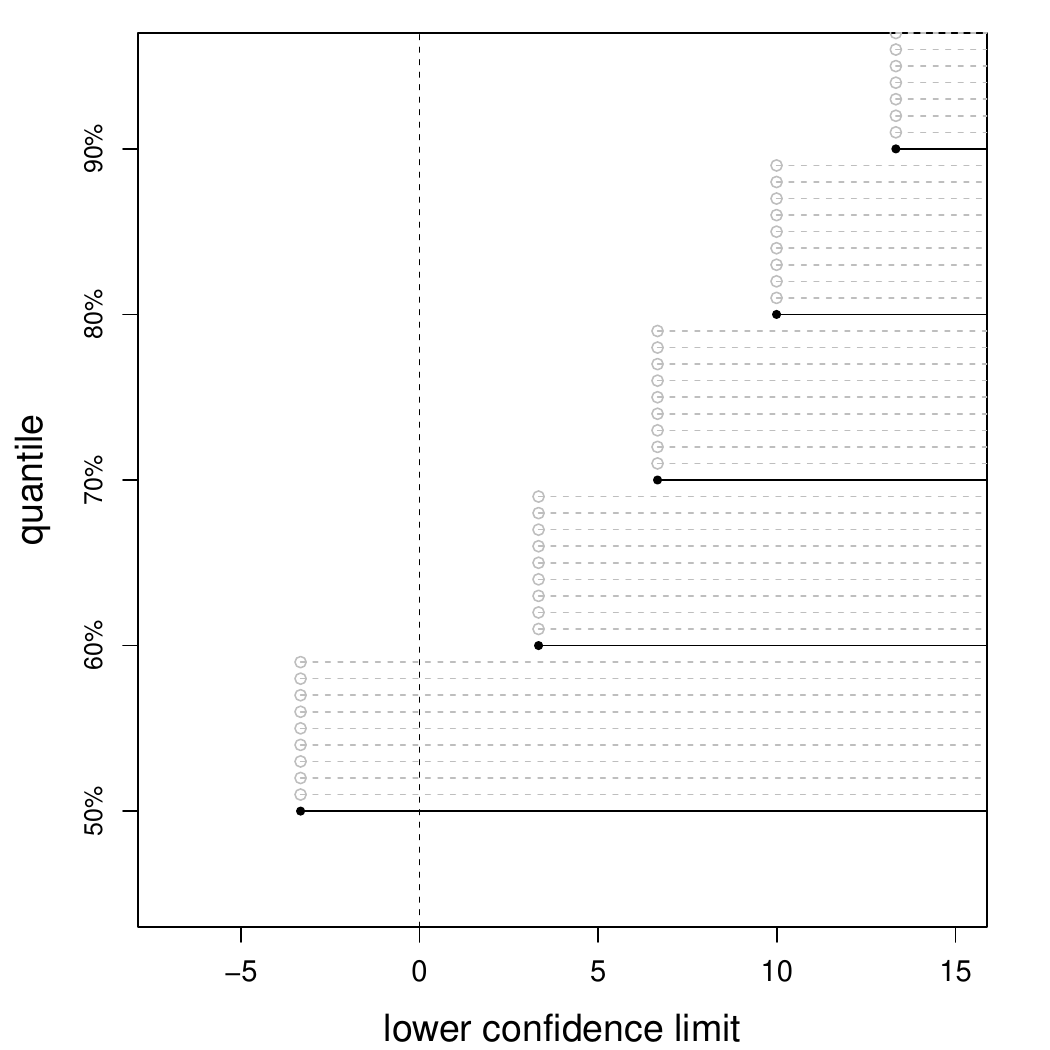}
        \caption{\small\centering $N = 5000$}  
        \label{}
    \end{subfigure}
    \caption[  ]
    {\small  
    90\% simultaneous confidence bands for population (of size $N$)
treatment effect quantiles. 
} 
    \label{fig: teacher_population_CIs}
\end{figure*}

We next assume that the $233$ teachers in the experiment are simple random samples from a larger finite population with $N=500$,  $1000$ or $5000$. 
Note that the inference for $N=5000$ is very close to that under i.i.d.\ sampling from a superpopulation; see 
Remark \ref{rmk:super_limit}. 
Figure \ref{fig: teacher_population_CIs} shows the $90\%$ simultaneous confidence bands for population quantile functions of individual treatment effects using Theorem \ref{thm:samp_fp}, based on the $50\%, 60\%, \ldots, 90\%$ quantiles. 
Not surprisingly, the confidence bands become less informative as the population size increases.

\subsection{Effect of smoking on the blood cadmium level}

We re-analyze the matched study in \citet{SL22quantile}, to demonstrate the benefit from our improved method in Section 4. 
The matched study contains 512 matched sets, each of which contains one treated unit (i.e., daily smoker) and two control units (i.e., nonsmokers). 
We are interested in the causal effect of smoking on the blood cadmium level ($ug/l$). 
Following \citet{SL22quantile}, we used the Wilcoxon rank sum statistic throughout this subsection. 

We first assume that matching is perfect so that the matched study essentially reduces to a SCRE, under which units are completely randomized within each matched set. 
Figure \ref{fig: scre}(a) and (b) show the $90\%$ simultaneous confidence intervals for all treatment effect quantiles from the original approach in \citet{SL22quantile} using either the original data or the data with switched treatment labels and changed outcome signs. 
As discussed in \ref{sec:scre}, the top $512$ intervals in Figure \ref{fig: scre}(a) and the top $1024$ intervals in Figure \ref{fig: scre}(b) can be equivalently interpreted as the simultaneous prediction intervals for effect quantiles among treated and control units, respectively. 
Figure \ref{fig: scre}(c) shows the $90\%$ confidence intervals from our improved method, which essentially combines those top intervals from (a) and (b) but with confidence level $95\%$. 
Obviously, the improved method provides more informative confidence intervals. 

\begin{figure}[htbp]
    \centering
    \begin{minipage}{0.33\textwidth}
        \centering
        \includegraphics[width=0.8\linewidth]{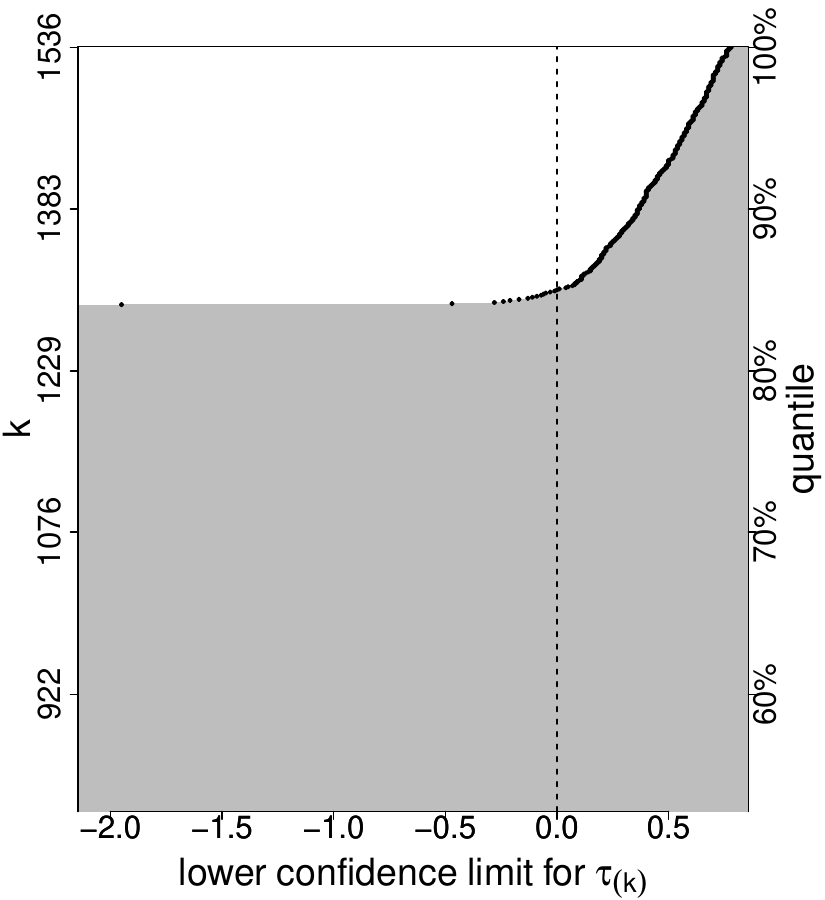}
        \subcaption{\textsf{No switching}}
        \label{subfig: scre a}
    \end{minipage}%
    \begin{minipage}{0.33\textwidth}
        \centering
        \includegraphics[width=0.8\linewidth]{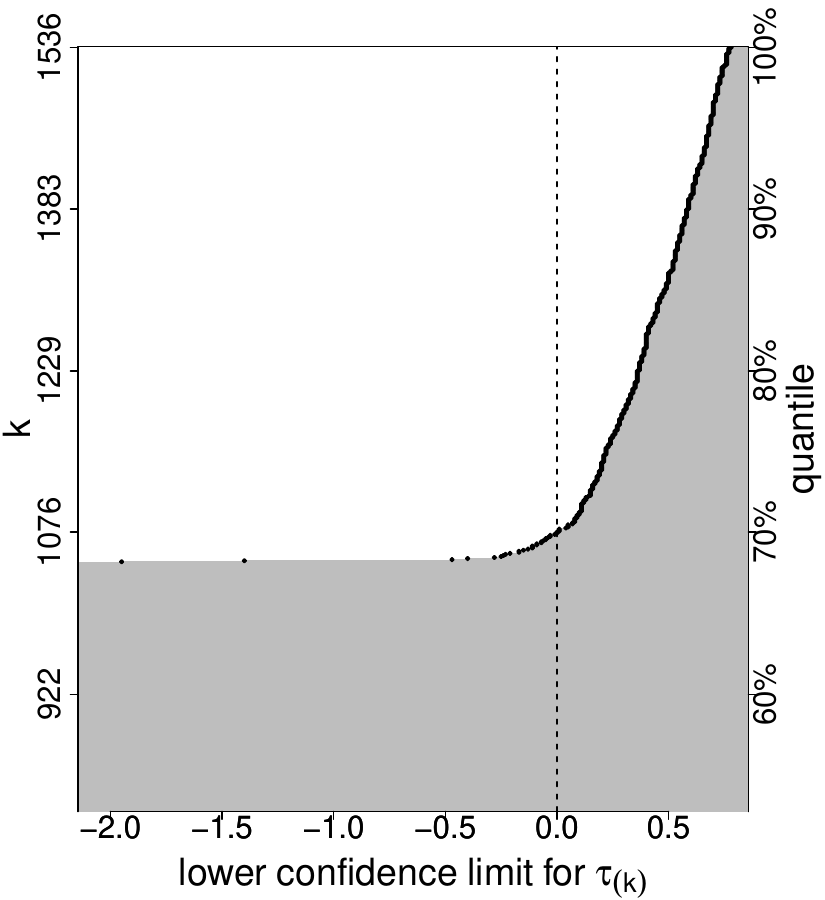}
        \subcaption{\textsf{Switching labels}}
        \label{subfig: scre b}
    \end{minipage}%
    \begin{minipage}{0.33\textwidth}
        \centering
        \includegraphics[width=0.8\linewidth]{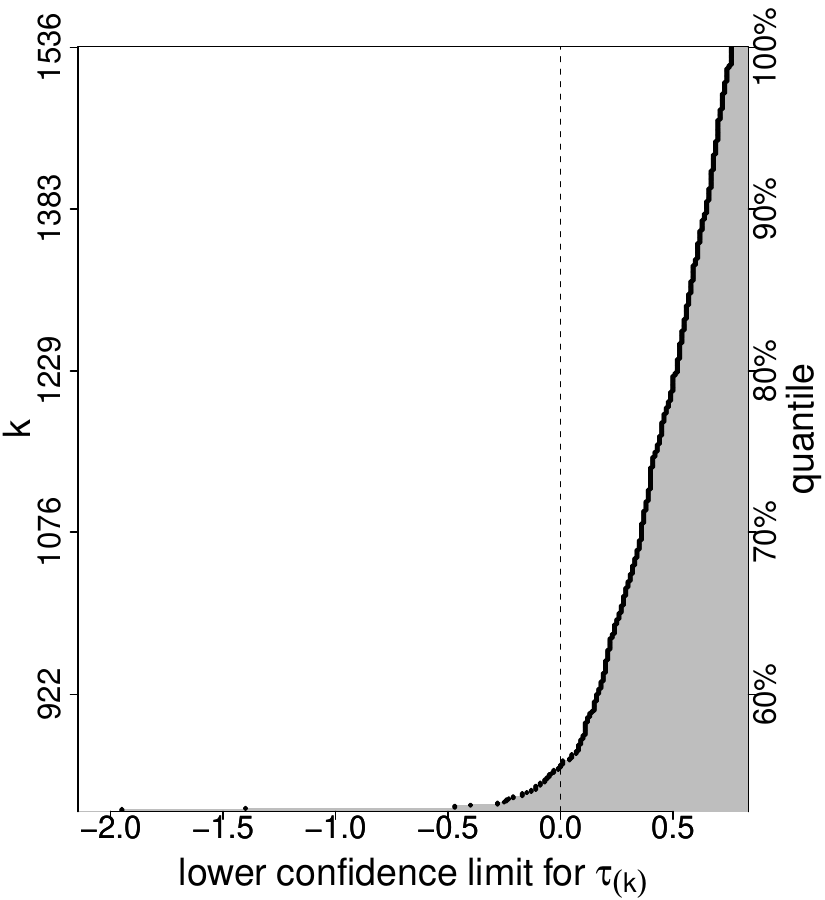}
        \subcaption{\textsf{Improved method}}
       \label{subfig: scre c}
    \end{minipage}
    \caption{90\% simultaneous confidence intervals for quantiles of individual treatment effects. 
    Uninformative intervals of $(-\infty, +\infty)$ for quantiles of lower ranks are omitted from the figure.
    }
    \label{fig: scre}
\end{figure}
Next, we conduct sensitivity analysis for treatment effect quantiles when the bound $\Gamma$ on confounding strength takes various values. 
Figure \ref{fig: sensitivity}(a) shows the $90\%$ simultaneous confidence intervals for effect quantiles of all units under various $\Gamma$ values using the original approach in \citet{SL22quantile} with switched treatment labels and changed outcome signs, and (b) shows the corresponding confidence intervals using our improved approach. 
We further consider treatment effects among only treated units, which are often the inferential target for a matched observational study. 
Figure \ref{fig: sensitivity}(c) shows the $90\%$ simultaneous prediction intervals for all effect quantiles among treated units. 
The values of $\Gamma$ for these figures, $1.0,$ $1.3$, $2.2$, $4.0$, $8.3$ and $38.4$, are the largest $\Gamma$ values such that the $90\%$ prediction intervals for the $55\%, 60\%, 70\%, 80\%, 90\%$ and $100\%$ quantiles of individual effects among treated units do not cover zero, respectively. 
From Figure \ref{fig: sensitivity}(c), smoking increases blood cadmium level for a significant proportion of smokers, and this conclusion is robust to moderate unmeasured confounding. 
For example, when the strength of unmeasured confounding is bounded by $2.2$,
we are $90\%$ confident that smoking increases the blood cadmium level for at least $30\%$ of the smokers in the study. 

\begin{figure}[htbp]
    \centering
    \begin{minipage}{0.33\textwidth}
        \centering
        \includegraphics[width=0.8\linewidth]{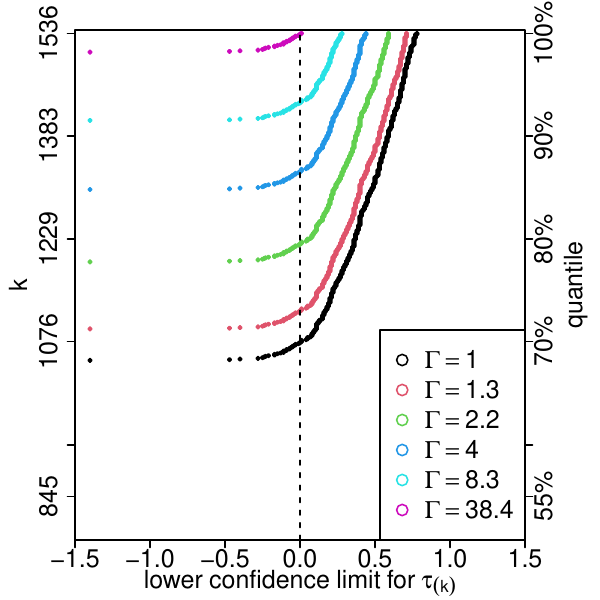}
        \subcaption{\textsf{ 
        Original method
        }}
        \label{subfig: scre sen a}
    \end{minipage}%
    \begin{minipage}{0.33\textwidth}
        \centering
        \includegraphics[width=0.8\linewidth]{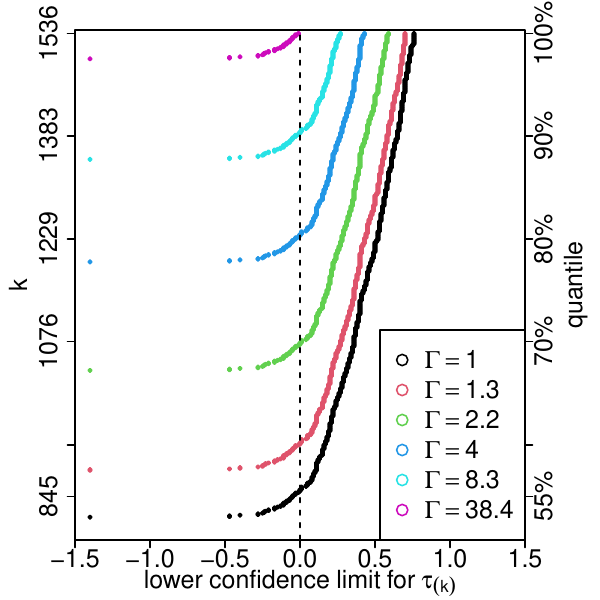}
        \subcaption{\textsf{
        Improved method
        }}
        \label{subfig: scre sen b}
    \end{minipage}%
    \begin{minipage}{0.33\textwidth}
        \centering
        \includegraphics[width=0.8\linewidth]{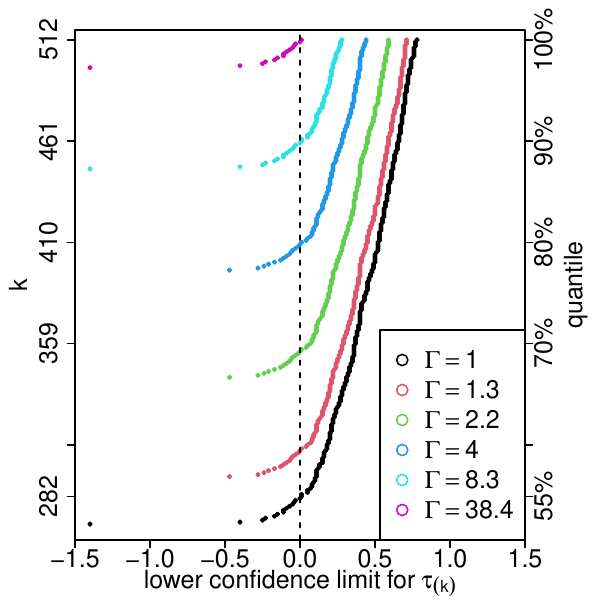}
        \subcaption{\textsf{
        Prediction intervals
        }}
       \label{subfig: scre sen c}
    \end{minipage}
    \caption{90\% lower confidence limits for effect
quantiles among all or treated units
    under various sensitivity models indexed by different $\Gamma$ values from:
    (a) the original approach in \citet{SL22quantile} with switched treatment labels and changed outcome signs, (b)  the improved approach for all units, and (c) 
    the improved method for treated units. 
    }
    \label{fig: sensitivity}
\end{figure}

\secadj
\section{Conclusion and discussion}\label{sec:conclusion}

In this paper we studied inference for distributions and quantiles of individual treatment effects in various randomized experiments as well as quasi-experiments from matching. 
The key ideas can be summarized as follows. 
First, the recent randomization-based approaches proposed by \citet{CDLM21quantile} and \citet{SL22quantile} can be interpreted as inferring treatment effects among only treated or control units. 
They can be 
combined 
to provide often more informative inference for treatment effects among all units. 
Second, under random sampling, we can use sample quantiles to construct confidence intervals for effects quantiles among larger populations. 
For example, we can use effect quantiles among treated units to infer effect quantiles among all experiment units as studied in Section \ref{sec:berger}; we can also use effect quantiles among experimental units to infer effect quantiles for larger populations 
as studied in Section \ref{sec:samp_based}.
Once we have prediction intervals for sample quantiles, we can have a two-step procedure to construct confidence intervals for population quantiles.

The proposed methods can significantly enhance existing randomization-based inference for treatment effect quantiles.
It is natural to ask whether the improved methods are 
accurate. 
For example, given that the true effect quantiles are generally only partially identified,
we can investigate 
whether the confidence bounds for treatment effect quantiles converge to their sharp bounds as the sample size goes to infinity. 
We leave this as future work. 

\secadj
\section*{Supplementary Material}
The supplementary material includes simulation studies, proofs of all theorems, discussion for binary outcomes and experiments with units sampled from superpopulation, and additional techinical details. 

\secadj
\section*{Disclosure Statement}
The authors report there are no competing interests to declare.


\singlespacing
\setlength{\bibsep}{0pt}
\bibliographystyle{plainnat}
\bibliography{reference}

 \newpage
 


\title{\bf 
 \Large
Supplementary Material to ``Enhanced inference for distributions and quantiles of individual treatment effects in various experiments''
}
\author{
}
\date{}
\maketitle

\setcounter{equation}{0}
\setcounter{section}{0}
\setcounter{figure}{0}
\setcounter{example}{0}
\setcounter{proposition}{0}
\setcounter{corollary}{0}
\setcounter{theorem}{0}
\setcounter{lemma}{0}
\setcounter{table}{0}
\setcounter{condition}{0}

\makeatletter
\renewcommand{\l@section}{\@dottedtocline{1}{0em}{2.3em}}
\makeatother

\renewcommand {\theproposition} {A\arabic{proposition}}
\renewcommand {\theexample} {A\arabic{example}}
\renewcommand {\thefigure} {A\arabic{figure}}
\renewcommand {\thetable} {A\arabic{table}}
\renewcommand {\theequation} {A\arabic{equation}}
\renewcommand {\thelemma} {A\arabic{lemma}}
\renewcommand {\thesection} {A\arabic{section}}
\renewcommand {\thetheorem} {A\arabic{theorem}}
\renewcommand {\thecorollary} {A\arabic{corollary}}
\renewcommand {\thecondition} {A\arabic{condition}}
\renewcommand {\thermk} {A\arabic{rmk}}

\renewcommand {\thepage} {A\arabic{page}}
\setcounter{page}{1}

\doublespacing


{\rev 
\section{Inference for average treatment effects on a binary outcome}
\label{supp: inference for attributable effects}

Below we extend the approach based on distribution free rank statistics, as in \citet{CDLM21quantile} and the main paper, to infer average treatment effects when the outcomes are binary taking values in $\{0, 1\}$. 
We consider the CRE and adopt the notation from the main paper. 
We now focus on the following hypotheses on total treatment effects among treated units: 
\begin{align}\label{eq:H_ave_treat}
    \bar{H}^\treat_{c}: \sum_{i}^n Z_i \tau_i = c 
    \Longleftrightarrow \bs{\tau} \in \bar{\mathcal{H}}^\treat_{c}, 
    \qquad (c \in \{-n_\treat, -n_\treat+1, \ldots, n_\treat\})
\end{align}
where $\bar{\mathcal{H}}^\treat_{c}$ denotes the set of all possible values of the individual treatment effect vector under $\bar{H}^\treat_{c}$.
The hypothesis in \eqref{eq:H_ave_treat} states that the total treatment effect among treated units is $c$, or equivalently the average treatment effect among treated units is $c/n_\treat$. 
Once we have a valid test for \eqref{eq:H_ave_treat}, by the standard test inversion, we can then construct prediction intervals for the average treatment effect among treated. By switching treatment labels and redefining the outcome as 1 minus the original outcome, we can also infer total or average treatment effects among control units. 
In addition, note that 
\begin{align}\label{eq:attribute_effect_Y0}
    \sum_{i=1}^n Z_i \tau_i = \sum_{i=1}^n Z_i Y_i(1) - \sum_{i=1}^n Z_i Y_i(0) = \sum_{i=1}^n Z_i Y_i - \sum_{i=1}^n Z_i Y_i(0), 
\end{align}
Thus, as discussed in \citet{Rosenbaum:2001},
given the observed data, 
the total effect among treated units can only take integer values between $\sum_{i=1}^n Z_i Y_i - n_\treat$ and $\sum_{i=1}^n Z_i Y_i$. 

By a similar logic as Theorem \ref{thm:treated}, we can construct a valid $p$-value for the null hypothesis $\bar{H}^\treat_{c}$ in \eqref{eq:H_ave_treat} based on a distribution free test statistic $t(\cdot, \cdot)$, by considering the minimum value of $t(\bs{Z}, \bs{Y}(0))$ under $\bar{H}^\treat_{c}$. 
We summarize the results in the theorem below. 

\begin{theorem}\label{thm:pval_ave}
    Consider a CRE with a binary outcome, and any rank score statistic $t(\cdot, \cdot)$ in \eqref{eq:rank_score}. 
    For any integer $c\in [-n_\treat, n_\treat]$, 
    the following is a valid $p$-value for testing the null hypothesis $\bar{H}^\treat_{c}$ in \eqref{eq:H_ave_treat}: 
    \begin{align}\label{eq:p_ave_treat}
        \bar{p}_c^\treat = G\big( \min_{\bs{\delta} \in \bar{\mathcal{H}}^\treat_{c}} t(\bs{Z}, \bs{Y} -\bs{Z} \circ \bs{\delta} ) \big), 
        \quad 
        \quad 
        \text{where } G(x) = \Pr\{ t(\bs{Z}, \bs{y}) \ge x \} \text{ for any } \bs{y}\in \mathbb{R}^n.
    \end{align}
\end{theorem}

In \eqref{eq:p_ave_treat}, $G(\cdot)$ is defined the same as that in \eqref{eq:p_nkc}, 
whose value does not vary with $\bs{y}$ in its definition due to the distribution free property of the rank score statistic under the CRE. 
The remaining question is how to minimize the test statistic $t(\bs{Z}, \bs{Y} -\bs{Z} \circ \bs{\delta} )$ subject to the constraint that $\bs{\delta} \in \bar{\mathcal{H}}^\treat_{c}$. 
Below we focus on the Wilcoxon rank sum statistic, i.e., the rank sum statistic in \eqref{eq:rank_score} with identity rank transformation, which is intuitive in the case of binary outcomes. 
It turns out that the corresponding optimization of the test statistic in \eqref{eq:p_ave_treat} has a closed-form solution, as explained in detail below.


By the equivalence between the usual Wilcoxon rank sum test and Mann–Whitney U test, the Wilcoxon rank sum statistic has the following equivalent form:
\begin{align}\label{eq:MWU}
    t(\bs{Z}, \bs{Y} -\bs{Z} \circ \bs{\delta} ) 
    & = 
    \sum_{i=1}^n  Z_i \left\{\sum_{j=1}^n (1- Z_j)\delta_{ij}(Y_i-Z_i\delta_i, Y_j - Z_j \delta_j)\right\} + \frac{n_\treat(n_\treat + 1)}{2}
    \nonumber
    \\
    & = 
    \sum_{i=1}^n  Z_i \left\{\sum_{j=1}^n (1- Z_j)\delta_{ij}(Y_i-\delta_i, Y_j)\right\} + \frac{n_\treat(n_\treat + 1)}{2},
\end{align}
where
\begin{align}\label{eq:delta_ij}
    \delta_{ij}(Y_i-\delta_i, Y_j)
    & = \I\{ Y_i-\delta_i > Y_j \} + \I\{ Y_i-\delta_i = Y_j \} \I\{i > j\}.
\end{align}
In \eqref{eq:delta_ij}, we break ties based on their indices; this is essentially random tie-breaking, since we will randomly shuffle the order of units before analysis. 
Under the constraint that $\bs{\delta}\in \bar{\mathcal{H}}^\treat_{c}$, we have 
\begin{align*}
    \sum_{i=1}^n Z_i (Y_i-\delta_i) = \sum_{i=1}^n Z_i Y_i -  \sum_{i=1}^n Z_i \delta_i = \sum_{i=1}^n Z_i Y_i - c.
\end{align*}
Moreover, by the fact that the outcomes are binary, we know $Y_i-\delta_i$, as a possible control potential outcome for unit $i$, must take values in $\{0,1\}$. 
Therefore, 
by substituting $Y_i-\delta_i$ with $x_i$, 
the optimization for $\min_{\bs{\delta} \in \bar{\mathcal{H}}^\treat_{c}} t(\bs{Z}, \bs{Y} -\bs{Z} \circ \bs{\delta} )$ in \eqref{eq:p_ave_treat} becomes equivalent to 
\begin{align}\label{eq:binary_opt}
    \min \quad  & \sum_{i:Z_i=1} \left\{\sum_{j=1}^n (1- Z_j)\delta_{ij}(x_i, Y_j)\right\} + \frac{n_\treat(n_\treat + 1)}{2}\\
    \text{subject to} \quad &  
    \sum_{i:Z_i=1} x_i = \sum_{i=1}^n Z_i Y_i - c, 
    \nonumber
    \\
    & x_i \in \{0,1\} \text{ for all } i \text{ such that } Z_i = 1.
    \nonumber
\end{align}
Note that 
\begin{align*}
    & \quad \ \sum_{i:Z_i=1} \left\{\sum_{j=1}^n (1- Z_j)\delta_{ij}(x_i, Y_j)\right\}
    \\
    & = \sum_{i:Z_i=1} \left\{\sum_{j=1}^n (1- Z_j)\delta_{ij}(0, Y_j)\right\}
    + 
    \sum_{i:Z_i=1}  x_i \left\{\sum_{j=1}^n (1- Z_j)\delta_{ij}(1, Y_j) - \sum_{j=1}^n (1- Z_j)\delta_{ij}(0, Y_j) \right\}\\
    & = 
    \sum_{i:Z_i=1} \left\{\sum_{j=1}^n (1- Z_j)\delta_{ij}(0, Y_j)\right\}
    +  
    \sum_{i:Z_i=1}  x_i \Delta_i,
\end{align*}
where $\Delta_i \equiv \sum_{j=1}^n (1- Z_j)\delta_{ij}(1, Y_j) - \sum_{j=1}^n (1- Z_j)\delta_{ij}(0, Y_j)$ for all $i$ such that $Z_i=1$. 
It is not difficult to see that the solution to \eqref{eq:binary_opt} is achieved when $x_i$ equals $1$ for the smallest $(\sum_{i=1}^n Z_i Y_i - c)$ elements of $\{\Delta_i: Z_i=1, 1\le i \le n\}$ and is $0$ otherwise. 
Let $\Delta_{(1)} \le \Delta_{(2)} \le \ldots \le \Delta_{(n)}$ be the sorted values of the $\Delta_i$s for treated units. 
The minimum value of the test statistic from the optimization in \eqref{eq:p_ave_treat} and \eqref{eq:binary_opt} then has the following form: 
\begin{align}\label{eq:min_test_binary}
    \min_{\bs{\delta} \in \bar{\mathcal{H}}^\treat_{c}} t(\bs{Z}, \bs{Y} -\bs{Z} \circ \bs{\delta} )
    & = \sum_{i:Z_i=1} \left\{\sum_{j=1}^n (1- Z_j)\delta_{ij}(0, Y_j)\right\}
    +  
    \sum_{i=1}^{\sum_{i=1}^n Z_i Y_i - c} \Delta_{(i)} + \frac{n_\treat(n_\treat + 1)}{2},
\end{align}
which can be easily calculated. 


From Theorem \ref{thm:pval_ave}, we can then invert the test to construct prediction intervals for the total effect among treated units. 
Specifically, a $1-\alpha$ prediction set for $\sum_{i=1}^n Z_i \tau_i$ is 
\begin{align}\label{eq:pred_set}
    \{ c \in \mathbb{Z}: \ \bar{p}_c^\treat > \alpha, \ \sum_{i=1}^n Z_i Y_i - n_\treat \le c \le \sum_{i=1}^n Z_i Y_i \}, 
\end{align}
where $\mathbb{Z}$ denotes the set of all the integers, and the second constraint in \eqref{eq:pred_set} follows from the discussion after \eqref{eq:attribute_effect_Y0} about the possible values of $\sum_{i=1}^n Z_i \tau_i$ given the observed data.  
Moreover, the $p$-value $\bar{p}_c^\treat$ is increasing in $c$, so that we can compute the prediction set in \eqref{eq:pred_set} efficiently. 
Specifically, through binary search, we can find the $1-\alpha$ lower prediction bound 
using at most $O(\log n_\treat)$ randomization tests. 
The monotonicity of the p-value $\bar{p}_c^\treat$ follows from the form of the minimum test statistic value in \eqref{eq:min_test_binary} and the fact that $\Delta_i \ge 0$ for any treated unit $i$, which can be verified by definition. 

The prediction set in \eqref{eq:pred_set} gives a lower prediction bound for the total effect among treated units. 
We can also construct upper prediction bounds for the total effect among treated units by applying the above procedure, but with outcomes redefined as one minus the original outcomes. 
We can then combine lower and upper prediction bounds, with a Bonferroni adjustment, to construct two-sided prediction bounds for the total effect among treated units.


Below we numerically compare our approach with 
the method in \citet{Rosenbaum:2001} and \citet{Rigdon:2015} for inferring the total treatment effect among the treated for a binary outcome. 
We use the education experiment in Section \ref{sec:teacher} and dichotomize the original outcome, 
assigning 0 to outcomes below 10 and 1 to those with values of 10 or greater.
Since our method depends on the random tie breaking, 
we reshuffle the order of units $100$ times and conduct inference based on each of the reshuffled data. 
The $90\%$ prediction interval for the total effect among treated using \citet{Rigdon:2015}'s approach is $[83, 119]$. 
The $90\%$ prediction intervals from our method vary across these 100 shuffled datasets, with an average length of about 67. 
Figure \ref{fig:Attributable_effect_CI_plot} shows the histograms of the prediction lower and upper bounds from our method, as well as the bounds from \citet{Rigdon:2015}'s approach. 
The intervals from our method have substantial randomness across different ordering of units and tend to be wider than those from \citet{Rigdon:2015}'s approach. 
Note that our method uses the random-tie breaking to ensure the distribution free property of the rank sum statistic. 
From Figure \ref{fig:Attributable_effect_CI_plot},  \citet{Rigdon:2015}'s approach is preferable in this context. 
Nevertheless, the $p$-value from our approach is monotone in the hypothesized total effect among treated units, which facilitates the search for lower (or upper) prediction bounds. Such a computational gain using the monotonicity of randomization $p$-values is similar to \citet{LD2016binary}'s improvement over another method in \citet{Rigdon:2015}, which is based on randomization tests with test statistics of the form $t(\bs{Z}, \bs{Y})$. 
It would be interesting to extend our method to cases with general discrete outcomes that have more than two levels, for which the potential computational benefit of our method could be more substantial; we leave this for future investigation. 
}


\begin{figure}[!htb]
    \centering
    \begin{subfigure}{.5\textwidth}
        \centering
       \includegraphics[width=1\linewidth]{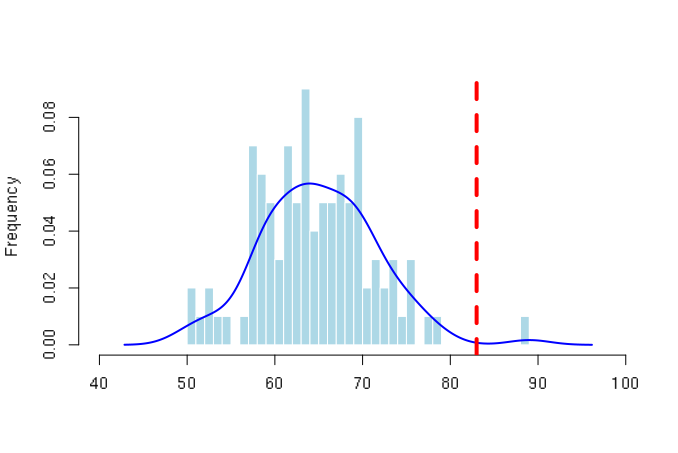}
       \caption{Lower prediction bound}
    \end{subfigure}%
    \begin{subfigure}{0.5\textwidth}
        \centering
       \includegraphics[width=1\linewidth]{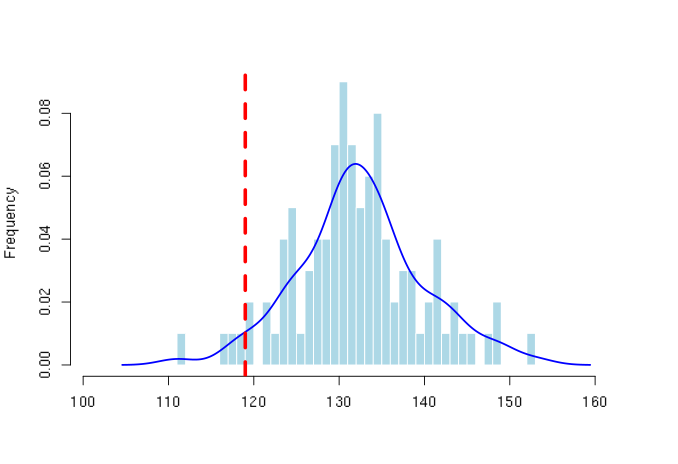}
    \caption{Upper prediction bound}
    \end{subfigure}
    \caption{Lower and upper prediction bounds of the $90\%$ prediction intervals for the total effect among treated units. 
    The histograms show the prediction bounds from our method across the $100$ reshuffled datasets. 
    The red dashed lines show the prediction bounds from the approach in \citet{Rigdon:2015}.
    }
    \label{fig:Attributable_effect_CI_plot}
\end{figure}

\section{Experiments with random sampling from a superpopulation}\label{sec:sp}

We consider now the superpopulation setting where the $n$ experimental units are independent and identically distributed (i.i.d.) samples from a superpopulation.  
Let $( Y^\sp(1), Y^{\sp}(0) )$ denote a representative sample from this superpopulation. 
We are interested in the superpopulation distribution of individual treatment effect $\tau^{\sp} \equiv Y^{\sp}(1) - Y^{\sp}(0)$. 
Let $F_{\sp}(c) \equiv \Pr(\tau^{\sp}\le c)$ and $F^{-1}_\sp (\cdot)$ denote the corresponding distribution and quantile functions.

We consider inference for superpopulation quantiles $F_{\sp}^{-1}(\beta_1) \le F_{\sp}^{-1}(\beta_2) \le \ldots \le F_{\sp}^{-1}(\beta_J)$ for some prespecified $0\le \beta_1 < \ldots < \beta_J \le 1$. 
Similar to finite population setting, we use the two-step procedure to construct simultaneous confidence intervals for them.  
Under i.i.d.\ sampling, the coverage probability of sample quantiles for population quantiles will involve Multinomial distributions.
For any $n\ge 1$ and $0 \leq \beta_1 < \ldots < \beta_J \leq 1$, 
define 
\begin{align}\label{eq:delta_M}
    \Delta_{\textup{M}}(k_{1:J}; n, \beta_{1:J})
    = 
    \Pr\Big( \bigcup_{j=1}^J \Big\{ \sum_{i=j}^J M_i > n - k_j \Big\} \Big),
\end{align}
where $(n-\sum_{j=1}^J M_j, M_1, \ldots, M_{J-1}, M_J) 
\sim \textup{MN}(\beta_1, \beta_2-\beta_1, \ldots, \beta_J - \beta_{J-1}, 1-\beta_J; n).$
Specifically, $\textup{MN}(a_0, a_1, \ldots, a_J; n)$ denotes the numbers of balls of colors $0$ to $J$ among in total $n$ balls, where each ball has color $j$ with probability $a_j$, independently across all balls. 

\begin{theorem}\label{thm:sp}
Consider an experiment where the units are i.i.d. samples from a superpopulation, 
and suppose that the population effect quantiles $F_{\sp}^{-1}(\beta_j)$s are of interest, where $0\le \beta_1 < \beta_2 < \ldots < \beta_J \le 1$ and $J\ge 1$. 
For any $\alpha \in (0,1)$ and $0\le k_1' \le k_2' \le \ldots \le k_J' \le n$, 
define 
$\Delta_{\textup{M}}(k_{1:J}'; n, \beta_{1:J})$ as in \eqref{eq:delta_M},
and let 
$\mathcal{I}^{\alpha'}_{(k_j')}$s be simultaneous $1-\alpha'$ one-sided prediction intervals for sample effect quantiles $\tau_{(k_j')}$s  using Theorem  \ref{thm:comb_all} or \ref{thm:treated_scre}, depending on the treatment assignment mechanism, 
where 
$\alpha' = \alpha - \Delta_{\textup{M}}(k_{1:J}'; n, \beta_{1:J})>0$.
Then the $\mathcal{I}^{\alpha'}_{(k_j')}$s 
are simultaneous confidence intervals for the quantiles $F_{\sp}^{-1}(\beta_j)$s, in the sense that 
$
    \Pr\{F_{\sp}^{-1}(\beta_j) \in  \mathcal{I}^{\alpha'}_{(k_j')} \textup{ for all } 1 \leq j \leq J \}\geq 1-\alpha. 
$       
\end{theorem}

The remarks for Theorem \ref{thm:samp_fp} holds analogously for Theorem \ref{thm:sp}. 
Below are some additional remarks. 
First, as long as the treatment assignment is independent of the potential outcomes, 
Theorem \ref{thm:sp} holds as if the experiment is a CRE. 
This is because the inference can also be justified by exchangeability of potential outcomes; see Appendix \ref{sec:tech_super} for details.

Second, when the experiment is a SCRE, we can infer population effect quantiles using inference for sample effect quantiles in Theorem \ref{thm:treated_scre}. 
Importantly, the stratification in the SCRE can depend on units' pretreatment covariates
in an arbitrary way. 
Note that, in such cases, inference for population average treatment effects can be challenging and may need assumptions on the stratification structure; see, e.g., \citet{Bai2022} and \citet{bai2022inference}.

Third, as discussed in Remark \ref{rmk:super_limit}, 
intervals for superpopulation quantiles in Theorem \ref{thm:sp} can be viewed as limits of that in Theorem \ref{thm:samp_fp} with finite population size 
$N\rightarrow \infty$. 
Specifically, the term $\Delta_{\textup{H}}$ in Theorem \ref{thm:samp_fp} 
will converge to $\Delta_{\textup{M}}$ in Theorem \ref{thm:sp} as $N\rightarrow \infty$; see Appendix \ref{sec:tech_super} for details.

{\rev 
Fourth, 
as discussed in Remark \ref{rmk:treated_generalize}, 
in matched observational studies, if the matched treated units are i.i.d.~samples from a superpopulation, we can conduct sensitivity analysis for effect quantiles among the superpopulation of treated units. 
Importantly, 
similar to Remark \ref{rmk:treated_generalize}, 
when applying Theorem \ref{thm:sp}, we should use prediction intervals for effect quantiles among treated units, rather than those for effect quantiles among all units as in Theorem \ref{thm:sp}. 
We summarize the inference procedure in the following corollary. 

\begin{corollary}\label{cor:super_match}
Consider a matched observational study where the $n_\treat$ matched treated units are i.i.d. samples from a superpopulation of treated units, for which the distribution function of individual treatment effect is denoted by $F_{\sp}(\cdot)$. 
Suppose that the population effect quantiles $F_{\sp}^{-1}(\beta_j)$s are of interest, where $0\le \beta_1 < \beta_2 < \ldots < \beta_J \le 1$ and $J\ge 1$. 
For any $\alpha \in (0,1)$ and $0\le k_1' \le k_2' \le \ldots \le k_J' \le n_\treat$, 
define 
$\Delta_{\textup{M}}(k_{1:J}'; n_\treat, \beta_{1:J})$ as in \eqref{eq:delta_M},
and let 
$\mathcal{I}^{\alpha'}_{\treat(k_j')}$s be simultaneous $1-\alpha'$ one-sided prediction intervals for sample effect quantiles $\tau_{\treat(k_j')}$s among treated units using Theorem  \ref{thm:treated_sen},
where 
$\alpha' = \alpha - \Delta_{\textup{M}}(k_{1:J}'; n_\treat, \beta_{1:J}) > 0$.
Then the $\mathcal{I}^{\alpha'}_{\treat(k_j')}$s 
are simultaneous confidence intervals for the quantiles $F_{\sp}^{-1}(\beta_j)$s, in the sense that 
$
    \Pr\{F_{\sp}^{-1}(\beta_j) \in  \mathcal{I}^{\alpha'}_{\treat(k_j')} \textup{ for all } 1 \leq j \leq J \}\geq 1-\alpha. 
$       
\end{corollary}
}

\section{Simulation studies}

\subsection{Comparison between the two improved methods}
\label{Sec: method comparison}

In this subsection, 
we conduct simulations to compare the two improved methods in Sections \ref{sec:effect_treated} and \ref{sec:berger}. 
We consider a CRE with sample size $n = 100$, where half of the units are randomly assigned to treatment and control, respectively. 
We generate the potential outcomes as i.i.d.\ samples from the following model:
\begin{align}\label{eq:DGP}
    Y_i(0) \overset{\text{i.i.d}}{\sim} N(0, \rho^2), \quad Y_i(1) \overset{\text{i.i.d}}{\sim} N(2, 1-\rho^2),
    \quad \text{ where } \rho^2 = 0.1, 0.2,\dots, 0.9. 
\end{align}
Specifically, for each iteration, we generate a new set of potential outcomes and a new treatment assignment vector. 
In this way, we are investigating the average performance of the proposed methods over the randomly generated potential outcomes.

For each of the 9 data generating models in \eqref{eq:DGP}, 
we construct confidence intervals for quantiles of individual treatment effects using the following three methods: 
\begin{enumerate}[label={(\alph*)}, topsep=1ex,itemsep=-0.3ex,partopsep=1ex,parsep=1ex]
    \item the original method in \citet{CDLM21quantile}, denoted by \textsf{M0};
    \item the first improved method in Section \ref{sec:effect_treated}, particularly the one described in Theorem \ref{thm:comb_all}, denoted by \textsf{M1};
    \item the second improved method in Section \ref{sec:berger}, particularly the one described in Remark \ref{rmk:BB_comb} with $k_j'$s chosen based on Remark \ref{rmk:choice_k_prime} with $\gamma = 0.5$, denoted by \textsf{M2}. 
\end{enumerate}

\subsubsection{Comparing the non-informativeness of the confidence intervals}

{\rev
As discussed in \citet{CDLM21quantile}, their method exhibits limitations in its ability to produce informative confidence intervals for lower quantiles. 
Below we assess the performance of the method \textsf{M0}  in \citet{CDLM21quantile} and our two improved methods \textsf{M1} and \textsf{M2}, focusing on the non-informativeness of the confidence intervals for treatment effect quantiles. 
Specifically, we will examine the proportion of non-informative confidence intervals, defined as $(-\infty, \infty)$, over 500 simulations for ten quantiles ($10\%$, $20\%$, \dots, $100\%$) of individual treatment effects. 
Note that both \textsf{M0} and \textsf{M1} provide simultaneous confidence intervals for all quantiles of individual treatment effects. 
For \textsf{M2}, we construct both individual and simultaneous confidence intervals for the selected $10$ quantiles of individual treatment effects. 
For all of these three methods, we use the Stephenson rank sum statistic with $s=6$.
The simulation results for $\rho^2 = 0.5$ are presented in Table \ref{table: non-informativeCI}, and the results for $\rho^2 = 0.1$ and $\rho^2 = 0.9$ are similar to those for $\rho^2 = 0.5$ and are thus omitted here.

\begin{table}[ht]
\centering
\resizebox{\textwidth}{!}{%
\begin{tabular}{llrrrrrrrrrr}
  \hline
  \multicolumn{2}{c}{} & \multicolumn{10}{c}{Quantile} \\ 
   & Method & 10\% & 20\% & 30\% & 40\% & 50\% & 60\% & 70\% & 80\% & 90\% & 100\% \\ 
  \hline
  & \textsf{M0} & 100\% & 100\% & 100\% & 100\% & 100\% & 100\% & 0\% & 0\% & 0\% & 0\% \\ 
  & \textsf{M1} & 100\% & 100\% & 0\% & 0\% & 0\% & 0\% & 0\% & 0\% & 0\% & 0\% \\ 
  & \textsf{M2 (Pointwise)} & 100\% & 100\% & 100\% & 81\% & 0\% & 0\% & 0\% & 0\% & 0\% & 0\% \\ 
  & \textsf{M2 (Simultaneous)} & 100\% & 100\% & 100\% & 100\% & 0\% & 0\% & 0\% & 0\% & 0\% & 0\% \\ 
  \hline
\end{tabular}}%
\caption{Proportions of non-informative confidence intervals across different methods and quantiles.}
\label{table: non-informativeCI}
\end{table}

From Table \ref{table: non-informativeCI}, \textsf{M0} performs poorly for lower quantiles from $10\%$ to $60\%$, always producing non-informative intervals. 
In contrast, 
the intervals from the improved method \textsf{M1} are always informative for quantiles at or above $30\%$, while they are still non-informative for the $10\%$ and $20\%$ quantiles. 
The performance of \textsf{M2} lies between that of \textsf{M0} and \textsf{M1}.
Specifically, the intervals from \textsf{M2} are always informative for quantiles at or above $50\%$, 
and the pointwise intervals from \textsf{M2} are also informative $19\%$ of the 
 time for the $40\%$ quantiles of individual effects. 
These results indicate that our improved methods can substantially reduce the number of non-informative confidence intervals when inferring all quantiles of individual treatment effects. }



\subsubsection{Comparing the magnitude of the confidence bounds}

To further evaluate the performance of the three methods, we next focus on the magnitude of the lower confidence bounds for $50\%, 60\%, \ldots, 90\%$ quantiles of individual treatment effects.
For each of the quantiles and under each data generating model determined by $\rho^2$, 
we calculate the median of the $90\%$ lower confidence limits derived from each method over 500 simulations. 
Intuitively, 
the larger the median value of the lower confidence limits, 
the more powerful a method is. 

\begin{figure}
    \centering
\includegraphics[scale=0.45]{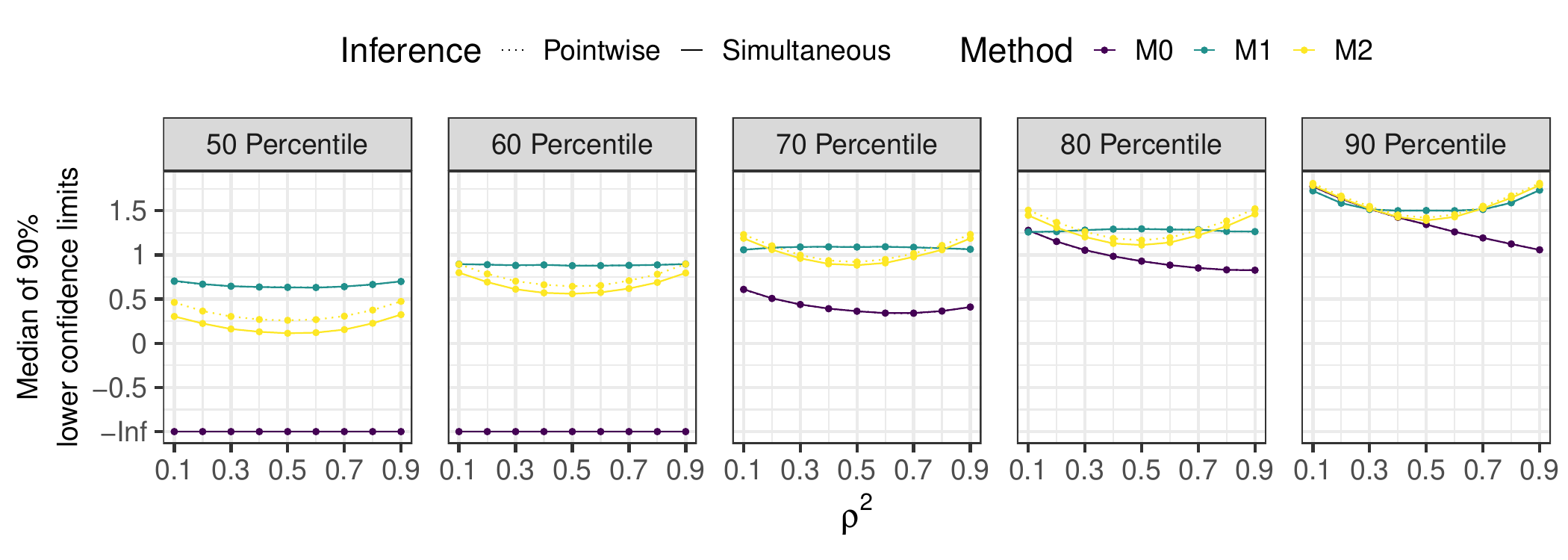}
    \caption{Median of the $90\%$ lower confidence limits of $\tau_{(k)}$s over 500 simulations using different methods with $s=6$ for $k = \lceil 0.5 \times N \rceil, \lceil 0.6 \times N \rceil, \lceil 0.7 \times N \rceil, \lceil 0.8 \times N \rceil, \lceil 0.9 \times N \rceil$ and $N=n=100$.}
    \label{fig: FP_method_comparison}
\end{figure}


We summarize the simulation results in Figure \ref{fig: FP_method_comparison} and observe three consistent trends.
First, \textsf{M1} and \textsf{M2} appear to be superior to \textsf{M0} in general; the gain is more pronounced for smaller quantiles of individual effects.
Second, 
\textsf{M1} outperforms \textsf{M2} when both potential outcomes have similar variances, e.g., $\rho^2 = 0.5$, and when inferring smaller quantiles of individual effects. 
However, \textsf{M2} can outperform \textsf{M1} when the two potential outcomes have rather different variances, especially for inferring larger quantiles of individual effects. 
We briefly defer the explanation of this phenomenon. 
Third, comparing 
the simultaneous and individual lower confidence limits,
\textsf{M2} suffers only a small amount from 
correction due to simultaneous inference. 


We now provide some insights into the performance of the two improved methods in the above simulation. 
The first improved method \textsf{M1} combines inference for treatment effects among treated and control units, 
while the second improved method \textsf{M2} tries to use only effects among treated or control units to infer effects for all units. 
It is worth noting that, when inferring treatment effects among treated units, we are essentially focusing on the comparison of right tails of the two potential outcome distributions. 
When inferring treatment effects among control units, we will switch treatment labels and change outcome signs, under which we are then focusing on the comparison of the left tails of the two potential outcome distributions. 
As an example, 
Figure \ref{fig: boxplot} shows the empirical distributions of potential outcomes based on simulated data from the two data generating models in \eqref{eq:DGP} with $\rho^2 = 0.5$ and $\rho^2 = 0.9$, respectively.


\begin{figure}
    \centering
\includegraphics[scale=0.4]{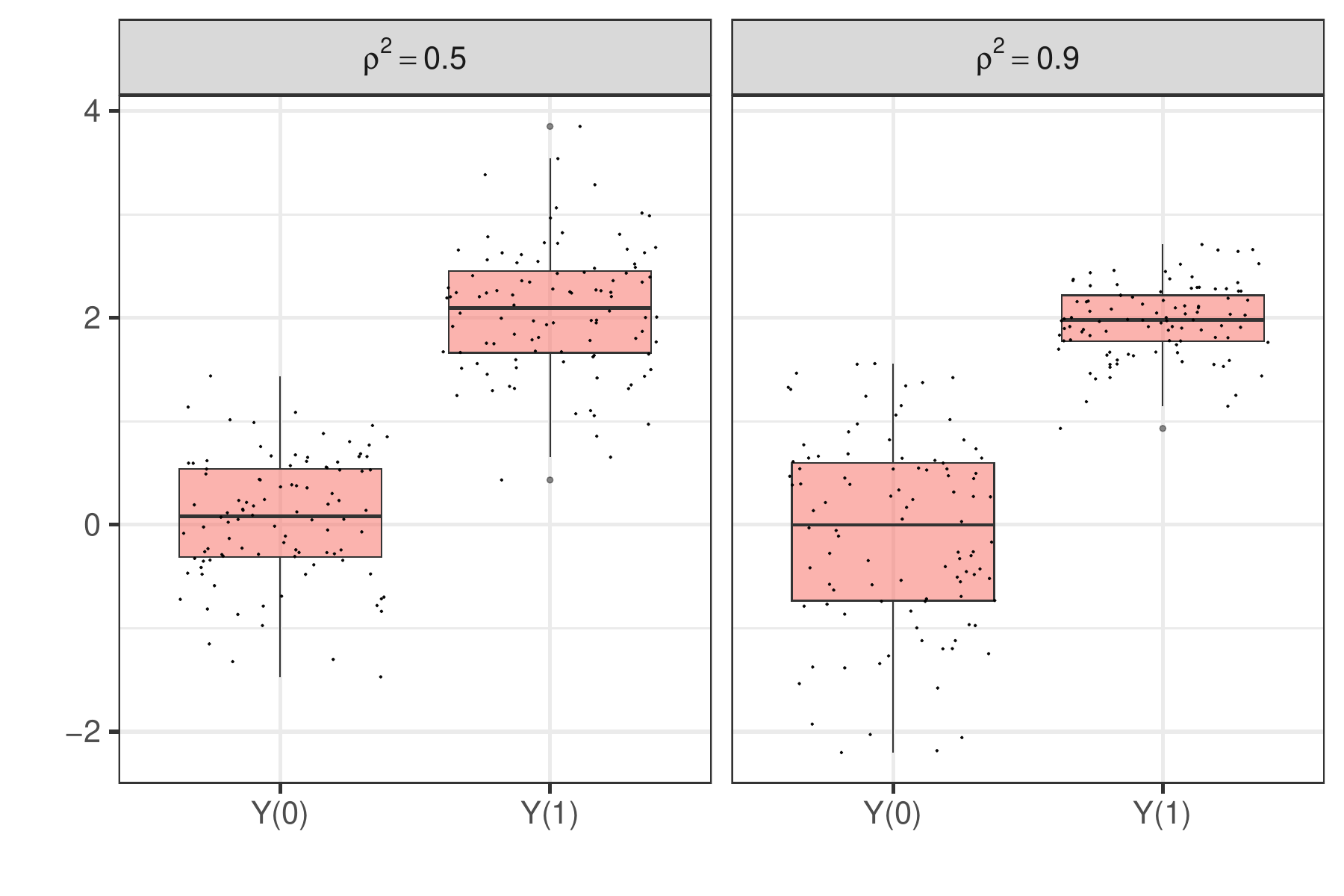}
    \caption{Boxplots of potential outcome distributions generated by $Y_i(0) \overset{\text{i.i.d}}{\sim} N(0, \rho^2)$ and $  Y_i(1) \overset{\text{i.i.d}}{\sim} N(2, 1-\rho^2)$ with $\rho^2 = 0.5$ (left panel) and $\rho^2 = 0.9$ (right panel), and $N=n=100$.}
    \label{fig: boxplot}
\end{figure}

First, when 
$\rho^2 = 0.5$, 
the comparison of the tails of the two potential outcomes distributions gives similar results using either the original data or the data with switched treatment labels and flipped outcomes. 
Specifically, with the original data, we are comparing the right tails of $Y(1) \sim \mathcal{N}(2, 0.5)$ to that of $Y(0) \sim \mathcal{N}(0, 0.5)$. 
With the transformed data, we are comparing the right tails of $-Y(0) \sim \mathcal{N}(0, 0.5)$ to that of $-Y(1) \sim \mathcal{N}(-2, 0.5)$. 
Not surprisingly, both comparisons yield similar results, or more precisely lead to similar prediction intervals for quantiles of treatment effects among treated and control units. 
Consequently, in this scenario, the first improved method \textsf{M1}, which directly integrates inference results for effects among treated and control units, tends to be more powerful.


Second, when $\rho^2 = 0.9$, 
the comparison of the tails of the two potential outcomes distributions gives rather different results when using the original data and the data with switched treatment labels and flipped outcomes.
Specifically, with the original data, we are comparing the right tails of $Y(1) \sim \mathcal{N}(2, 0.1)$ to that of $Y(0) \sim \mathcal{N}(0, 0.9)$. 
With the transformed data, we are comparing the right tails of $-Y(0) \sim \mathcal{N}(0, 0.9)$ to that of $-Y(1) \sim \mathcal{N}(-2, 0.1)$. 
Not surprisingly, with the transformed data, the ``treatment'' potential outcomes exhibit heavier right tails compared to the ``control'' potential outcomes. 
Consequently, the prediction intervals for quantiles of treatment effects among control units will be more informative than that for treated units.
Therefore, the second improved method \textsf{M2}, which tries to utilize only effects among treated or control units to infer effects for all units, can be  more preferable in this scenario. 

\subsection{Tuning parameter $\gamma$ for the second improved method}

We conduct simulation to examine the sensitivity of the second improved method in Section \ref{sec:berger}, denoted again by \textsf{M2}, to the tuning parameter $\gamma$. 
Similar to the simulation 
in Section \ref{Sec: method comparison}, we apply \textsf{M2} 
as described in
Remark \ref{rmk:BB_comb} and employ the Stephenson rank sum statistic with $s = 6$.
We use the same data generating model as in Section \ref{Sec: method comparison} and consider a series of $\gamma$ values: $0, 0.1, \dots, 0.9$. For each $\gamma$ value and each quantile of individual effects, we report the median of the $90\%$ lower confidence limits across $500$ simulations.
The results are summarized in Figure \ref{fig: gamma}.

\begin{figure}
    \centering
\includegraphics[scale=0.45]{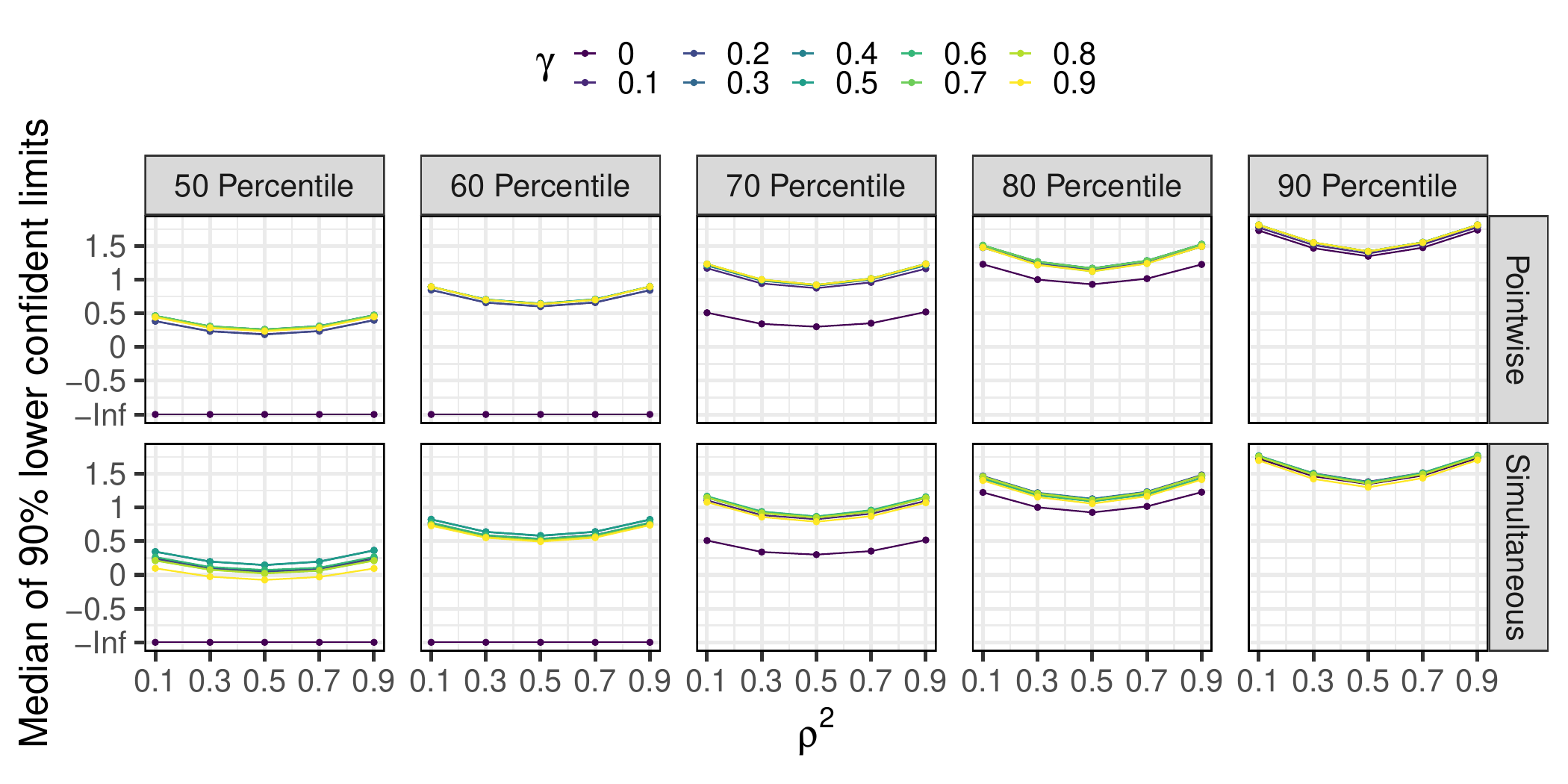}
    \caption{Median of the $90\%$ lower confidence limits derived from the improved method in Section \ref{sec:berger} with different values of $\gamma$ across 500 simulations.}
    \label{fig: gamma}
\end{figure}

First, \textsf{M2} is generally robust to different choices of $\gamma > 0$ for both simultaneous and pointwise inference. 
Second, when $\gamma = 0$, \textsf{M2} is significantly less powerful. Specifically, for lower effect quantiles such as the 50th and 60th percentiles, \textsf{M2} with 
$\gamma = 0$ returns the entire real line as confidence intervals across all data-generating processes. On the contrary, \textsf{M2} with any choice of $\gamma > 0$ can detect positive median treatment effects under all data generating processes. 
Third, for simultaneous inference, large $\gamma$ values, such as $\gamma = 0.9$, perform slightly worse compared to other positive $\gamma$ values. In general, the performance of \textsf{M2} is stable and robust to positive $\gamma$ values across all data generating processes.

\section{Proofs of theorems and corollaries}

\subsection{Proof of Theorem \ref{thm:treated}}


\begin{proof}[\bf Proof of Theorem \ref{thm:treated}(a)]
First, we prove that $p_{k,c}^{n, \treat}$ is a valid $p$-value for testing $H_{k,c}^{n, \treat}$. 
Below we prove a stronger version, where we allow $c$ to be an arbitrary random variable, not necessarily a constant. 
If $H_{k,c}^{n, \treat}$ holds, then $\bs{\tau} \in \mathcal{H}^{n,\treat}_{k,c}$. 
Consequently $\inf_{\bs{\delta}\in \mathcal{H}^{n, \treat}_{k,c}} 
            t(\bs{Z}, \bs{Y} - \bs{Z} \circ \bs{\delta}) \leq  t(\bs{Z}, \bs{Y} - \bs{Z} \circ \bs{\tau})=t(\bs{Z}, \bs{Y}(0))$.
Recall that $G(x) = \Pr\{ t(\bs{Z}, \bs{y}) \ge x \}$ is the tail probability of a distribution free test statistic $ t(\bs{Z}, \bs{y} )$ for any $\bs{y}\in \mathbb{R}^n$. Because $G(\cdot)$ is monotone decreasing function,
we can know that, when $H_{k,c}^{n, \treat}$ holds, 
\begin{equation}\label{ineq: deterministic bound for pval} 
p_{k,c}^{n, \treat} \equiv G\left(\inf_{\bs{\delta}\in \mathcal{H}^{n, \treat}_{k,c}} 
            t(\bs{Z}, \bs{Y} - \bs{Z} \circ \bs{\delta})
            \right) \geq G
            \left(
            t(\bs{Z}, \bs{Y}(0))
            \right). 
 \end{equation}           
 Therefore, 
 for any $\alpha \in (0,1)$, any $0\le k \le n_\treat$ and any  random variable $c$, we have 
\begin{align*} 
    \Pr(p_{k,c}^{n, \treat} \le \alpha \text{ and }  H_{k,c}^{n, \treat} \text{ holds} ) &
    \leq
    \Pr\left\{G
    \left(
    t(\bs{Z}, \bs{Y}(0))
    \right) \le \alpha \right\}
    \leq \alpha.
\end{align*}
The last inequality holds because $G\left(
            t(\bs{Z}, \bs{Y}(0))\right)$ is the tail
probability of $t(\bs{Z}, \bs{Y}(0))$ evaluated at its realized value, which must be stochastically larger than or equal to $\Unif(0,1)$; see, e.g, \citet[][Lemma A4]{CDLM21quantile}.       Therefore, $p_{k,c}^{n, \treat}$ is a valid $p$-value for testing $H_{k,c}^{n, \treat}$.

Second, we prove the equivalence forms of $p_{k,c}^{n, \treat}$ in \eqref{eq:p1}. 
By definition, $  
\bs{\tau} \in \mathcal{H}^{n}_{n_\control+k,c} \Longleftrightarrow n(c) \le n_\treat-k  \Longrightarrow n_{\treat}(c) \le n_\treat-k \Longleftrightarrow \bs{\tau} \in \mathcal{H}^{n,\treat}_{k,c}$. Hence,  $\mathcal{H}^{n}_{n_\control+k,c} \subset \mathcal{H}^{n,\treat}_{k,c}$ and thus $ \inf_{\bs{\delta}\in \mathcal{H}^{n, \treat}_{k,c}} 
    t(\bs{Z}, \bs{Y} - \bs{Z} \circ \bs{\delta}) \leq  \inf_{\bs{\delta}\in \mathcal{H}^{n}_{n_\control+k, c}} 
    t(\bs{Z}, \bs{Y} - \bs{Z} \circ \bs{\delta})$. For any $\bs{\delta} \in  \mathcal{H}^{n, \treat}_{k,c}$, let $\mathcal{I}_{c} = \{i: Z_i = 1, \delta_i > c, 1\leq i\leq n \}$. By the definition of $\mathcal{H}^{n, \treat}_{k,c}$, we have $|\mathcal{I}_{c}|\leq n_\treat-k$. Define a vector $\bs{\theta} = (\theta_1, \theta_2, \dots, \theta_n)^{\top}$ with $\theta_i = \infty$ if $i\in \mathcal{I}_{c} $ and $c$ otherwise. For any units with $Z_i = 1$, we have $\delta_i \leq \theta_i$ and thus $Y_i - Z_i\delta_i \geq Y_i - Z_i\theta_i$. For any units with $Z_i = 0$, we have $Y_i - Z_i\delta_i = Y_i = Y_i - Z_i\theta_i$. Using the fact that $t(\cdot, \cdot)$ is effect increasing \citep[][Proposition 2]{CDLM21quantile}, we further have $ 
    t(\bs{Z}, \bs{Y} - \bs{Z} \circ \bs{\delta}) \geq  
    t(\bs{Z}, \bs{Y} - \bs{Z} \circ \bs{\theta})$. Note that $\bs{\theta} \in  \mathcal{H}^{n}_{n_\control+k,c}$ by construction. Thus, we have shown that for any $\bs{\delta} \in  \mathcal{H}^{n, \treat}_{k,c}$, there exists a $\bs{\theta} \in  \mathcal{H}^{n}_{n_\control+k,c}$ such that $ 
    t(\bs{Z}, \bs{Y} - \bs{Z} \circ \bs{\delta}) \geq  
    t(\bs{Z}, \bs{Y} - \bs{Z} \circ \bs{\theta})$.
    Therefore,
    $ \inf_{\bs{\delta}\in \mathcal{H}^{n, \treat}_{k,c}} 
    t(\bs{Z}, \bs{Y} - \bs{Z} \circ \bs{\delta}) \geq  \inf_{\bs{\delta}\in \mathcal{H}^{n}_{n_\control+k, c}} 
    t(\bs{Z}, \bs{Y} - \bs{Z} \circ \bs{\delta})$. 
    From the above, we have  $ \inf_{\bs{\delta}\in \mathcal{H}^{n, \treat}_{k,c}} 
    t(\bs{Z}, \bs{Y} - \bs{Z} \circ \bs{\delta}) =  \inf_{\bs{\delta}\in \mathcal{H}^{n}_{n_\control+k, c}} 
    t(\bs{Z}, \bs{Y} - \bs{Z} \circ \bs{\delta})$ and thus $p_{k,c}^{n, \treat} =
    p_{n_0+k, c}^n$.

    

    From the above, Theorem \ref{thm:treated}(a) holds.
\end{proof}
\begin{proof}[\bf Proof of Theorem \ref{thm:treated}(b)]
For any $1 \le k\le n_\treat$ and $\alpha\in (0,1)$, 
\begin{align*}
    \Pr\{ \tau_{\treat(k)} \notin \mathcal{I}_{\treat(k)}^{\alpha} \} 
    & = 
    \Pr\{ p_{k,\tau_{\treat(k)}}^{n, \treat} \le \alpha  \}
    = 
    \Pr\{ p_{k,\tau_{\treat(k)}}^{n, \treat} \le \alpha,  \tau_{\treat(k)} \le \tau_{\treat(k)} \}
    \\
    & = \Pr\{ p_{k,\tau_{\treat(k)}}^{n, \treat} \le \alpha
    \text{ and }  H_{k,\tau_{\treat(k)}}^{n, \treat} \text{ holds} 
    \} \le \alpha, 
\end{align*}
where the last equality holds from Theorem \ref{thm:treated}(a). Thus, 
$\mathcal{I}_{\treat(k)}^{\alpha} \equiv \{c\in \mathbb{R}: p_{k,c}^{n, \treat} > \alpha \}$ is a $1-\alpha$ prediction set for $\tau_{\treat(k)}$.
From the equivalent forms in \eqref{eq:p1} and the monotonicity of $p_{k,c}^n$ established in \citet[][Theorem 4]{CDLM21quantile}, 
$p_{k,c}^{n, \treat}$ is increasing in $c$. 
This implies that the prediction set must be of form $(c, \infty)$ or $[c, \infty)$ with $c = \inf\{c: p_{k,c}^{n, \treat} > \alpha\}$.
From the above, Theorem \ref{thm:treated}(b) holds. 
\end{proof}


\begin{proof}[\bf Proof of Theorem \ref{thm:treated}(c)]
We first prove that $ \bigcap_{k=1}^{n_\treat} 
\big\{ 
\tau_{\treat(k)} \in \mathcal{I}_{\treat(k)}^{ \alpha}
\big\}    \Longleftrightarrow \bs{\tau} \in \bigcap_{k, c: p_{k,c}^{n, \treat}\leq\alpha} \mathcal{H}^{n,\treat\complement}_{k,c}$.       
Suppose $ \tau_{\treat(k)} \in \mathcal{I}_{\treat(k)}^{ \alpha}$ for any $1\leq k \leq n_\treat$. By the definition of $\mathcal{I}_{\treat(k)}^{ \alpha}$, we then have $p_{k,\tau_{\treat(k)}}^{n, \treat}>\alpha$ for all $1\leq k \leq n_\treat$. Because $p_{k,c}^{n, \treat}$ is increasing in $c$, 
for any $1\leq k \leq n_\treat$ and $c \in \mathbbm{R}$ such that $p_{k,c}^{n, \treat}\leq\alpha$, 
we have $ \tau_{\treat(k)}>c$, or equivalently $\bs{\tau} \in \mathcal{H}^{n,\treat \complement}_{k,c}$, where 
$\mathcal{H}^{n,\treat \complement}_{k,c}$ denotes the complement of the set $\mathcal{H}^{n,\treat}_{k,c}$. 
This implies that $\bs{\tau} \in \bigcap_{k, c: p_{k,c}^{n, \treat}\leq\alpha} (\mathcal{H}^{n,\treat}_{k,c})^{\complement}$. Now suppose $ \tau_{\treat(k)} \not\in \mathcal{I}_{\treat(k)}^{\alpha}$ for some $1\leq k \leq n_\treat$, which implies $p_{k,\tau_{\treat(k)}}^{n, \treat}\leq\alpha$. By definition, $\bs{\tau} \in \mathcal{H}^{n,\treat}_{k,\tau_{\treat(k)}}\subset \bigcup_{k, c: p_{k,c}^{n, \treat}\leq\alpha} \mathcal{H}^{n,\treat}_{k,c}$. Thus, $\bs{\tau} \not\in \bigcap_{k, c: p_{k,c}^{n, \treat}\leq\alpha} (\mathcal{H}^{n,\treat}_{k,c})^{\complement}$. From the above, 
we can know that $ \bigcap_{k=1}^{n_\treat} 
\big\{ 
\tau_{\treat(k)} \in \mathcal{I}_{\treat(k)}^{ \alpha}
\big\}    \Longleftrightarrow \bs{\tau} \in \bigcap_{k, c: p_{k,c}^{n, \treat}\leq\alpha} \mathcal{H}^{n,\treat\complement}_{k,c}$.

We then prove that $\bigcap_{k, c: p_{k,c}^{n, \treat}\leq\alpha} \mathcal{H}^{n,\treat\complement}_{k,c}$ covers the true treatment effect vector
with probability at least $1-\alpha$.
Suppose that $ \bs{\tau} \in  \bigcup_{k,c: p_{k,c}^{n, \treat}\leq\alpha} \mathcal{H}^{n,\treat}_{k,c}$. Then there exists $(k,c)$ such that 
$p_{k,c}^{n, \treat}\leq\alpha$ and $\bs{\tau} \in   \mathcal{H}^{n,\treat}_{k,c}$, 
which further implies that $ \inf_{\bs{\delta}\in \mathcal{H}^{n, \treat}_{k,c}} 
t(\bs{Z}, \bs{Y} - \bs{Z} \circ \bs{\delta})\leq  t(\bs{Z}, \bs{Y} - \bs{Z} \circ \bs{\tau})=t(\bs{Z}, \bs{Y}(0))$. Because the tail probability function $G(\cdot)$ is decreasing, we further have that
$G
\left(
t(\bs{Z}, \bs{Y}(0))
\right) \leq G
\left(
\inf_{\bs{\delta}\in \mathcal{H}^{n, \treat}_{k,c}} 
t(\bs{Z}, \bs{Y} - \bs{Z} \circ \bs{\delta})
\right) = p_{k,c}^{n, \treat}\leq\alpha$. Therefore,
$$ 
\Pr\Big( 
\bs{\tau} \in  \bigcap_{k,c: p_{k,c}^{n, \treat}\leq\alpha} \mathcal{H}^{n,\treat\complement}_{k,c} \Big) 
= 
1-  \Pr\Big( 
\bs{\tau} \in  \bigcup_{k,c: p_{k,c}^{n, \treat}\leq\alpha} \mathcal{H}^{n,\treat}_{k,c} \Big)\geq
1-\Pr\left\{  G
\left(
t(\bs{Z}, \bs{Y}(0))
\right)\leq \alpha \right\} 
\geq 1-\alpha,
$$
where the last inequality holds due to the validity of the Fisher randomization test; see also the proof of Theorem \ref{thm:treated}(a).   

From the above, Theorem \ref{thm:treated}(c) holds.
\end{proof}  

\subsection{Proof of Theorem \ref{thm:comb_all}}

\begin{proof}[\bf Proof of Theorem \ref{thm:comb_all}]
From the condition of Theorem \ref{thm:comb_all} and by Bonferroni's inequality,  
\begin{align*}
    \Pr\left( 
    \bigcap_{k=1}^{n_\treat} 
    \big\{
    \tau_{\treat(k)} \in \mathcal{I}_{\treat(k)}^{\alpha}
    \big\}
    \ 
    \bigcap
    \ 
\bigcap_{j=1}^{n_\control} 
    \big\{ 
    \tau_{\control(j)} \in \mathcal{I}_{\control(j)}^{\alpha}
    \big\}
    \right)
    \ge 1-2\alpha, 
\end{align*}   
recalling that $\tau_{\treat(k)}$s and $\tau_{\control(j)}$s are the sorted individual treatment effects among treated and control units, respectively. 
Below we show that, as long as the two events $ \bigcap_{k=1}^{n_\treat} 
\big\{ 
\tau_{\treat(k)} \in \mathcal{I}_{\treat(k)}^{\alpha}
\big\}$ and $   \bigcap_{k=1}^{n_\control} 
\big\{ 
\tau_{\control(k)} \in \mathcal{I}_{\control(k)}^{\alpha}
\big\}$
hold, 
the event $ \bigcap_{k=1}^{n} 
\big\{ 
\tau_{(k)} \in \mathcal{I}_{(k)}^\alpha
\big\}$ must occur.

By definition, 
$ \{\mathcal{I}_{(1)}^{\alpha}, \mathcal{I}_{(2)}^{\alpha}, \dots, \mathcal{I}_{(n)}^{\alpha}\}$ is a permutation of 
$\{\mathcal{I}_{\treat(1)}^{\alpha}, \mathcal{I}_{\treat(2)}^{\alpha},\dots, \mathcal{I}_{\treat(n_\treat)}^{\alpha},\mathcal{I}_{\control(1)}^{\alpha}, \mathcal{I}_{\control(2)}^{\alpha},\dots,$ $ \mathcal{I}_{\control(n_\control)}^{\alpha}\}$, 
and 
$\{\tau_{(1)}, \tau_{(2)}, \dots, \tau_{(n)}\}$ is a permutation of
$ 
\{\tau_{\treat(1)}, \tau_{\treat(2)}, \dots, \tau_{\treat(n_\treat)},
\tau_{\control(1)}, \tau_{\control(2)}, \dots, \tau_{\control(n_\control)}\}.
$  
Suppose that the two events $\bigcap_{k=1}^{n_\treat} 
    \big\{ 
    \tau_{\treat(k)} \in \mathcal{I}_{\treat(k)}^{\alpha}
    \big\}
    $ and $  
\bigcap_{j=1}^{n_\control} 
    \big\{ 
    \tau_{\control(j)} \in \mathcal{I}_{\control(j)}^{\alpha}
    \big\}$
hold. 
We must have $\tau_{(k)} \in \mathcal{I}^\alpha_{(\sigma_k)}$ for $1\leq k\leq n$ for some permutation $\{\sigma_1, \sigma_2, \dots, \sigma_n\}$ of the index set $\{1,2,\dots,n\}$.
Because $\mathcal{I}^\alpha_{(k)}$s are one-sided intervals of form $(c, \infty)$ or $[c, \infty)$,  
for any $1\leq j \leq k \leq n$, 
we must have $\tau_{(k)} \in \mathcal{I}^\alpha_{(\sigma_j)}$. 
Thus, for any given $1\leq k \leq  n$, we have
$\tau_{(k)} \in \bigcap_{j=1}^{k}\mathcal{I}^\alpha_{(\sigma_j)} \subseteq \mathcal{I}^\alpha_{(k)}$, 
where the last inclusion holds because the set $\{\sigma_1, \sigma_2, \ldots, \sigma_k\}$ contains at least one element in $\{k, k+1, \ldots, n\}$ and the set $\mathcal{I}^\alpha_{(j)}$ becomes smaller as $j$ increases. 
From the above, 
$ \bigcap_{k=1}^{n} 
\big\{ 
\tau_{(k)} \in \mathcal{I}_{(k)}^\alpha
\big\}$ occurs when the two events  $\bigcap_{k=1}^{n_\treat} 
    \big\{ 
    \tau_{\treat(k)} \in \mathcal{I}_{\treat(k)}^{\alpha}
    \big\}$ and $   \bigcap_{k=1}^{n_\control} 
    \big\{ 
    \tau_{\control(k)} \in \mathcal{I}_{\control(k)}^{\alpha}
    \big\}$
occur. 
Therefore, 
\begin{align*}
\Pr\left( 
\bigcap_{k=1}^{n} 
\big\{ 
\tau_{(k)} \in \mathcal{I}_{(k)}^\alpha
\big\}
\right) \geq
    \Pr\left( 
\bigcap_{k=1}^{n_\treat} 
    \big\{ 
    \tau_{\treat(k)} \in \mathcal{I}_{\treat(k)}^{\alpha}
    \big\}
\ 
    \bigcap
\ 
\bigcap_{j=1}^{n_\control} 
    \big\{ 
    \tau_{\control(j)} \in \mathcal{I}_{\control(j)}^{\alpha}
    \big\}
    \right)
    \ge 1-2\alpha, 
\end{align*}   
i.e., Theorem \ref{thm:comb_all} holds. 
\end{proof}

\subsection{Proof of Theorem \ref{thm:BB_pval}}

\begin{proof}[\bf Proof of Theorem \ref{thm:BB_pval}]
Under the CRE, treated units are a simple random sample of size $n_\treat$ from a total of $n$ experimental units.
Therefore, 
$n_{\treat}(c) 
= 
\sum_{i=1}^n Z_i \I(\tau_i > c) 
$
follows 
a Hypergeometric distribution with parameters $(n, n(c), n_\treat)$.
Under the null hypothesis $H^n_{k,c}$ in \eqref{eq:Hnkc}, we have $n(c) \leq n-k$, which implies that 
 $\hyper(k) \sim \text{HG}(n, n-k, n_\treat)$ is stochastically larger than $n_{\treat}(c)$ and thus
$\Pr\left(n_{\treat}(c) > n_{\treat}-k'\right)\leq \Pr( 
    \mathfrak{H}(k)> n_{\treat} - k'
    )$ for all $0 \leq k' \leq n_\treat$. Thus, under $H^n_{k,c}$, for any $0\le k \le n$, $0\le k' \le n_{\treat}$ and any $c\in \mathbb{R}$, we have
\begin{align*}
\Pr\left(\overline{p}_{k, c}^n(k')\leq \alpha \right)
     & = \Pr\left(p^{n,\treat}_{k', c} + \Pr( 
    \mathfrak{H}(k) > n_{\treat} - k'
    )\leq \alpha \right)\nonumber\\
   & = \Pr\left(p^{n,\treat}_{k', c} \leq \alpha-\Pr( 
    \mathfrak{H}(k) > n_{\treat} - k'
    ),  n_{\treat}(c) \leq n_{\treat}-k'\right)
    \\&\quad \ 
    +\Pr\left(p^{n,\treat}_{k', c} \leq \alpha-\Pr( 
    \mathfrak{H}(k) > n_{\treat} - k'
    ),  n_{\treat}(c) > n_{\treat}-k'\right)\nonumber\\
  & \leq \alpha-\Pr( 
    \mathfrak{H}(k)> n_{\treat} - k'
    )
  +\Pr\left(n_{\treat}(c) > n_{\treat}-k'\right)\nonumber \\
  & \leq \alpha
\end{align*}  
where the first inequality holds because $p^{n,\treat}_{k', c}$ is a valid $p$-value for testing $H_{k',c}^{n, \treat}$ in \eqref{eq:Hn1kc} as shown in Theorem 1(a). Therefore, $\overline{p}_{k, c}^n(k')$ is a valid $p$-value for testing $H^n_{k,c}$ in \eqref{eq:Hnkc} for any prespecified integer $0\le k'\le n_{\treat}$, 
i.e., Theorem \ref{thm:BB_pval} holds. 
\end{proof}

\subsection{Proof of Theorem \ref{thm:sing_interval_BB}}
\begin{proof}[\bf Proof of Theorem \ref{thm:sing_interval_BB}]
    The proof of Theorem \ref{thm:sing_interval_BB} follows immediately from the Lehmann-style test inversion, and is thus omitted.
\end{proof}




\subsection{Proof of Theorem \ref{thm:BB_simulCI}}
\begin{proof}[\bf Proof of Theorem \ref{thm:BB_simulCI}]
By the definition of $\mathcal{I}_{\treat(k_j')}^{\alpha'}$, we have
\begin{align}
& \quad \  \Pr\left(\bigcap_{j=1}^{J}\left\{\tau_{(k_j)} \in  \mathcal{I}_{\treat(k_j')}^{\alpha'}\right\}  \right)
\nonumber
\\
&=\Pr\left(\bigcap_{j=1}^{J}\left\{\tau_{(k_j)} \in \{c: p^{n,\treat}_{k_j', c} > \alpha'\}\right\}\right)
=\Pr\left(\bigcap_{j=1}^{J}\left\{p^{n,\treat}_{k_j', \tau_{(k_j)} } > \alpha'\right\}\right)
=1-\Pr\left(\bigcup\limits_{j=1}^{J}\left\{p^{n,\treat}_{k_j', \tau_{(k_j)} } \leq \alpha'\right\}\right)\nonumber\\
&=1-\Pr\left(\bigcup\limits_{j=1}^{J}\left\{p^{n,\treat}_{k_j', \tau_{(k_j)} } \leq \alpha'\right\}, \bigcap_{j=1}^{J}H^{n,\treat}_{k_j', \tau_{(k_j)} }
\right)
-\Pr\left(\bigcup\limits_{j=1}^{J}\left\{p^{n,\treat}_{k_j', \tau_{(k_j)} } \leq \alpha'\right\}, 
\left\{\bigcap_{j=1}^{J}H^{n,\treat}_{k_j', \tau_{(k_j)} }\right\}^{\complement}
\right) \nonumber\\
&\geq 1-\Pr\left(\bigcup\limits_{j=1}^{J}\left\{p^{n,\treat}_{k_j', \tau_{(k_j)} } \leq \alpha'\right\}, \bigcap_{j=1}^{J}H^{n,\treat}_{k_j', \tau_{(k_j)} }
\right)-\Pr\left(
\left\{\bigcap_{j=1}^{J}H^{n,\treat}_{k_j', \tau_{(k_j)} }\right\}^{\complement}
\right)\label{thm5: bound}
\end{align}
Below we bound the two probabilities in \eqref{thm5: bound}, respectively.

First, suppose the events $\bigcup\limits_{j=1}^{J}\left\{p^{n,\treat}_{k_j', \tau_{(k_j)} } \leq \alpha'\right\}$ 
and $\bigcap_{j=1}^{J}H^{n,\treat}_{k_j', \tau_{(k_j)} }$ occur. 
From \eqref{ineq: deterministic bound for pval}, we can know that 
$ G\left(t(\bs{Z}, \bs{Y}(0))
\right)\leq 
p^{n,\treat}_{k_j', \tau_{(k_j)}}
$ for all $1\le j \le J$. 
This immediately implies that $ G\left(t(\bs{Z}, \bs{Y}(0))
\right)\leq 
\min_{1\le j \le J}p^{n,\treat}_{k_j', \tau_{(k_j)}} \le \alpha'. 
$
Consequently,  we have
\begin{align*}
&\Pr\left(\bigcup\limits_{j=1}^{J}\left\{p^{n,\treat}_{k_j', \tau_{(k_j)} } \leq \alpha'\right\}, 
\bigcap_{j=1}^{J}H^{n,\treat}_{k_j', \tau_{(k_j)} }
\right)  \leq 
\Pr\left\{ 
    G\left(t(\bs{Z}, \bs{Y}(0))\right)\leq \alpha'
\right\}\leq 
\alpha',
\end{align*}
where the last inequality holds due to the validity of Fisher randomization test; see also the proof of Theorem \ref{thm:treated}(a).  

We then bound the probability that event $\bigcap_{j=1}^{J}H^{n,\treat}_{k_j', \tau_{(k_j)} }$ fails.
Let $\phi(\cdot)$ be a permutation of $\{1, 2, \dots, n\}$ such that $(\rank_{1}(\bs{\tau}), \rank_{2}(\bs{\tau}),\ldots,\rank_{n}(\bs{\tau}))=(\phi(1),\phi(2),...,\phi(n))$, where $\rank_{i}(\bs{\tau})$ denotes the rank of the
$i$th coordinate of $\bs{\tau}$.
Let $\tau_{(0)} = -\infty$.
Define $w_j = \sum_{i=1}^{n} Z_{i} \mathbbm{1}\left\{ \tau_{(k_j)} < \tau_i \leq \tau_{(k_{j+1})}\right\}$ and $\tilde w_j = \sum_{i=1}^{n} Z_{i} \mathbbm{1}\left\{ k_j < \phi(i) \leq k_{j+1}\right\}$ for $0 \leq j \leq J$, where $k_0 = 0$ and $k_{J+1}=n$.  
By definition, 
$ n_{\treat}(\tau_{(k_j)}) = \sum_{i=1}^n Z_i \I(\tau_i > \tau_{(k_j)}) = \sum_{i=j}^{J}w_i$.
Since we break ties using index ordering, $\I\{\tau_{(k_j)}<\tau_{\left(\phi(i)\right)}\} \leq \I\{k_j<\phi(i)\}$, where the equality is achieved when there are no ties of individual treatment effects, i.e., $\tau_{(1)}< \tau_{(2)}<\dots<\tau_{(n)}$. 
Thus, for each $j$, $\sum_{i=j}^{J} w_{i}=\sum_{i=1}^n Z_i \I( \tau_{(k_j)}<\tau_i) = \sum_{i=1}^n Z_i\I\{\tau_{(k_j)}<\tau_{\left(\phi(i)\right)}\} \leq \sum_{i=1}^n Z_i\I\{k_j<\phi(i)\}=\sum_{i=j}^{J}\tilde w_{i}$.
Consequently, $ n_{\treat}(\tau_{(k_j)})\leq \sum_{i=j}^{J}\tilde w_{i}$ for each $j$, which further implies that
\begin{align}\label{eq:bound_MHG}
\Pr\left(
\left\{\bigcap_{j=1}^{J}H^{n,\treat}_{k_j', \tau_{(k_j)} }\right\}^{\complement}
\right)
& =\Pr\left(
\bigcup\limits_{j=1}^{J}\left\{n_{\treat}( \tau_{(k_j)})> n_{\treat} - k_j' \right\}
\right) 
\leq \Pr\left(
\bigcup\limits_{j=1}^{J}\left\{\sum_{i=j}^{J}\tilde w_i> n_{\treat} - k_j' \right\}
\right). 
\end{align}
Note that the treatment assignments $Z_i$s can be equivalently viewed as the indicators for simple random samples of size $n_\treat$. 
By definition, for each $j$, $\tilde w_j$ counts the number of units from the set $\{i: k_j < \phi(i) \le k_{j+1}\}$, which has size $k_{j+1} - k_j$. 
It is then not difficult to see that  $(n_{\treat}-\sum_{j=1}^J \tilde w_j, \tilde w_1, \ldots, \tilde w_{J-1}, \tilde w_J) \sim \textup{MHG}(
    k_1, 
    k_2-k_1, 
    \ldots, 
    k_J - k_{J-1}, 
    n-k_J; n_{\treat}
    )$
follows a multivariate Hypergeometric distribution. 
By definition, the last term in \eqref{eq:bound_MHG} is equal to 
$\Delta_{\textup{H}}(k_{1:J}'; n, n_{\treat}, k_{1:J})$. 

From the above, we can know that  
$$
\Pr\left(\bigcap_{j=1}^{J}\left\{\tau_{(k_j)} \in  \mathcal{I}_{\treat(k_j')}^{\alpha'}\right\}  \right)
\geq 1-\alpha'-
\Delta_{\textup{H}}(k_{1:J}'; n, n_{\treat}, k_{1:J}) = 1-\alpha,$$
i.e., Theorem \ref{thm:BB_simulCI} holds.
\end{proof}

\subsection{Proof of Theorem \ref{thm:treated_scre}}
Below we first give a complete version of Theorem \ref{thm:treated_scre}.
\begin{theoremA}\label{supp thm:treated_scre}
    Consider the CSRE and any stratified rank score statistic $\tilde{t}(\cdot, \cdot)$ in \eqref{eq:strat_rank_score}.  
    \begin{itemize}
        \item[(i)] For any $0\le k \le n_\treat \equiv \sum_{s=1}^S n_{s\treat}$ and any $c\in \mathbb{R}$, 
        \begin{align}\label{eq: pvsl_scre}
            \tilde{p}_{k,c}^{n, \treat} 
            \equiv 
            \tilde{G}\Big( \inf_{\bs{\delta} \in \mathcal{H}_{k,c}^{n, \treat}} \tilde{t}(\bs{Z}, \bs{Y} - \bs{Z} \circ \bs{\delta} ) \Big)
            = 
            \tilde{G}\Big( \inf_{\bs{\delta} \in \mathcal{H}_{n_{\control}+k,c}^{n}} \tilde{t}(\bs{Z}, \bs{Y} - \bs{Z} \circ \bs{\delta} ) \Big)
            = \tilde{p}_{n_{\control}+k,c}^{n} 
        \end{align}
    is a valid $p$-value for testing $H_{k,c}^{n, \treat}$ in \eqref{eq:Hn1kc}, 
    where $\tilde{p}_{n_{\control}+k,c}^{n}$ is defined as in \eqref{eq:p_SL}.  
    \item[(ii)] 
       For any $1 \le k\le n_{\treat}$ and $\alpha\in (0,1)$, 
        $\tilde{\mathcal{I}}_{\treat(k)}^{\alpha} \equiv \{c\in \mathbb{R}: \tilde{p}_{k,c}^{n, \treat} > \alpha \}$ is a $1-\alpha$ prediction set for $\tau_{\treat(k)}$, in the sense that 
        $\Pr\{ \tau_{\treat(k)} \in \tilde{\mathcal{I}}_{\treat(k)}^{\alpha} \}\ge 1-\alpha$, and $\tilde{\mathcal{I}}_{\treat(k)}^\alpha$ must be an interval of form $(c, \infty)$ or $[c, \infty)$, 
        where $c = \inf\{c: \tilde{p}_{k,c}^{n, \treat} > \alpha, c \in \mathbb{R}\}$. 
                 \item[(iii)] 
        The intersection of all prediction intervals in (ii) over $0\le k \le n_{\treat}$, viewed as a prediction set for treatment effects on treated units, 
        covers the true treatment effects for treated units with probability at least $1-\alpha$, 
        or equivalently, the prediction intervals in (ii) are  simultaneously valid:
        \begin{equation}\label{eq: thmA1_PI_equivalence}
            \Pr\left( 
            \bigcap_{k=1}^{n_{\treat}} 
            \big\{ 
            \tau_{\treat(k)} \in \tilde{\mathcal{I}}_{\treat(k)}^{\alpha}
            \big\}
            \right)
            \ge 1-\alpha. 
         \end{equation}
    \end{itemize}
\end{theoremA}

\begin{proof}[\bf Proof of Theorem \ref{supp thm:treated_scre}]
    The proof follows by almost the same logic as the proof of Theorem \ref{thm:treated}, except that it uses the following facts. 
    First, the $p$-value $\tilde{p}^n_{k,c}$ is increasing in $c$ and decreasing in $k$; see \citet[][Theorem 3]{SL22quantile}. 
    Second, the stratified rank score statistic is effect increasing, which can be immediately implied by the effect increasing property of a single rank score statistic \citep[][Proposition 2]{CDLM21quantile}.
    For conciseness, we omit the detailed proof.
\end{proof}

\subsection{Proof of Theorem \ref{thm:treated_sen}}
Below we first give a complete version of Theorem \ref{thm:treated_sen}.
\begin{theoremA}\label{supp thm:treated_sen}
    Consider a matched observational study and any stratified rank score statistic $\tilde{t}(\cdot, \cdot)$ in \eqref{eq:strat_rank_score}. 
    Suppose that the unmeasured confounding strength is bounded by $\Gamma\ge 1$, i.e., the treatment assignment follows \eqref{eq:sen_assign} for some unknown $\bs{u}\in [0,1]^n$. 
    \begin{itemize}
        \item[(i)] For any $0\le k \le n_\treat \equiv \sum_{s=1}^S n_{s\treat}$ and any $c\in \mathbb{R}$, 
        \begin{align*}
            \tilde{p}_{k,c}^{n, \treat} (\Gamma)
            \equiv 
            \tilde{G}_{\Gamma}\Big( \inf_{\bs{\delta} \in \mathcal{H}_{k,c}^{n, \treat}} \tilde{t}(\bs{Z}, \bs{Y} - \bs{Z} \circ \bs{\delta} ) \Big)
            = 
            \tilde{G}_{\Gamma}\Big( \inf_{\bs{\delta} \in \mathcal{H}_{n_{\control}+k,c}^{n}} \tilde{t}(\bs{Z}, \bs{Y} - \bs{Z} \circ \bs{\delta} ) \Big)
            = \tilde{p}_{n_{\control}+k,c}^{n} (\Gamma)
        \end{align*}
    is a valid $p$-value for testing $H_{k,c}^{n, \treat}$ in \eqref{eq:Hn1kc}, 
    where $\tilde{p}_{n_{\control}+k,c}^{n} (\Gamma)$ is defined as in \eqref{eq:p_SL_sen}.  
    \item[(ii)] 
       For any $1 \le k\le n_{\treat}$ and $\alpha\in (0,1)$, 
        $\tilde{\mathcal{I}}_{\treat,\Gamma}^{\alpha} (k)\equiv \{c\in \mathbb{R}: \tilde{p}_{k,c}^{n, \treat}(\Gamma) > \alpha \}$ is a $1-\alpha$ prediction set for $\tau_{\treat(k)}$, in the sense that 
        $\Pr\{ \tau_{\treat(k)} \in \tilde{\mathcal{I}}_{\treat,\Gamma}^{\alpha} (k) \}\ge 1-\alpha$, and $\tilde{\mathcal{I}}_{\treat,\Gamma}^{\alpha} (k)$ must be an interval of form $(c, \infty)$ or $[c, \infty)$, 
        where $c = \inf\{c: \tilde{p}_{k,c}^{n, \treat}(\Gamma) > \alpha, c \in \mathbb{R}\}$. 
      
        \item[(iii)] 
        The intersection of all prediction intervals in (ii) over $0\le k \le n_{\treat}$, viewed as a prediction set for treatment effects on treated units, 
        covers the true treatment effects for treated units with probability at least $1-\alpha$, 
        or equivalently, the prediction intervals in (ii) 
        are all simultaneously valid:
        \begin{equation}\label{eq: thmA2_PI_equivalence}
            \Pr\left( 
            \bigcap_{k=1}^{n_{\treat}} 
            \big\{ 
            \tau_{\treat(k)} \in \tilde{\mathcal{I}}_{\treat,\Gamma}^{\alpha} (k)
            \big\}
            \right)
            \ge 1-\alpha. 
        \end{equation}
    \end{itemize}
\end{theoremA}

\begin{proof}[\bf Proof of Theorem \ref{supp thm:treated_sen}]
First, we prove the validity of the $p$-value $\tilde{p}_{k,c}^{n, \treat} (\Gamma)$ for testing $H_{k,c}^{n, \treat}$ in \eqref{eq:Hn1kc} under the sensitivity model with bias at most $\Gamma$. Let $\bs{u}_0$ denote the true unmeasured confounding vector associated with the sensitivity model with bias at most $\Gamma$. Under $H_{k,c}^{n, \treat}$,  we have $\bs{\tau} \in \mathcal{H}^{n,\treat}_{k,c}$, and thus 
$\inf_{\bs{\delta}\in \mathcal{H}^{n, \treat}_{k,c}} 
\tilde t(\bs{Z}, \bs{Y} - \bs{Z} \circ \bs{\delta}) \leq \tilde t(\bs{Z}, \bs{Y} - \bs{Z} \circ \bs{\tau})=\tilde t(\bs{Z}, \bs{Y}(0))$. Because $\tilde{G}_{\Gamma}(\cdot)$ is a monotone decreasing function, this further implies that  $\tilde{p}_{k,c}^{n, \treat} (\Gamma)\geq \tilde{G}_{\Gamma}(\tilde t(\bs{Z}, \bs{Y}(0)))\geq 
G_{\bs{u}_0, \Gamma}(\tilde t(\bs{Z}, \bs{Y}(0))).
$
Here $G_{\bs{u}_0, \Gamma}(\cdot)$ is the tail probability of $\tilde{t}(\bs{Z}, \bs{Y}(0))$ under under the actual treatment assignment mechanism, which is the sensitivity model with bias at most $\Gamma$ and unmeasured confounding $\bs{u}_0$.
Thus, 
\begin{align*}
    \Pr\{
        \tilde{p}_{k,c}^{n, \treat} (\Gamma) \leq \alpha \text{ and } H_{k,c}^{n, \treat} \text{ holds}
    \}
    & \le 
    \Pr\{ 
        G_{\bs{u}_0, \Gamma}(\tilde t(\bs{Z}, \bs{Y}(0))) \leq \alpha
    \} \le \alpha,
\end{align*}
where the last inequality follows by the validity of standard Fisher randomization test. 
Consequently, $\tilde{p}_{k,c}^{n, \treat} (\Gamma)$ 
is a valid $p$-value for testing $H_{k,c}^{n, \treat}$ under the sensitivity model with bias at most $\Gamma$.

The proof for the remaining parts of Theorem \ref{supp thm:treated_sen} follows by the same logic as the proof of 
Theorem \ref{thm:treated}. 
For conciseness, we omit the detailed proof.
\end{proof}




\subsection{Proof of Theorem \ref{thm:samp_fp}}
\label{Proof: fp}
\begin{proof}[\bf Proof of Theorem \ref{thm:samp_fp}]
We first consider bounding $\Pr(\cup_{j=1}^{J}\{F_\fp^{-1}(\beta_j) < \tau_{(k_j')}\})$ by the same logic as the proof of Theorem \ref{thm:BB_simulCI}.
Let $\{s1, \dots, sn\}$ be the set of indices of sampled units such that $\tau_{si}^{\fp}=\tau_i$ for $i = 1,\dots, n$. 
Let $\{\phi(i)\}_{i=1}^n$ be a subset of $\{1, \dots, N\}$ such that $(\rank_{s1}(\bs{\tau}^{\fp}), \rank_{s2}(\bs{\tau}^{\fp}),\ldots,\rank_{sn}(\bs{\tau}^{\fp}))=(\phi(1),\phi(2),...,\phi(n))$, where $\rank_{i}(\bs{\tau}^{\fp})$ denotes the rank of the
$i$th coordinate of $\bs{\tau}^{\fp}$.  
It is not different to see that 
$\{\phi(i)\}_{i=1}^n$ is a simple random sample of size $n$ from the set $\{1, \ldots, N\}$. 
By definition of the quantile function, 
we can verify that $F_N^{-1}(\beta_j) = \tau^{\fp}_{(k_j)}$, where $k_j = \lceil N \beta_j \rceil$ for all $1\le j\le J$. 
Define further $\tau^{\fp}_{(0)} = -\infty$.
Let 
$v_{j} = \sum_{i=1}^{n}  \mathbbm{1}\left\{ k_j < \phi(i) \leq k_{j+1}\right\}$ for $0 \leq j \leq J$, where $k_0 = 0$ and $k_{J+1}=N$.
By definition, $n(F_\fp^{-1}(\beta_j)) = n(\tau^{\fp}_{(k_j)}) = \sum_{i=1}^n \I(\tau_i > \tau^{\fp}_{(k_j)}) =  \sum_{i=1}^n  \I(\tau^{\fp}_{\left(\phi(i)\right)} > \tau^{\fp}_{(k_j)}) $.
Because we break ties using index ordering, $\I\{\tau^{\fp}_{(k_j)}<\tau^{\fp}_{\left(\phi(i)\right)}\} \leq \I\{k_j<\phi(i)\}$, where the equality holds when there are no ties among individual effects in the whole population, 
i.e., $\tau^{\fp}_{(1)}< \tau^{\fp}_{(2)}<\dots<\tau^{\fp}_{(N)}$. 
Thus, for each $j$, we have $n(F_\fp^{-1}(\beta_j))  \leq \sum_{i=1}^n\I\{\phi(i)>k_j\}=\sum_{i=j}^{J}v_{i}$.
Note that, by definition, 
$F_\fp^{-1}(\beta_j) < \tau_{(k_j')}$ implies that 
$n(F_\fp^{-1}(\beta_j))> n-k_j'$. 
These then imply that 
\begin{align}
\Pr\left(\bigcup\limits_{j=1}^{J}\left\{F_\fp^{-1}(\beta_j) < \tau_{(k_j')}\right\}\right)
&\leq 
\Pr\left(\bigcup\limits_{j=1}^{J}\left\{
    n(F_\fp^{-1}(\beta_j))> n-k_j'\right\}\right)\nonumber
\\
& \le 
\Pr\left(
\bigcup\limits_{j=1}^{J}\left\{\sum_{i=j}^{J}v_{i} > n-k_j' \right\}
\right). \label{Ineq: FP_MHG}
\end{align}
Note that the set of $n$ experimental units is a
simple random sample from the finite population of size $N$. 
Moreover, by definition, $v_j$s count the numbers of sampled units with indices in $J+1$ mutually disjoint sets, $\{i: k_j < \phi(i) \le k_{j+1}\}$, for $0\leq j\leq J$, each of which has a size of $k_{j+1} - k_j$. 
It is then not difficult to see that  $(n-\sum_{j=1}^J v_j, v_1, \ldots, v_{J-1}, v_J) \sim \textup{MHG}(
    k_1, 
    k_2-k_1, 
    \ldots, 
    k_J - k_{J-1}, 
    N-k_J; n
    )$
follows a multivariate Hypergeometric distribution. 
Consequently, the upper bound in \eqref{Ineq: FP_MHG} is equivalent to 
$\Delta_{\textup{H}}(k_{1:J}'; N, n, k_{1:J})$.


Second, we prove the validity of the simultaneous confidence intervals.  
Recall that $\mathcal{I}^{\alpha'}_{(k_j')}$s are simultaneous $1-\alpha'$ one-sided prediction intervals for sample effect quantiles $\tau_{(k_j')}$s. Thus, 
$\Pr\left( \bigcap_{j=1}^{J}\left\{\tau_{(k_j')} \in  \mathcal{I}^{\alpha'}_{(k_j')} \right\}\right) \geq 1-\alpha'$.
Moreover, because $\mathcal{I}^{\alpha'}_{(k_j')}$s are one-sided intervals of form $(c, \infty)$ or $[c, \infty)$, for every $j$, $F_\fp^{-1}(\beta_j) \geq \tau_{(k_j')}$ and $\tau_{(k_j')} \in  \mathcal{I}^{\alpha'}_{(k_j')}$ imply that $F_\fp^{-1}(\beta_j) \in  \mathcal{I}^{\alpha'}_{(k_j')}$. 
From the above, we have
\begin{align}\label{proof:fp}
& \quad \ \Pr\left(\bigcap_{j=1}^{J} \left\{F_\fp^{-1}(\beta_j) \in  \mathcal{I}^{\alpha'}_{(k_j')} \right\}\right)
\nonumber
\\
&\geq 
  \Pr\left(\bigcap_{j=1}^{J} \left\{F_\fp^{-1}(\beta_j) \geq \tau_{(k_j')}, \tau_{(k_j')} \in  \mathcal{I}^{\alpha'}_{(k_j')} 
\right\}\right)\nonumber\\
&=1-\Pr\left(\bigcup\limits_{j=1}^{J} \left\{F_\fp^{-1}(\beta_j) < \tau_{(k_j')} \right\}
\cup
\left(\bigcap_{j=1}^{J}\left\{\tau_{(k_j')} \in  \mathcal{I}^{\alpha'}_{(k_j')} 
\right\}\right)^{\complement}\right)\nonumber\\
&\geq 1-\Pr\left(\bigcup\limits_{j=1}^{J} \left\{F_\fp^{-1}(\beta_j) < \tau_{(k_j')} \right\}
\right)-
\Pr\left(\left(\bigcap_{j=1}^{J}\left\{\tau_{(k_j')} \in  \mathcal{I}^{\alpha'}_{(k_j')} 
\right\}\right)^{\complement}\right)\nonumber\\
&\geq 1-\Delta_{\textup{H}}(k'_{1:J}; N, n, k_{1:J})-\alpha'=1-\alpha.
\end{align}   
Therefore, Theorem \ref{thm:samp_fp} holds.
\end{proof}

\subsection{Proofs of Theorem \ref{thm:sp} and Corollary \ref{cor:super_match}}
\begin{proof}[\bf Proof of Theorem \ref{thm:sp}]
We first consider bounding 
$\Pr\left(\bigcup_{j=1}^{J} \left\{F_{\sp}^{-1}(\beta_j)< \tau_{(k_j')} \right\}
\right)$. 
By definition, $\bigcup_{j=1}^{J}\left\{F_{\sp}^{-1}(\beta_j) < \tau_{(k_j')}\right\}$ implies that there exists some $j$ such that $\sum_{i=1}^n \I(\tau_i > F_{\sp}^{-1}(\beta_j))> n-k_j'$.
By the inverse probability integral transform, 
$(\tau_{1}, \ldots, \tau_{n})
$
follows the same distribution as
$(F_{\sp}^{-1}(U_1), \ldots, F_{\sp}^{-1}(U_{n}))$, 
where $U_i$s are i.i.d.\ \Unif(0,1) random variables. 
These then imply that
\begin{align}
    \Pr\left(\bigcup\limits_{j=1}^{J} \left\{F_{\sp}^{-1}(\beta_j)< \tau_{(k_j')} \right\}
\right)&\leq \Pr\left(\bigcup\limits_{j=1}^{J} \left\{\sum_{i=1}^n \I(\tau_i > F_{\sp}^{-1}(\beta_j))> n-k_j'\right\}
\right)\nonumber\\
&=\Pr\left(\bigcup\limits_{j=1}^{J} \left\{\sum_{i=1}^n  \I(F_{\sp}^{-1}(U_{i}) > F_{\sp}^{-1}(\beta_j))> n-k_j'\right\}
\right)\nonumber\\
& \leq \Pr\left(\bigcup\limits_{j=1}^{J} \left\{\sum_{i=1}^n  \I(\beta_j < U_{i})> n-k_j'\right\}
\right) \label{proof: thm9_bound},
\end{align}
where the last equality holds because 
$F_{\sp}^{-1}(\beta_j)< F_{\sp}^{-1}(U_{i}) $ 
implies $\beta_j < U_{i}$.  

Define $m_{j} = \sum_{i=1}^{n}  \mathbbm{1}\left\{ \beta_j < U_i \leq \beta_{j+1}\right\}$ for $0 \leq j \leq J$, where $\beta_0 = 0$, $\beta_{J+1}=1$.  For each $j$, $\sum_{i=j}^{J} m_{i}=\sum_{i=1}^n  \I(U_i > \beta_j)$. 
Then the upper bound in \eqref{proof: thm9_bound} is equivalent to
$$\Pr\left(\bigcup\limits_{j=1}^{J} \left\{\sum_{i=j}^{J} m_{i}> n-k_j'\right\}
\right) \equiv \Delta_{\textup{M}}(k_{1:J}'; n, \beta_{1:J}),$$
where $(n-\sum_{j=1}^J  m_{j},  m_{1}, \ldots,  m_{J-1},  m_{J}) \sim \textup{MN}(
\beta_1, 
\beta_2-\beta_1, 
\ldots, 
\beta_J - \beta_{J-1}, 
1-\beta_J; n
)$.
This is because the $U_i$s are i.i.d.\ standard uniform random variables,
each of which 
falls in 
one of the $J+1$ mutually exclusive sets, $(\beta_j, \beta_{j+1}]$ for $0\le j \le J$, 
with probabilities $\beta_{j+1}-\beta_{j}$ for $0\le j \le J$, respectively. Therefore, we have 
$\Pr\left(\bigcup\limits_{j=1}^{J} \left\{F_{\sp}^{-1}(\beta_j)< \tau_{(k_j')} \right\}
\right)\leq \Delta_{\textup{M}}(k_{1:J}'; n, \beta_{1:J})$.
 
The proof for the remaining parts follows by the same logic as the proof of 
Theorem \ref{thm:samp_fp}. We omit the detailed proof for conciseness. 
\end{proof}

\begin{proof}[\bf Proof of Corollary \ref{cor:super_match}]
    Corollary \ref{cor:super_match} follows by the same logic as Theorem \ref{thm:sp}, and its proof is omitted for conciseness. 
\end{proof}

\section{Additional details}
\subsection{Additional details for the choice of $k'$ in Theorem \ref{thm:BB_pval}}
In the discussion after Theorem \ref{thm:BB_pval}, we suggest choosing $k'$ such that the correction term  is close to, yet bounded from above by, $\gamma\alpha$, for some $\gamma\in [0,1)$. 
Such a 
choice of $k'$ has the following equivalent forms: 
\begin{align*}
k' = \max\{ k': \Pr( 
    \hyper(k) > n_{\treat} - k'
    )\leq \gamma\alpha\}
    &= n_t - \min\{n_{\treat} - k': \Pr( 
    \hyper(k) > n_{\treat} - k'
    )\leq \gamma\alpha\}\\
    &= n_t -\min\{n_{\treat} - k': \Pr( 
    \hyper(k) \leq n_{\treat} - k'
    )\geq 1-\gamma\alpha\}\\
    &= n_t - q_{\textup{HG}}(1-\gamma\alpha; n, n-k, n_{\treat}),
    \end{align*}
where $\hyper(k) \sim \text{HG}(n, n-k, n_{\treat})$ follows a Hypergeometric distribution and $q_{\textup{HG}}(\cdot; n, n-k, n_{\treat})$ denotes its quantile function. 

\subsection{Additional technical details for Theorem \ref{thm:sp}}\label{sec:tech_super}
\subsubsection{Additional justification by the exchangeability of potential outcomes}

Recall that $(Z_i, Y_i(1), Y_i(0))$ for $1\le i \le n$ denote the treatment assignment and potential outcomes for the $n$ experimental units. 
Under i.i.d.\ sampling from a superpopulation, $(Y_i(1), Y_i(0))$s are i.i.d.\ across all the experimental units. 
Moreover, by our assumption, $(Z_1, \ldots, Z_n)$ is independent of the potential outcomes for the $n$ units. 
Consider inference conditioning on the treatment assignment and the permutations of the potential outcomes. 
Let $(y_{i}(1), y_i(0))$s denote the realized values of the potential outcomes. 
Under this conditioning, $Z_i$s become fixed constant,  
and $(Y_i(1), Y_i(0))$ for $1\le i \le n$ will have equal probabilities to be $(y_{\sigma(i)}(1), y_{\sigma(i)}(0))$ for $1\le i \le n$, for any permutation $\sigma$ of $(1,2,\ldots, n)$. 
In particular, under each permutation $\sigma$, units with potential outcomes in $\{(y_{j}(1), y_{j}(0)): j = \sigma(i), Z_{i} = 1, 1\le i \le n \}$ would be in the treated group, and units with potential outcomes in $\{(y_{j}(1), y_{j}(0)): j = \sigma(i), Z_{i} = 0, 1\le i \le n \}$ would be in the control group. 
Note that $\sigma$ is equally likely to be any permutation of $(1,2,\ldots, n)$. 
Therefore, under this conditioning, we are essentially conducting a CRE with fixed potential outcomes $(y_{i}(1), y_i(0))$s. 
In other words, through conditioning, we can equivalently view the data from i.i.d.\ potential outcomes and independent treatment assignment as if it was from a CRE with a finite population of units with fixed potential outcomes. Theorem \ref{thm:sp} will then hold by this equivalence. 

\subsubsection{Connection between finite population and superpopulation inferences}

Below we discuss the connection between the finite population and superpopulation inferences in Theorems \ref{thm:samp_fp} and \ref{thm:sp}. 
Specifically, we will show that, as the size of finite population $N$ goes to infinity, the confidence intervals for population quantiles of individual effects from Theorem \ref{thm:samp_fp} will converge to that in Theorem \ref{thm:sp}. 
Intuitively, this is because simple random sampling becomes increasingly similar to i.i.d.\ sampling as the size of finite population increases. 
To prove this, it suffices to show that for any given $n\ge 1$, $0\le \beta_1 \le \beta_2 \ldots \le \beta_J \le 1$ and $0 \le k_1' \le \ldots \le k_J' \le n $, 
\begin{align}\label{eq:hyper_converge}
    \Delta_{\textup{H}}(k_{1:J}'; N, n, k_{1:J}) \rightarrow \Delta_{\textup{M}}(k_{1:J}'; n, \beta_{1:J}) 
    \quad \textup{as } N\rightarrow \infty, 
\end{align}
where 
$k_j = \lceil N \beta_j \rceil$ for all $j$. 
To prove \eqref{eq:hyper_converge}, it suffices to prove that the probability mass function of $\textup{MHG}(k_1, k_2-k_1, \ldots, k_J - k_{J-1}, N-k_J; n)$ converges to that of 
$\textup{MN}(\beta_1, \beta_2 - \beta_1, \ldots, \beta_J - \beta_{J-1}, 1 - \beta_J; n)$.

Define $k_0 = 0$, $k_{J+1} = N$, $\beta_0 = 0$ and $\beta_{J+1}=1$ for descriptive convenience. We further assume $\beta_1>0$ and $\beta_J < 1$; this does not lose any generality since the proof when $\beta_1 = 0$ or $\beta_J=1$ is almost the same. 
For any $x_0, x_1, \ldots, x_J\ge 0$ such that $\sum_{j=0}^J x_j = n$, the probability mass function of $\textup{MHG}(k_1, k_2-k_1, \ldots, k_J - k_{J-1}, N-k_J; n)$ evaluated at $(x_0, \ldots, x_J)$ is 
$\binom{N}{n}^{-1} \prod_{j=0}^J \binom{k_{j+1} - k_j}{x_j}$. 
As $N\rightarrow \infty$, we have, for $0\le j \le J$, 
\begin{align*}
    \binom{k_{j+1} - k_j}{x_j}
    = 
    \frac{(k_{j+1} - k_j) (k_{j+1} - k_j - 1) \cdots (k_{j+1} - k_j - x_j+1) }{x_j!}
    = \frac{(k_{j+1} - k_j)^{x_j}}{x_j!} \cdot (1+o(1)), 
\end{align*}
and analogously, 
\begin{align*}
    \binom{N}{n} = \frac{N(N-1)\cdots (N-n+1)}{n!} = \frac{N^n}{n!} \cdot (1+o(1)). 
\end{align*}
These imply that, as $N\rightarrow \infty$, 
\begin{align*}
    \binom{N}{n}^{-1} \prod_{j=0}^J \binom{k_{j+1} - k_j}{x_j} 
    & = \left( \frac{N^n}{n!} \right)^{-1} \frac{\prod_{j=0}^J (k_{j+1} - k_j)^{x_j}}{\prod_{j=0}^J x_j!}\cdot (1+o(1))
    \\
    & =  \frac{n!}{\prod_{j=0}^J x_j!} \prod_{j=0}^J (k_{j+1}/N - k_j/N)^{x_j} \cdot (1+o(1))
    \\
    & = \frac{n!}{\prod_{j=0}^J x_j!} \prod_{j=0}^J (\beta_{j+1} - \beta_j)^{x_j} \cdot (1+o(1)), 
\end{align*}
which is asymptotically the same as the probability mass function of $\textup{MN}(\beta_1, \beta_2 - \beta_1, \ldots, \beta_J - \beta_{J-1}, 1 - \beta_J; n)$ evaluated at $(x_0, \ldots, x_J)$. 

From the above, as the size of finite population $N$ goes to infinity, the inference for population effect quantiles in Theorem \ref{thm:samp_fp} will converge to that in Theorem \ref{thm:sp}.

{\rev 
\subsection{Two-sided intervals for the education experiment}

In Figure \ref{fig:teacher_data_two_sided_CIs}, we show the two-sided confidence intervals for all quantiles of individual treatment effects using the original approach in \citet{CDLM21quantile} and our two improved methods. 
All intervals from \citet{CDLM21quantile} lack either an informative upper or lower bound. 
In contrast, our first improved method provides intervals with both informative upper and lower bounds for the 30\%–70\% quantiles of individual treatment effects, while the second improved method does so for approximately the 40\%–60\% quantiles.
In sum, the improved methods provide more informative two-sided confidence intervals for quantiles of individual treatment effects.

}

\begin{figure}[ht]
    \centering
    \begin{minipage}{.33\textwidth}
        \centering
        \includegraphics[width=0.9\linewidth]{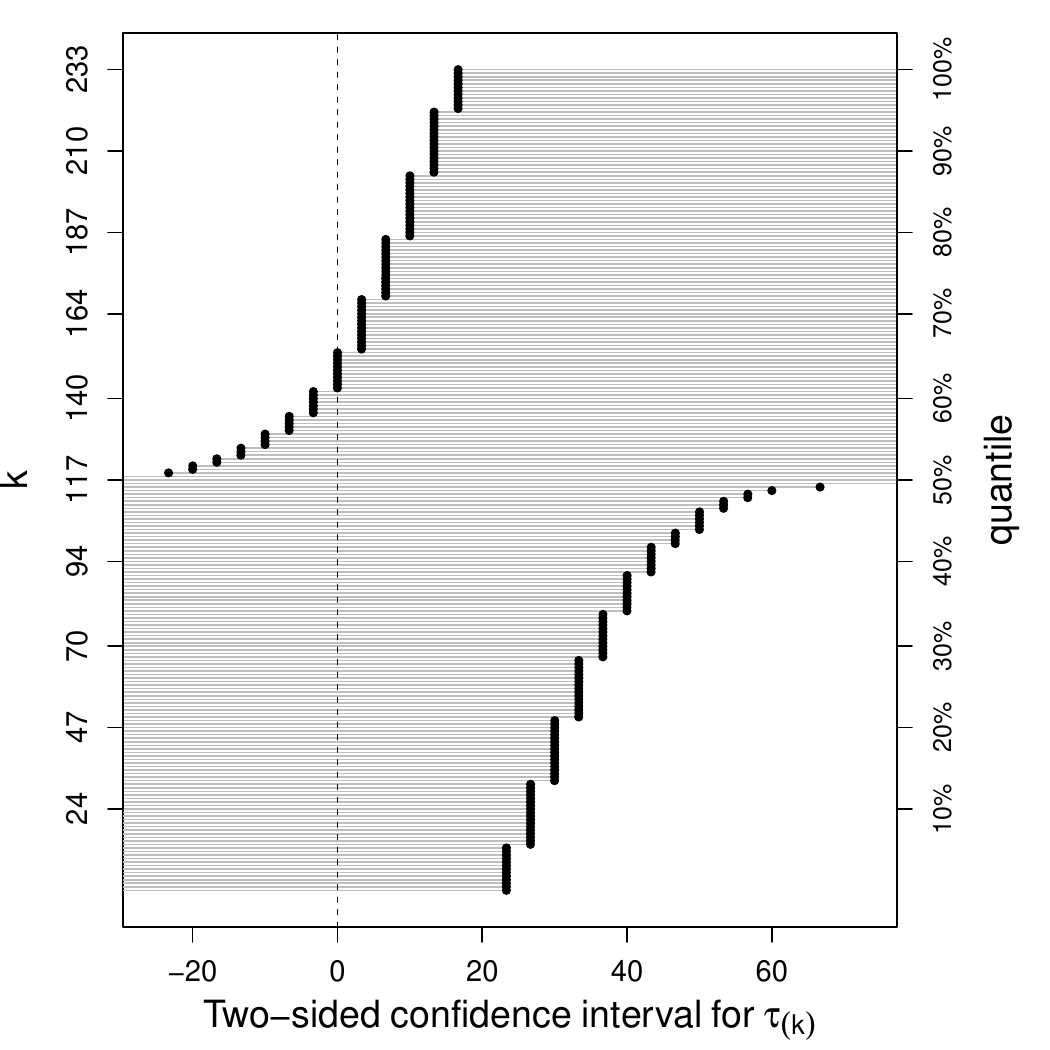}
        \subcaption{\citeauthor{CDLM21quantile}}
        \label{subfig:realdata_M0_two_sided}
    \end{minipage}%
    \begin{minipage}{0.33\textwidth}
        \centering
        \includegraphics[width=0.9\linewidth]{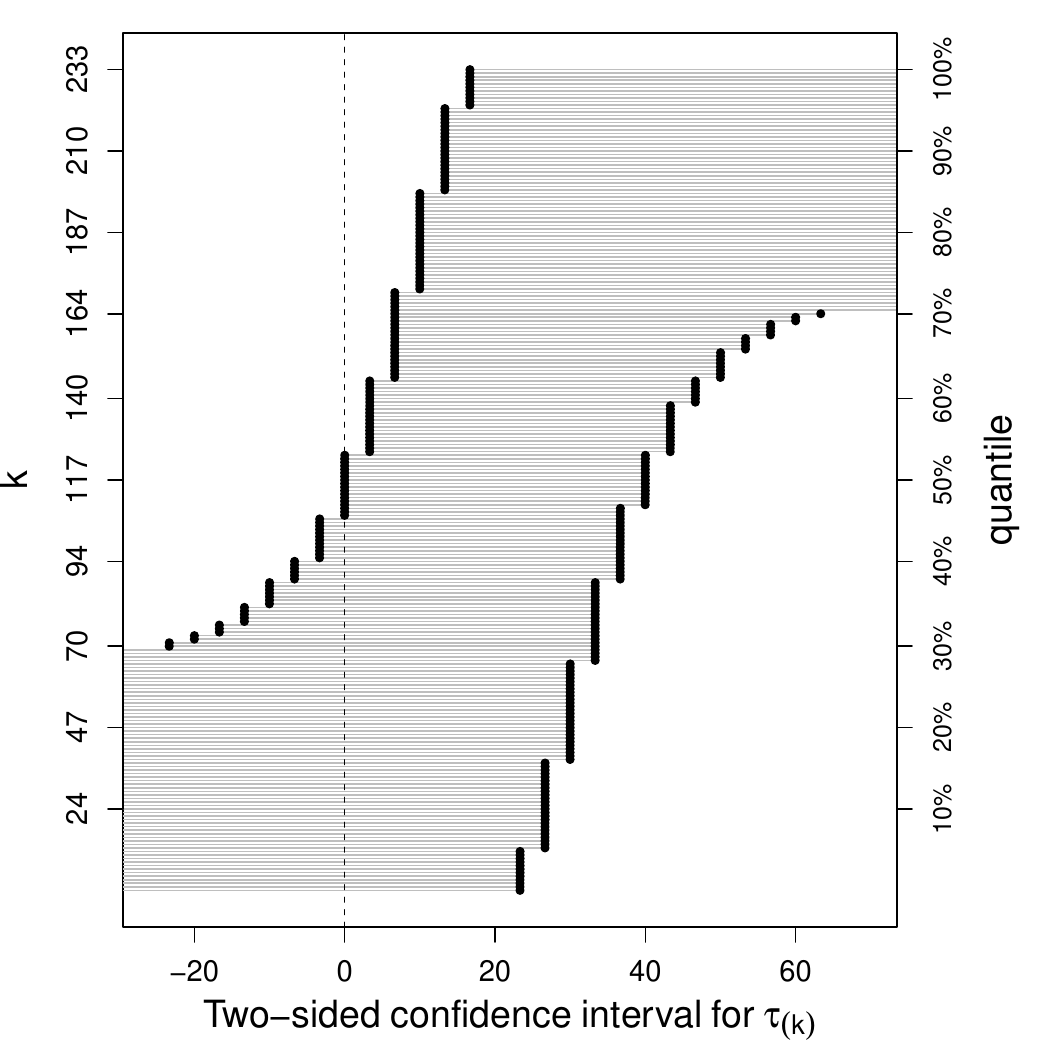}
        \subcaption{Section \ref{sec:effect_treated}}
        \label{subfig:realdata_M1_two_sided}
    \end{minipage}%
    \begin{minipage}{0.33\textwidth}
        \centering
        \includegraphics[width=0.9\linewidth]{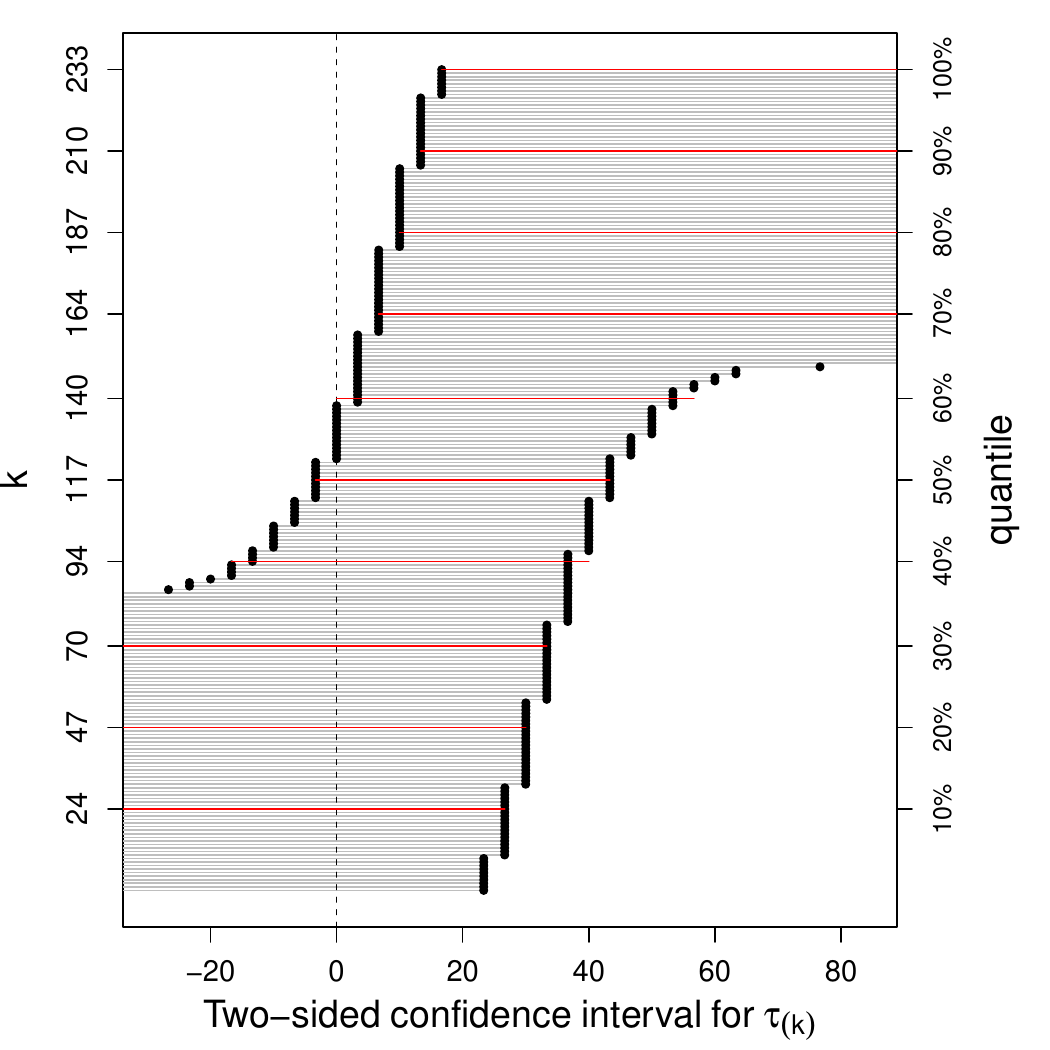}
        \subcaption{Section \ref{sec:berger}}
        \label{subfig:realdata_M2_two_sided}
    \end{minipage}
    \caption{90\% two-sided confidence intervals for 
    treatment effect quantiles among all 233 teachers in the experiment.
    (a) and (b) show the simultaneous confidence intervals using the original \cite{CDLM21quantile}'s method and our improved method from Section \ref{sec:effect_treated}. 
    (c) shows the individual confidence intervals for all quantiles (black lines) and simultaneous confidence intervals for 10 selected quantiles (red lines) using our improved method from Section \ref{sec:berger}.
   }
    \label{fig:teacher_data_two_sided_CIs}
\end{figure}

\end{document}